\documentclass[manuscript,screen]{acmart}

\AtBeginDocument{%
  }

\setcopyright{acmlicensed}
\copyrightyear{2024}
\acmYear{2024}

\usepackage{stmaryrd}
\usepackage{mathbbol}

\usepackage{latexsym,amsmath,epsfig,amsthm,graphics,tabularx}
\usepackage{color,graphicx}
\usepackage{prftree}

\usepackage{mathbbol}
\newcommand\lvar{\mathcal{L}}

\newcommand{\mycomment}[1]{{\color{red}#1} }

\newtheorem{lemma}{Lemma}
\newtheorem{theorem}{Theorem}

\newtheorem{definition}{Definition}

\newcommand{\pskip}{\textmd{skip}}

\newcommand{\evo}[3]{\langle \dot{#1}=#2\& #3\rangle}

\newcommand{\pwait}{\textrm{wait}}
\newcommand{\exempt}[4]{#1 \unrhd \talloblong_{#2} (#3 \rightarrow #4)}

\newcommand{\external}[3]{\talloblong_{#1} (#2 \rightarrow #3)}

\newcommand{\fracN}[2]{\frac{\textstyle #1}{\textstyle #2}}

\newcommand{\chop}{\smallfrown}

\newcommand{\sm}[1]{{$#1$}}

\newcommand{\oomit}[1]{}

\newcommand{\seman}[1]{[\![#1 ]\!]}
\newcommand{\nseman}[3]{\seman{#1}_{#2, #3}^{\lvar}}
\newcommand{\nsemans}[4]{\seman{#1}_{#2, #3}^{#4}}

\newcommand{\rdy}{\mathit{rdy}}

\newcommand{\ommit}[1]{}

\newcommand{\IFE}[3]{\textrm{if}\ #1\ \textrm{then}\ #2\ \textrm{else}\ #3}
\newcommand{\spec}[3]{\{#1\}\ #2\ \{#3\}}

\newcommand{\joinop}{\,\textsf{@}\,}
\renewcommand{\wp}{\mathit{wlp}}
\newcommand{\var}{\mathit{var}}
\newcommand{\wvar}{\mathit{wvar}}
\newcommand{\freev}{\mathit{free}}
\renewcommand{\sp}{\mathit{sp}_{\|}}
\newcommand{\compat}{\operatorname{compat}}
\newcommand{\len}{\operatorname{len}}
\renewcommand{\evo}[3]{\langle \vec{\dot{#1}}=\vec{#2}\& #3\rangle}
\newcommand{\dL}{d$\mathcal{L}$}
\newcommand{\closeb}{\mathit{cl}}
\newcommand{\tracev}{\gamma}
\newcommand{\specs}{\mathit{specs}}
\newcommand{\xx}{\vec{x}}

\newcommand{\true}{\mathsf{true}}

\newcommand{\LieBound}{N_{p\mkern-2mu, \mkern-1mu\ff}}

\renewcommand{\vv}{\vec{v}}
\newcommand{\uu}{\vec{u}}

\newcommand{\sol}{\vec{\zeta}}
\newcommand{\define}{\ \widehat{=}\ }
\newcommand{\invt}{\Psi}

\newcommand{\ff}{\vec{f}}

\newcommand{\RR}{\mathbb{R}}

\newcommand{\NN}{\mathbb{N}}

\newcommand{\FODT}{{$\textrm{FOD}_{\Gamma}$}}
\begin{document}
%
\title{A Generalized Hybrid  Hoare Logic}

\author{Naijun Zhan}
\affiliation{
    \department{School of Computer Science}
    \institution{Peking University\&}
    \institution{Institute of Software, Chinese Academy of Sciences}
    \city{Beijing}
    \country{China}
}
\email{njzhan@pku.edu.cn}

\author{Xiangyu Jin}
\affiliation{
    \department{State Key Lab. of Computer Science}
    \institution{Institute of Software, Chinese Academy of Sciences\&}
    \institution{University of CAS}
    \city{Beijing}
    \country{China}
}
\email{jinxy@ios.ac.cn}

\author{Bohua Zhan}
\authornote{Corresponding authors}
\affiliation{
    \institution{Huawei Technologies Co. Ltd.}
    \city{Beijing}
    \country{China}
}
\email{bzhan@ios.ac.cn}

\author{Shuling Wang}
\authornotemark[1]
\affiliation{
    \department{National Key Laboratory of Space Integrated Information System}
    \institution{Institute of Software, Chinese Academy of Sciences\&}
    \institution{University of CAS}
    \city{Beijing}
    \country{China}
}
\email{wangsl@ios.ac.cn}

\author{Dimitar Guelev}
\affiliation{
    \institution{Institute of Mathematics and Informatics 
Bulgarian Academy of Sciences}
    \city{Sofia}
    \country{Bulgaria}
}
\email{gelevdp@math.bas.bg}

\begin{abstract}
Deductive verification of hybrid systems (HSs) increasingly attracts more attention in recent years because of its power and scalability, where a powerful specification logic for HSs is the cornerstone. 
Often, HSs are naturally modelled by concurrent processes that communicate with each other. However, existing specification logics cannot easily handle such models. In this paper, we present a specification logic and proof system for Hybrid Communicating Sequential Processes (HCSP), that extends CSP with ordinary differential equations (ODE) and interrupts to model interactions between continuous and discrete evolution. Because it includes a rich set of algebraic operators, complicated hybrid systems can be easily modelled in an algebra-like compositional way in HCSP. Our logic can be seen as 
a generalization and simplification of existing hybrid Hoare logics (HHL)  based on duration calculus (DC), as well as a conservative extension of existing Hoare logics for concurrent programs. Its assertion logic is the first-order theory of differential equations (FOD), together with 
assertions about traces recording communications, readiness, and continuous evolution.
We prove continuous relative completeness of the logic w.r.t. FOD, as well as discrete relative completeness in the sense  that continuous behaviour can be arbitrarily approximated by discretization. 
Finally,  we  
implement the above logic  
in Isabelle/HOL, and apply it to verify two case studies to illustrate the power and scalability of our logic.
\end{abstract}

\ccsdesc[500]{Theory of computation~Logic and verification}
\ccsdesc[500]{Theory of computation~Hoare logic}
\ccsdesc[500]{Theory of computation~Timed and hybrid models}

\keywords{Hybrid systems, hybrid Hoare logic, Hybrid CSP, proof system, relative completeness}

\maketitle

\section{Introduction}

{\em Hybrid systems} (HSs) exhibit combinations of discrete jumps and
continuous evolution.  Applications of HSs are everywhere in our daily life, e.g. in industrial automation, transportation, and so on. Many of these applications are \emph{safety-critical}.
How to design correct and reliable complex safety-critical HSs so that people can bet their life on them becomes a grand challenge in computer science and control theory \cite{Wing08}.

There have been a huge bulk of work on formal modeling and verification of HSs, e.g., \cite{Alur:1992,Manna93,Manna93b,Henzinger96,LSVW96,LPY02}, most of which are automata-based. In automata-based approaches,
HSs are modeled as {\em hybrid automata} (HA) ~\cite{Alur:1992,Manna93,Henzinger96},
and verified by computing reachable sets. Unfortunately, as shown in \cite{Henzinger96,HenzingerKPV98}, reachability for most of these systems is undecidable, except for some special 
linear \cite{AD94,LPY02} and non-linear \cite{Gan17} ones. Therefore, in practice, people mainly focus on how to over- and under-approximate reachable sets by using different geometric objects to represent abstract states like Ariadne~\cite{benvenuti2012ariadne} and CORA~\cite{althoff2015introduction}. \oomit{ e.g. polyhedral \cite{Frehse11}, zonotope \cite{Girard05,Girard08}, ellipsoid \cite{Kurzhanski00},  Taylor model \cite{Xin2014} etc., or by applying numeric computation based SMT solvers \cite{Eggers2012,dReach}, or by (sub) level-set together with optimization  \cite{Stipanovic03,Mitchell05,Mitchell07,xue2018,xue2020}, or by discretization together interval arithmetic analysis \cite{Dang10}, or by homemorphism \cite{Xue16}, and so on.} 
The advantages of automata-based approaches are twofold:  a HA describes the whole behavior of the system to be developed, and therefore it is very intuitive; and the verification is fully  automatic.  However, their disadvantages are also twofold: 
HA is analogous to state machines, with little support for structured description, and it is thus difficult to model complex systems; moreover,
 existing techniques for computing reachable sets are not scalable, particularly, most of them can only be used to compute reachable sets in bounded time or in unbounded time with constraints, for example with an invariant region in SpaceEx~\cite{frehse2011spaceex}.

Deductive verification presents an alternative way to ensure correctness of HSs. Several formalisms for reasoning about HSs have been proposed, including those based on differential dynamic logic (\dL)   \cite{Platzer2008,Platzer}, extended duration calculus~\cite{ChaochenRH92}, and hybrid Hoare logic (HHL)  ~\cite{LLQZ10,WZG12,GWZZ13}. For {\dL}, an initial version of the proof system~\cite{Platzer2008} is stated in terms of explicit solutions to ODEs, and is proved to be relatively complete with respect to first-order logic of differential equations (FOD), with the assumption that any valid statements involving ODEs can be proved. The ensuing work~\cite{Platzer12a} gives equivalent discrete versions of ODE rules using Euler approximation. Further work~\cite{Platzer10,Platzer12-cut,Platzer20} discusses additional rules, such as differential invariants, differential cut, and differential ghosts for reasoning about ODEs. {\dL} does not provide explicit operators for concurrency and communication, requiring these characteristics of HSs to be encoded within its sequential hybrid programs, meaning that a complicated HSs with communication and parallel composition cannot be specified and reasoned about in an explicit and  compositional way with {\dL}.

Process algebras such as Communicating Sequential Processes (CSP)~\cite{Hoare85} provide a natural compositional way to model systems with concurrency and communication. Extending classical Hoare logic \cite{Hoare69} to CSP has been studied by Apt \emph{et al.}~\cite{Apt80,Apt83} and by Levin and Gries~\cite{LevinG81}. In both works, Hoare triples for input and output statements are essentially arbitrary when reasoning about sequential processes. Then, for each pair of input and output commands in a parallel process, a cooperation test is introduced to relate the global state before and after communication. As with Owicki-Gries' method for shared-memory concurrency \cite{OG76a,OG76b,Owicki76}, these initial proof systems are not compositional, in the sense that there are proof obligations involving every pair of processes that communicate with each other. Moreover, auxiliary variables are usually needed to keep track of progress within each sequential process. The work by Soundararajan~\cite{Soundararajan84} proposes a compositional proof system for CSP. The main idea is to explicitly introduce a trace recording the history of communications, and allow assertions to also depend on traces. For stating the rule for parallel processes, a compatibility condition is defined, characterizing when the records of communications in different traces are consistent with each other.

For modelling HSs, CSP was extended to Hybrid CSP (HCSP) by introducing ordinary differential equations (ODE) to model continuous evolution and interruptions to model interactions between continuous and discrete evolution ~\cite{Jifeng:1994,Zhou:1996}. Because it has a rich set of algebraic operators, complicated HSs can be easily modelled  in an algebra-like compositional way in HCSP.
Like CSP, it is desired to invent a specification logic for HCSP in order to specify and reason about HSs with concurrency and communication in a compositional way. 
There are several attempts to extend Hoare logic to HCSP based on duration calculus (DC) \cite{ZHR91} in the literature.  \cite{LLQZ10} first extended Hoare logic to HCSP, in which postconditions and history formulas in terms of DC that specify invariant properties  are given separately; however it fails to define compositional rules for communications, parallelism and interruptions, let alone its logical foundations. 
Later, \cite{WZG12} 
proposed an assume-guarantee proof system for HCSP, still based on DC.  \cite{WZG12}'s proof system supports compositional reasoning,  but it 
cannot handle super-dense computation well, nor its logical foundations. Super-dense computation assumes that the
computers are much faster
than other physical devices and computation time of a control program  is therefore
negligible, although the temporal order of computations is still there. Super-dense computation provides a comfortable abstraction of HSs, and 
is thus commonly adopted in most models of HSs. To solve this problem,  \cite{GWZZ13} proposed another DC-based proof system 
for HCSP by introducing the notion of infinitesimal time to model computation cost of control events, that changes the semantics slightly in a counter-intuitive way.
In a word, the disadvantages of existing DC-based HHL include:
\begin{itemize}
    \item it cannot deal with all important features of HSs very well, such as compositionality and super-dense computation;
    \item DC-part complicates the verification very much, as it involves too much details of a system, lacks of abstraction. Our case studies demonstrate this point, particular, related to the implementation of theorem proving;
    \item it lacks of logical foundations, specifically completeness. 
\end{itemize}

\oomit{\mycomment{
Except for the respective limitation of each above work based on DC, all of them lack logical foundations of the proof systems such as relative completeness  proved in this paper; second, although  DC and its extension~\cite{ChaochenRH92} are able to express real-time and continuous  properties, the inference rules of the calculus are incomplete and how to prove the validity of DC formulas is a challenging work, especially in the implementation of the DC-based logic in theorem provers. }}

In this paper,  we re-investigate the proof theory for HCSP by providing a compositional proof system with continuous and discrete relative completeness. 
In order to deal with communication and parallelism in a compositional way, inspired by Soundararajan's work  \cite{Soundararajan84} and Hoare and He's work \cite{UTP1998}, we explicitly introduce the notion of  trace. Different from \cite{Soundararajan84} and \cite{UTP1998}, to deal with continuous evolution, in our setting traces record not only the history of communications and readiness of communication events, but also continuous behavior, which are uniformly called \emph{generalized events}.  So, 
unlike existing proof systems for HCSP based on DC, the assertion logic of  our proof system is first-order logic with assertions on traces. For expressing rules about parallel processes, we define a synchronization operator on traces.
Thus, our proof system can be seen as 
 a weakest liberal precondition calculus for sequential processes and a strongest postcondition calculus for parallel processes, together with rules for reasoning about synchronization on traces.

Clearly, our proof system  can be seen as 
a generalization and simplification of existing DC-based hybrid Hoare logics in the sense: 
\begin{itemize}
  \item first, discarding the DC part in the assertion logic simplifies the proof system in both theory and implementation; 
   \item second, the notions of \emph{generalized event}, \emph{trace} and \emph{trace synchronization} provide  the possibility that parallelism, typically,  
   communication synchronization and time synchronization, can be coped with uniformly in  a compositional way;
  \item finally,  super-dense computation is well naturally accounted by allowing that a trace can 
  contain many discrete events happening at the same instant, ordered by  their causal dependency. 
\end{itemize}
Our proof system is also essentially a conservative extension of Hoare logic for concurrent programs (including CSP) by allowing continuous events (wait events, the definition will be given in  Section~\ref{sec:OperationalSem}) and by introducing traces and trace synchronization so that 
non-interference in Owicki/Gries's logic \cite{OG76a,OG76b,Owicki76} and cooperativeness in Apt \emph{et al.}’s logic \cite{Apt80,Apt83} can be reasoned about explicitly. 

The completeness properties of the proof system are analogous to  continuous and discrete completeness results for differential dynamic logic {\dL} shown in~\cite{Platzer2008,Platzer12a}. However, the situation in HCSP is different in several ways. 
\begin{itemize}
    \item First,  for relative continuous completeness with respect to first-order logic of differential equations (FOD~\cite{Platzer12a}), we need to consider the encoding of the trace assertions involved with communications and continuous evolution.
    \item Second, for discrete completeness, the semantics for continuous evolution is different in HCSP compared to {\dL}, for termination, in HCSP only the state along the continuous evolution  at the boundary is considered, whereas 
all reachable states along it inside the boundary 
are considered in \dL. This gives rise to  extra difficulties for the detection of reaching the boundary, which we have to address in this paper.
\item Finally, for both continuous and discrete completeness, we need to consider the additional constructs in HCSP, e.g., interruptions and parallelism.
\end{itemize}

 
In summary, our main contributions are as follows:
\begin{itemize}
    \item  We present a generalized Hoare logic for HCSP using assertions about traces recording  communications, readiness, and continuous evolution.
    \item We show both continuous and discrete completeness of the proof system. 
    \item We implement the proof system, including the semantics of HCSP and the soundness of all proof rules, in Isabelle/HOL. We have applied the logic to verify two case studies: the simplified lunar lander control system, which involves ODE dynamics, interrupts and parallel composition; and a scheduler controlling  tasks executed in parallel, involving communications, interrupts, and complex control logics. 
\end{itemize}
All the Isabelle code can be found at https://github.com/AgHHL/gHHL.git, including the implementation of our prover and case studies.
 
\subsection{Related Work}
Inspired by the success of Floyd-Hoare logic \cite{Hoare69,Floyd67} in
the verification of sequential programs, several extensions of Floyd-Hoare logic to concurrent programs were proposed in the 1970s- 1980s.\oomit{including Hoare and Zhou's \cite{Hoare72,ZH81,Hoare81}, Owicki and Gries' \cite{OG76a,OG76b,Owicki76}, Apt \emph{et al.}'s \cite{Apt80,Apt83}, Lamport's \cite{Lamport77,Lamport80,Lamport84} and so on.}\oomit{In \cite{Hoare72}, Hoare first investigated how to specify and reason about partial correctness of concurrent programs with shared variables by extending Hoare logic.} Owicki and Gries\oomit{found Hoare's proof system is incomplete as it lacks an inference rule for auxiliary variables, thus they} established a complete proof system for concurrent programs with shared variables\oomit{(called \emph{resource regions})} in \cite{OG76b}, in which  an important notion called \emph{non-interference} was introduced in order to deal with parallelism. Owicki proved in \cite{Owicki76} the completeness of 
the proof system in the sense of Cook \cite{Cook78}. To this end, she introduced two types of auxiliary variables, \emph{traces} and \emph{clocks}, to show an interference-free property between any two component processes. In another direction, Zhou and Hoare \cite{ZH81,Hoare81}, and Apt \emph{et al.} \cite{Apt80,Apt83} studied proof systems for concurrent programs with message passing (i.e., CSP).  Particularly, the notion of \emph{cooperative}, similar to \emph{interference-free} in \cite{OG76a,OG76b}, was introduced in 
\cite{Apt80}. In \cite{Apt83}, using a similar technique to \cite{Owicki76}, Apt proved the completeness of the Hoare logic for CSP. Lamport and Schneider unified Hoare logics for sequential programs and different models of concurrency within a single paradigm, called \emph{generalized Hoare logic}~\cite{Lamport77,Lamport80,Lamport84}. Cousot and Cousot proved the relative completeness of generalized Hoare logic in \cite{CC89}. 

In \cite{Hooman94}, Hooman extended Hoare logic to timed CSP. Extension of Hoare logic to HCSP was first tried by Liu \emph{et al.} \cite{LLQZ10}. They established hybrid Hoare logic (HHL). In HHL, a hybrid Hoare assertion consists of four parts: \emph{pre-} and \emph{post-conditions}, a HCSP program, and a history formula  in terms of DC ~\cite{ZHR91} to specify invariants during continuous evolution. A compositional proof system of HHL using assume/guarantee is presented in \cite{WZG12,GWZZ13}, then \cite{GuelevWZ17} shows the relative completeness of the proof system w.r.t.  DC by exploiting the notion of negligible time to cope with super-dense computation. In \cite{WZZ15}, a theorem prover for HHL was implemented in Isabelle/HOL. However, reliance on 
DC complicates and prevents practical applications of these proof systems, as DC is not able to cope with general continuous behaviours of HSs.

In the literature, there are many other proof-theoretic approaches to verification of HSs, e.g., \dL~\cite{Platzer2008,Platzer18}, hybrid  action systems \cite{BackPP00}, and Hybrid Event-B \cite{Abrial2012}. As  mentioned, {\dL} extends dynamic logic \cite{Pratt76} to HSs by allowing modalities over hybrid programs, that extend classical sequential programs with ODEs to model continuous evolution. To deal with more complex behaviours of HSs, several variants of \dL~ were established, e.g., stochastic differential dynamic logic \cite{Platzer} and differential game logic \cite{Platzer15}. \cite{LP16} proposed differential refinement logic to cope with refinement among different levels of abstraction for a given HS;  \cite{Garlan14} investigated how to apply \dL~to define architecture of  CPSs. 
Recently, component-based verification methodologies developed in \dL~\cite{Muller2016,MullerMRSP18} introduced composition operators to split verification of systems into more manageable pieces.\oomit{\cite{Lunel19} extended them with parallel composition (true concurrency) using design patterns defined using the primitive constructs of \dL.} A temporal logic for \dL~based on trace semantics is proposed in~\cite{Platzer07-temporal}. 
However, as we argued, \dL\ cannot handle 
communication and concurrency in an explicitly compositional way, although there have been some attempts e.g. \cite{Lunel19}.

The rest of this paper is organized as follows: Sect.~\ref{sec:HCSP} recalls HCSP and Sect.~\ref{sec:OperationalSem} defines its operational semantics;  Sect.~\ref{sec:hoare} defines our notion of HHL, its continuous and discrete relative completeness are proved in Sect.~\ref{sec:completeness}\&\ref{sec:discretecomplete}; Sect.~\ref{sec:examples} provides implementation and the case studies, and Sect.~\ref{sec:Conclusion} concludes this paper.  

\section{HCSP}
\label{sec:HCSP}

Hybrid CSP (HCSP)~\cite{Jifeng:1994,Zhou:1996} is a formal language for describing HSs, which is an extension of CSP  by introducing  timing constructs, interrupts, and ODEs for modelling continuous evolution. Exchanging data among processes is described solely by communications, so no shared variable is allowed between different processes in parallel and each program variable is local to the respective sequential component.

The syntax for HCSP is given as follows:
\[  
\begin{array}{lll}  
	c  & ::= & 
	\pskip \mid
	x := e \mid 
	ch?x \mid 
	ch!e \mid 
	c_1 \sqcup c_2 \mid
	c_1; c_2 \mid c^* \mid\\
	& &
	\IFE{B}{c_1}{c_2} \mid 
	\evo{x}{e}{B} \mid  
  \exempt{\evo{x}{e}{B}}{i\in I}{ch_i*}{c_i} \\
	pc &::=& c \mid pc_1\|_{cs} pc_2 
\end{array} 
\]
where \sm{c} and \sm{c_i} are sequential processes, and \sm{pc} and 
 \sm{pc_i} are parallel processes; \sm{x} is a variable over reals in a process, \sm{\dot{x}} stands for its derivative w.r.t. time, $\vec{x}$ (resp. $\vec{e}$) is a vector of variables (expressions) and its $i$-th element is denoted by $x_i$ (resp. $e_i$); $ch$ is a channel name, $ch_i*$ is  either an input event $ch_i?x$ or output  event $ch_i!e$, and $I$ is a non-empty set of indices; $B$ and $e$ are Boolean and arithmetic expressions, respectively; $cs$ is a set of channel names. 

The meaning of $\pskip$, assignment, internal choice, sequential composition, and 
conditional statement are as usual. We explain the intuitive meaning of the additional constructs as follows:
\begin{itemize}
	\item $ch?x$ receives a value along the channel $ch$ and assigns it to variable $x$. It may block waiting for the corresponding output to be ready.
	\item $ch!e$ sends the value of $e$ along $ch$. It may block waiting for the corresponding input to be ready. 
	\item  The repetition \sm{c^*} executes $c$ for a nondeterministic finite number of times.
 
    \item \sm{\evo{x}{e}{B}} is a continuous evolution, which evolves continuously according to the differential equation \sm{\vec{\dot{x}}=\vec{e}} as long
	as the \emph{domain} $B$ holds, and terminates whenever $B$ becomes false. In order to guarantee the existence and uniqueness of the solution of any differential equation, we require as usual that 
the right side $\vec{e}$ satisfies the \emph{local Lipschitz condition} on the interval at
least up to the boundary of the ODE. This is necessary to guarantee that
the ODE has a unique solution before escaping the boundary.  
	\item 
	\sm{\exempt{\evo{x}{e}{B}}{i\in I}{ch_i*}{c_i}} behaves like \sm{\evo{x}{e}{B}}, except it is preempted as soon as one of the communication events \sm{ch_i*} takes place, and then is followed by the corresponding \sm{c_i}. Notice that, if the continuous evolution terminates (reaches the boundary of $B$) before a communication in \sm{\{ch_i*\}_{i\in I}} occurs, the process terminates immediately. 
 
	\item \sm{pc_1\|_{cs} pc_2} behaves as  $pc_1$ and $pc_2$ run independently except that all communications along the set of common channels $cs$ between $pc_1$ and $pc_2$ are synchronized. We assume $pc_1$ and $pc_2$ do not share any variables, nor does the same channel with the same direction (e.g. $ch!$) occur in $pc_1$ and $pc_2$.
\end{itemize}
When there is no confusion in the context, we will use $c$ to represent either sequential or parallel process below. 
The other constructs of HCSP in \cite{Jifeng:1994,Zhou:1996} are definable, e.g., $\pwait\ d$,  external choice, \emph{etc}. 

\section{Operational semantics} \label{sec:OperationalSem}
In this section, we define a big-step semantics for HCSP, and prove that it is equivalent to the existing small-step semantics~\cite{Zhan17}. Both semantics are defined using the new concept of \emph{generalized events}, in order to fit better with the trace-based development later.

We begin by defining some basic notions. A state for a sequential process is a mapping from variable names to real values. A state for a parallel process \sm{pc_1\|_{cs} pc_2} is a pair $(s_1,s_2)$, where $s_1$ is a state for $pc_1$ and $s_2$ is a state for $pc_2$. This enforces the requirement that $pc_1$ and $pc_2$ do not share variables.
A \emph{ready set} is a set of channel directions (of the form \sm{ch*}), indicating that these channel directions are waiting for communication. Two ready sets \sm{\rdy_1} and \sm{\rdy_2} are \emph{compatible}, denoted by \sm{\compat(\rdy_1, \rdy_2)}, if there does not exist a channel $ch$ such that \sm{ch?\in \rdy_1 \wedge ch!\in \rdy_2} or \sm{ch!\in \rdy_1 \wedge ch?\in \rdy_2}. 
 Intuitively, it means input and output along the same channel cannot be both waiting at the same time, that is, as soon as both channel ends are ready, a communication along the channel occurs immediately. This is consistent with the maximal synchronization semantics as in CSP \cite{Hoare85} and Calculus of Communicating Systems (CCS) \cite{Milner80}.

A \emph{generalized event}  defines an one step execution of observable behavior for a (sequential or parallel) HCSP process. There are two types of generalized events:
\begin{itemize}
	\item A \emph{communication event} is of the form $\langle ch\triangleright , v\rangle$, where $\triangleright$ is one of $?$, $!$, or nothing, indicating input, output, and synchronized input/output (IO) event, respectively, and $v$ is a real number indicating the transferred value. 
	\item A \emph{continuous event}, also called 
	a \emph{wait event},  is of the form $\langle d, \vec{p}, \rdy\rangle$ representing an evolution of time length $d>0$. Here $\vec{p}$ is a continuous function from $[0,d]$ to states, that is the unique  trajectory of the considered ODE starting with the given initial state,  and $\rdy$ is the set of channels that are waiting for communication during this period. We allow $d=\infty$ to indicate waiting for an infinite amount of time.
\end{itemize}
We will use \sm{s[\vec{x} \mapsto \vec{e}]} to stand for another state, which is the same as $s$ except for mapping each $x_i$ to the corresponding $e_i$ and \sm{s(e)} (resp. \sm{s(B)} for Boolean expression $B$) for the evaluated value of $e$ (resp. \sm{B}) under \sm{s}. Given a formula $\phi$ (resp. expression $e$), \sm{\phi[\vec{e}/\vec{x}]} (resp. \sm{e[\vec{e}/\vec{x}]}) denotes substituting all  $x_i$ occurring in  $\phi$ (resp. $e$) by  $e_i$ simultaneously. 

\subsection{Trace-based Big-step Semantics of HCSP}
\label{sec:tracesemantics}
A \emph{trace} is an ordered sequence of generalized events  as the result of executing a (sequential or parallel) HCSP process. We denote the empty trace by $\epsilon$, the trace for a deadlocked process by $\delta$, and use the operator $^\chop$ to denote concatenating  two traces.

\paragraph{Trace synchronization:}
Given two traces $tr_1, tr_2$ and a set of shared channels $cs$, we define the relation  to \emph{synchronize} $tr_1$ and $tr_2$  over $cs$ and result in a trace $\textit{tr}$, denoted by  $tr_1 \|_{cs} tr_2 \Downarrow \textit{tr}$. The derivation rules defining synchronization are given in Fig.~\ref{fig:rule-synchronization}. Rule SyncIO defines that when the two parallel traces start with the compatible input and output events along the same channel \sm{ch} (that belongs to the common channel set \sm{cs}) with  same value \sm{v}, then a synchronized event $\langle ch, v\rangle$ is produced,  followed by the synchronization of the remainders of the traces. Rule NoSyncIO defines the case when an external communication event occurs on one side. 
Rules SyncEmpty1-3 deal with the cases where one side terminates earlier than the other side. As in CSP and CCS, 
 a parallel process terminates only if all subprocesses in parallel terminate. 
Rules SyncWait1-2 define the cases when both sides are wait events, i.e. waiting for a communication or evolving w.r.t. an ODE, then the wait events of the same length will synchronize if they have compatible ready sets.  
\begin{figure*} \centering 
{
\begin{eqnarray*} 
 &  \prftree[r]{SyncIO}{ch\in cs}{tr_1\|_{cs} tr_2 \Downarrow \textit{tr}}
{\langle ch!,v\rangle^\chop tr_1 \|_{cs} \langle ch?,v\rangle ^\chop tr_2 \Downarrow \langle ch,v\rangle ^\chop \textit{tr}}
\quad
\prftree[r]{NoSyncIO}{ch\notin cs}{tr_1\|_{cs} tr_2 \Downarrow \textit{tr}}
{\langle ch\triangleright,v\rangle^\chop tr_1 \|_{cs} tr_2 \Downarrow \langle ch\triangleright,v\rangle^\chop \textit{tr}} &\\[1mm]
\oomit{\quad
\prftree[r]{SyncEmpty1}{ch\notin cs}{tr_1\|_{cs} \epsilon  \Downarrow \textit{tr}}
{\langle ch\triangleright,v\rangle^\chop tr_1 \|_{cs} \epsilon  \Downarrow \langle ch\triangleright,v\rangle^\chop \textit{tr}} &\\[1mm]}
 &\prftree[r]{SyncEmpty1}{ch\in cs}{ }
{\langle ch\triangleright,v\rangle^\chop tr_1 \|_{cs} \epsilon  \Downarrow \delta}
\quad
\prftree[r]{SyncEmpty3}{\epsilon \|_{\mathit{cs}} \epsilon \Downarrow \epsilon}  & \\[1mm]
& \prftree[r]{SyncEmpty2}{tr_1\|_{cs} \epsilon  \Downarrow \textit{tr}}
{\langle d,\vec{p}_1,\rdy_1\rangle^\chop tr_1 \|_{cs} \epsilon  \Downarrow \langle d,\vec{p}_1,\rdy_1\rangle^\chop \textit{tr}} & \\[1mm]
& \prftree[r]{SyncWait1 }
{tr_1\|_{cs} tr_2 \Downarrow \textit{tr}}
{\compat(\rdy_1,\rdy_2) \quad d >0 }
{\langle d,\vec{p}_1,\rdy_1\rangle^\chop tr_1 \|_{cs} \langle d,\vec{p}_2,\rdy_2\rangle^\chop tr_2 \Downarrow \langle d, \vec{p}_1\uplus \vec{p}_2,(\rdy_1\cup \rdy_2)-cs\rangle^\chop \textit{tr}} & \\[1mm] 
& \prftree[r]{SyncWait2 }
{\begin{array}{cc}
		d_1>d_2>0 \quad \langle d_1-d_2,\vec{p}_1(\cdot+d_2),\rdy_1 \rangle^\chop tr_1\|_{cs} tr_2\Downarrow \textit{tr} \quad
		\compat(\rdy_1,\rdy_2)
\end{array}}
{\langle d_1,\vec{p}_1,\rdy_1\rangle^\chop tr_1 \|_{cs} \langle d_2,\vec{p}_2,\rdy_2\rangle^\chop tr_2 \Downarrow \langle d_2,\vec{p}_1\uplus \vec{p}_2,(\rdy_1\cup \rdy_2)-cs \rangle^\chop \textit{tr}} & 
\end{eqnarray*} }
\caption{Trace synchronization rules}
\label{fig:rule-synchronization} \end{figure*}

\paragraph{Big-step semantics}
A big-step semantics for HCSP is presented in Fig.~\ref{fig:full-big-semantics}. This leads naturally to the trace-based Hoare logic in Sect.~\ref{sec:hoare}. For a sequential process $c$, its semantics is defined as a mapping, denoted by  \sm{(c,s)\Rightarrow (s',\textit{tr})}, which 
means that $c$ carries initial state $s$ to final state $s'$ with resulting trace $\textit{tr}$ \footnote{It can also be defined as a mapping from an initial state and trace to a final state and trace. The two  definitions are equivalent.}.

\begin{figure*}
{ 
\begin{eqnarray*}
& \prftree[r]{SkipB}{(\pskip,s) \Rightarrow (s,\epsilon)} 
	\quad \prftree[r]{AssignB}{(x:=e,s) \Rightarrow (s[x \mapsto e],\epsilon)} 
& \\[1mm]
 & \prftree[r]{OutB1}{(ch!e, s) \Rightarrow (s, \langle ch!, s(e)\rangle)}
\quad
\prftree[r]{OutB2}{(ch!e, s) \Rightarrow (s, \langle d, I_s, \{ch!\}\rangle ^\chop \langle ch!, s(e)\rangle)} & \\[1mm]
 & \prftree[r]{OutB3}{(ch!e, s) \Rightarrow (s, \langle \infty,I_s,\{ch!\}\rangle )} 
 \quad
\prftree[r]{InB1}{(ch?x, s) \Rightarrow (s[x \mapsto v], \langle ch?, v\rangle) } & \\[1mm]
 	& \prftree[r]{InB2}{(ch?x, s) \Rightarrow (s[x \mapsto v], \langle d, I_s, \{ch?\} \rangle^\chop \langle ch?, v\rangle)} 
	\quad\prftree[r]{InB3}{(ch?x,s)\Rightarrow (s, \langle \infty,I_s,\{ch?\})}
& \\[1mm]
&
\prftree[r]{RepB1}{ }
	{(c^*, s) \Rightarrow (s, \epsilon)}
 \quad
\prftree[r]{RepB2}{ (c, s)  \Rightarrow (s_1, \textit{tr}_1) \quad (c^*, s_1) \Rightarrow (s_2, \textit{tr}_2)  }
	{ (c^*, s) \Rightarrow (s_2, {\textit{tr}_1}^\chop \textit{tr}_2)} 
& \\[1mm]
&\prftree[r]{CondB1}{s_1(B)}{(c_1, s_1) \Rightarrow (s_2, \textit{tr})}
	{(\IFE{B}{c_1}{c_2}, s_1) \Rightarrow (s_2, \textit{tr})} 
 \quad
 \prftree[r]{CondB2}{\neg s_1(B)}{(c_2, s_1) \Rightarrow (s_2, \textit{tr})}
 {(\IFE{B}{c_1}{c_2}, s_1) \Rightarrow (s_2, \textit{tr})} 
& \\[1mm]
 & \prftree[r]{IChoiceB1}{(c_1, s_1) \Rightarrow (s_2, \textit{tr})}
	{(c_1 \sqcup c_2, s_1) \Rightarrow (s_2, \textit{tr})} 
	\quad\prftree[r]{IChoiceB2}{(c_2, s_1) \Rightarrow (s_2, \textit{tr})}
	{(c_1 \sqcup c_2, s_1) \Rightarrow (s_2, \textit{tr})} 
& \\[1mm]
& \prftree[r]{SeqB}{(c, s_1) \Rightarrow (s_2, tr_1)}{(c_2, s_2) \Rightarrow (s_3, tr_2)}
	{(c_1; c_2, s_1) \Rightarrow (s_3, {tr_1}^\chop tr_2)} 
 \quad
\prftree[r]{ContB1}{
		\neg B(s)}
	{(\evo{x}{e}{B}, s) \Rightarrow (s, \epsilon)}
& \\[1mm]
&\prftree[r]{ContB2}{
		\begin{array}{cc}
			\vec{p} \mbox{ is a solution of the ODE $\vec{\dot{x}}=\vec{e}$} \quad d>0\\
			\vec{p}(0)=s(\vec{x})\quad \forall t\in[0,d).\,s[\vec{x}\mapsto \vec{p}(t)](B) \quad \neg s[\vec{x}\mapsto \vec{p}(d)](B)
	\end{array}}
	{(\evo{x}{e}{B}, s) \Rightarrow (s[\vec{x}\mapsto \vec{p}(d)], \langle d, \vec{p}, \emptyset\rangle) }& \\[1mm]
 	& 
\prftree[r]{IntB1}{i\in I}{ch_i*=ch!e}{(c_i,s_1) \Rightarrow (s_2,\textit{tr})}
	{(\exempt{\evo{x}{e}{B}}{i\in I}{ch_i*}{c_i}, s_1) \Rightarrow
		(s_2, \langle ch!, s_1(e)\rangle^\chop \textit{tr})} 
		& \\[1mm]
 	& \prftree[r]{IntB2}{
		\begin{array}{cc}
			\vec{p} \mbox{ is a solution of the ODE $\vec{\dot{x}}=\vec{e}$} \quad \vec{p}(0)=s_1(\vec{x}) \quad d>0\\
			\forall t\in[0,d).\,s_1[\vec{x}\mapsto \vec{p}(t)](B) \quad i\in I\quad ch_i*=ch!e\quad (c_i,s_1[\vec{x}\mapsto \vec{p}(d)])\Rightarrow (s_2,\textit{tr})
	\end{array}}
	{(\exempt{\evo{x}{e}{B}}{i\in I}{ch_i*}{c_i}, s_1) \Rightarrow (s_2,
		\langle d, \vec{p}, \{\cup_{i\in I} ch_i*\}\rangle^\chop \langle ch!,s_1[\vec{x}\mapsto \vec{p}(d)](e)\rangle ^\chop \textit{tr})}
  & \\[1mm]
 	& 
		\prftree[r]{IntB3}{i\in I}{ch_i*=ch?y}{(c_i,s_1[y\mapsto v]) \Rightarrow (s_2,\textit{tr})}
	{(\exempt{\evo{x}{e}{B}}{i\in I}{ch_i*}{c_i}, s_1) \Rightarrow
		(s_2, \langle ch?, v\rangle^\chop \textit{tr})}& \\[1mm]
 	& \prftree[r]{IntB4}{
		\begin{array}{cc}
			\vec{p} \mbox{ is a solution of the ODE $\vec{\dot{x}}=\vec{e}$} \quad \vec{p}(0)=s_1(\vec{x}) \quad d>0\\
			\forall t\in[0,d).\,s_1[\vec{x}\mapsto \vec{p}(t)](B) \quad i\in I\quad ch_i*=ch?y\quad (c_i, s_1[\vec{x}\mapsto \vec{p}(d),y\mapsto v])\Rightarrow (s_2,\textit{tr})
	\end{array}}
	{(\exempt{\evo{x}{e}{B}}{i\in I}{ch_i*}{c_i}, s_1) \Rightarrow (s_2,
		\langle d, \vec{p}, \{\cup_{i\in I} ch_i*\}\rangle^\chop \langle ch?,v\rangle ^\chop \textit{tr})}& \\[1mm]
 	& \prftree[r]{IntB5}{\neg s_1(B)}{(\exempt{\evo{x}{e}{B}}{i\in I}{ch_i*}{c_i}, s_1) \Rightarrow (s_1, \epsilon)}
& \\[1mm]
 & \prftree[r]{IntB6}{
		\begin{array}{cc}
			\vec{p} \mbox{ is a solution of the ODE $\vec{\dot{x}}=\vec{e}$} \quad\vec{p}(0)=s_1(\vec{x})\quad d>0\\
			 \forall t\in[0,d).\,s_1[\vec{x}\mapsto \vec{p}(t)](B) \quad \neg s_1[\vec{x}\mapsto \vec{p}(d)](B) 
	\end{array}}
	{(\exempt{\evo{x}{e}{B}}{i\in I}{ch_i*}{c_i}, s_1) \Rightarrow (s_1[\vec{x}\mapsto \vec{p}(t)],\langle d,\vec{p},\{\cup_{i\in I} ch_i*\}\rangle)}
	& \\[1mm]
 	& 
	\oomit{
	\[ \prftree[r]{RepB1}{(c^*, s) \Rightarrow (s, \epsilon)} 
	\quad\prftree[r]{RepB2}{(c, s_1) \Rightarrow (s_2, tr_1)}{(c^*, s_2) \Rightarrow (s_3, tr_2)}
	{(c^*, s_1) \Rightarrow (s_3, {tr_1}^\chop tr_2)} \] }
	\prftree[r]{ParB}{(c_1, s_1) \Rightarrow (s_1', tr_1)}{(c_2, s_2) \Rightarrow (s_2', tr_2)}{tr_1\|_{cs}tr_2 \Downarrow \textit{tr}}
{(c_1 \|_{cs} c_2, s_1\uplus s_2) \Rightarrow (s_1'\uplus s_2', \textit{tr})} & 
\end{eqnarray*}
}
\caption{Big-step operational semantics}
\label{fig:full-big-semantics}
\end{figure*}

\begin{itemize}
	\item For output \sm{ch!e}, there are three cases depending on whether the communication occurs immediately, waits for some finite time, or waits indefinitely.  Input \sm{ch?x} is defined similarly.  \sm{I_s} represents a constant mapping from \sm{[0, d]} to the initial state \sm{s}.
	\item $c^*$ can be understood in a standard way. 
	\item The execution of \sm{\evo{x}{e}{B}} produces a trajectory of  \sm{\vec{\dot{x}}=\vec{e}}  with the given initial state, represented as a wait event. \sm{B} must become false at the end of the trajectory, while remaining true before that. During the evolution, the ready set is empty. 

	\item For interruption \sm{\exempt{\evo{x}{e}{B}}{i\in I}{ch_i*}{c_i}}, communications have a chance to interrupt up to and including the time at which the ODE reaches the boundary. 
	
\item The semantics of parallel composition is defined by the semantics of its components. Given \sm{s_1} for \sm{c_1} and \sm{s_2} for \sm{c_2},  \sm{s_1\uplus s_2} denotes the pair of states \sm{(s_1,s_2)} as a state for \sm{c_1\|_{cs}c_2}.   
\end{itemize}

\paragraph{Example}
A possible trace for \sm{(\pwait\,1; ch!3)} is \sm{tr_1 = \langle 1, I_{s_1}, \emptyset\rangle ^\chop
\langle ch!, 3\rangle} (here we use \sm{I_s} to denote a constant trajectory that maps each time point in the interval to state $s$). A possible trace for \sm{ch!3} is \sm{tr_2 = \langle 1, I_{s_2}, \{ch!\} \rangle ^\chop \langle ch!, 3\rangle}. A possible trace for \sm{ch?x} is \sm{tr_3 = \langle 1, I_{s_3}, \{ch?\} \rangle ^\chop
\langle ch?, 3\rangle}.
Traces \sm{tr_1} and \sm{tr_3} can synchronize with each other, and form the trace \sm{\langle 1, I_{s_1}\uplus I_{s_3}, \{ch?\} \rangle ^\chop \langle ch, 3\rangle}. \oomit{This corresponds to running \sm{(\pwait\,1; ch!3)} and \sm{ch?x} in parallel, resulting in waiting for one time unit and then communicating 3 along channel \sm{ch}.} However, \sm{tr_2} and \sm{tr_3}  \emph{cannot} synchronize with each other, as the ready sets \sm{\{ch!\}} and \sm{\{ch?\}} are not compatible. \oomit{This corresponds to the fact that running \sm{ch!3} and \sm{ch?x} in parallel cannot result in waiting for one time unit and then communicating.}

\subsection{Equivalence with Small-step Semantics}
We now rephrase the existing small-step semantics in \cite{Zhan17} using generalized events. Each transition in the small-step semantics is of the form $(c,s)\stackrel{e}{\to}(c',s')$, meaning that starting from process $c$
and state $s$, executing one step yields an event $e$ 
(either $\tau$ or a communication event or a wait event) and ends with process $c'$ and state $s'$. $\tau$ represents an internal discrete event.
The full semantics is in Fig.~\ref{fig:full-small-step}.

\begin{figure*}
{
\begin{eqnarray*} &\prftree[r]{AssignS}{(x:=e,s) \xrightarrow{\tau} (\pskip,s[x \mapsto e])}
\quad
\prftree[r]{SeqS1}{(c_1,s) \xrightarrow{e} (c_1',s')}{(c_1;c_2, s) \xrightarrow{e} (c_1';c_2,s')}
\quad
\prftree[r]{SeqS2}{(\pskip;c, s) \xrightarrow{\tau} (c,s)} 
& \\[0.3mm]
&
\prftree[r]{CondS1}{s(b)}{(\IFE{b}{c_1}{c_2}, s) \xrightarrow{\tau} (c,s)}
\quad
\prftree[r]{CondS2}{\neg s(b)}{(\IFE{b}{c_1}{c_2}, s) \xrightarrow{\tau} (c_2,s)}
& \\[0.3mm]
&
\prftree[r]{OutS1}{(ch!e, s) \xrightarrow{\langle ch!,s(e)\rangle} (\pskip, s)}
\quad
\prftree[r]{OutS2}{(ch!e, s) \xrightarrow{\langle d,I_s,\{ch!\}\rangle}  (ch!e, s)}
& \\[0.3mm]
& 
\prftree[r]{OutS3}{(ch!e, s) \xrightarrow{\langle \infty,I_s,\{ch!\}\rangle} (\pskip, s)} \quad
\prftree[r]{InS1}{(ch?x, s) \xrightarrow{\langle ch?, v\rangle } (\pskip, s[x\mapsto v])}
& \\[0.3mm]
&
\prftree[r]{InS2}{(ch?x, s) \xrightarrow{\langle d,I_s,\{ch?\}\rangle}(ch?x, s)}
\quad
\prftree[r]{InS3}{(ch?x, s) \xrightarrow{\langle \infty,I_s,\{ch?\}\rangle}  (\pskip, s)}
& \\[0.3mm]
& 
\prftree[r]{RepS1}
{ }
{(c^*, s) \xrightarrow{\tau} (\pskip, s)} 
\quad 
\prftree[r]{RepS2}
{ (c,s)\xrightarrow{e} (c',s')  }
{(c^*, s) \xrightarrow{e} (c';c^*, s')} 
 & \\[0.3mm]
&
\prftree[r]{IChoiceS1}{(c_1 \sqcup c_2, s) \xrightarrow{\tau} (c_1, s)}
\quad
\prftree[r]{IChoiceS2}{(c_1 \sqcup c_2, s) \xrightarrow{\tau} (c_2, s)}
\quad
\prftree[r]{ContS1}
{\neg s(B)}
{(\evo{x}{e}{B}, s) \xrightarrow{\tau} (\pskip, s)}
& \\[0.3mm]
&
\prftree[r]{ContS2}{\begin{array}{cc}
\vec{p} \mbox{ is a solution of the ODE $\vec{\dot{x}}=\vec{e}$} \quad
\vec{p}(0) = s(\vec{x}) \quad \forall t\in[0,d).\,s[\vec{x}\mapsto \vec{p}(t)](B)
\end{array}}
{(\evo{x}{e}{B}, s) \xrightarrow{\langle d,\vec{p},\emptyset\rangle}  (\evo{x}{e}{B}, s[\vec{x} \mapsto \vec{p}(d)])}
& \\[0.3mm]
&  
\prftree[r]{IntS1}
{\begin{array}{cc}
\vec{p} \mbox{ is a solution of the ODE $\vec{\dot{x}}=\vec{e}$} \quad
\vec{p}(0) = s(\vec{x}) \quad \forall t\in[0,d).\,s[\vec{x}\mapsto \vec{p}(t)](B)
\end{array}}
{(\exempt{\evo{x}{e}{B}}{i\in I}{ch_i*}{c_i}, s)\xrightarrow{\langle d,\vec{p},\{\cup_{i\in I} ch_i*\} \rangle} 
(\exempt{\evo{x}{e}{B}}{i\in I}{ch_i*}{c_i}, s[\vec{x} \mapsto \vec{p}(d)])}
& \\[0.3mm]
&  
\prftree[r]{IntS2}
{\neg s(B)}
{(\exempt{\evo{x}{e}{B}}{i\in I}{ch_i*}{c_i}, s) \xrightarrow{\tau}(\pskip, s)}
\ 
\prftree[r]{IntS3}
{i\in I}{ch_i* = ch!e}
{(\exempt{\evo{x}{e}{B}}{i\in I}{ch_i*}{c_i}, s)
\xrightarrow{\langle ch!,s(e)\rangle}  {(c_i, s)}}
& \\[0.3mm]
& 
\prftree[r]{IntS4}
{i\in I}{ch_i* = ch?x}
{(\exempt{\evo{x}{e}{B}}{i\in I}{ch_i*}{c_i}, s) \xrightarrow{\langle ch?,v\rangle} {(c_i, s[x\mapsto v])}}
\ 
\prftree[r]{ParTauS}
{(c_1,s_1) \xrightarrow{\tau} (c_1',s_1')}
{(c_1\|_{cs} c_2,s_1\uplus s_2) \xrightarrow{\tau} (c_1'\|_{cs} c_2, s_1'\uplus s_2)}
& \\[0.3mm]
& 
\prftree[r]{ParDelayS}
{\begin{array}{cc}
		\compat(\rdy_1, \rdy_2) 
        \quad
		(c_1,s_1) \xrightarrow{\langle d,\vec{p}_1,\rdy_1\rangle}  (c_1',s_1')
		\quad 
		(c_2,s_2) \xrightarrow{\langle d,\vec{p}_2,\rdy_2\rangle}  (c_2',s_2')
\end{array}}
{(c_1\|_{cs} c_2, s_1\uplus s_2)
\xrightarrow{\langle d,\vec{p}_1\uplus \vec{p}_2,\rdy_1\cup \rdy_2\rangle}  (c_1'\|_{cs} c_2', s_1'\uplus s_2')}
& \\[0.3mm]
	&
\prftree[r]{ParPairS1}
{ch \in cs} {(c_1,s_1) 
\xrightarrow{\langle ch!, v\rangle} (c_1',s_1')}
{(c_2,s_2) \xrightarrow{\langle ch?, v\rangle}  (c_2',s_2')}
{(c_1\|_{cs} c_2, s_1\uplus s_2) 
\xrightarrow{\langle ch,v\rangle} 
	(c_1'\|_{cs} c_2', s_1'\uplus s_2')}
& \\[0.3mm]
 &  
\prftree[r]{ParPairS2}
{ch \in cs} {(c_1,s_1) \xrightarrow{\langle ch?, v\rangle}  (c_1',s_1')}
{(c_2,s_2) \xrightarrow{\langle ch!, v\rangle}  (c_2',s_2')}
{(c_1\|_{cs} c_2, s_1\uplus s_2) \xrightarrow{\langle ch,v\rangle} 
	(c_1'\|_{cs} c_2', s_1'\uplus s_2')}
& \\[0.3mm]
	&
	\prftree[r]{ParUnpairS1}
{ch \notin cs} {(c_1,s_1) \xrightarrow{\langle ch\triangleright,v\rangle}  (c_1',s_1')}
{(c_1\|_{cs} c_2, s_1\uplus s_2) \xrightarrow{\langle ch\triangleright,v\rangle} 
	(c_1'\|_{cs} c_2, s_1'\uplus s_2)}
\ \prftree[r]{ParUnpairS2}
{ch \notin cs} {(c_2,s_2) \xrightarrow{\langle ch \triangleright,v\rangle}  (c_2',s_2')}
{(c_1\|_{cs} c_2, s_1\uplus s_2) \xrightarrow{\langle ch\triangleright,v\rangle}
	(c_1\|_{cs} c_2', s_1\uplus s_2')}
& 
\end{eqnarray*}
}
\caption{Small-step operational semantics for sequential and parallel processes}
\label{fig:full-small-step}
\end{figure*}

We use \sm{(c,s) \stackrel{\textit{\textit{tr}}}{\rightarrow^*}(c',s')} to indicate that starting from process \sm{c} and state \sm{s}, a sequence of small-step transitions results in process \sm{c'} and state \sm{s'}, and \sm{\textit{\textit{tr}}} collects the events that occurred in between, ignoring any \sm{\tau} events. Also, we use \sm{\textit{\textit{tr}} \leadsto_r \textit{\textit{tr}}'} to mean \sm{\textit{\textit{tr}}'} can be obtained from \sm{\textit{\textit{tr}}} by combining some of the neighboring events that can be joined together like connecting two continuous events as one.

The following theorem asserts the equivalence between big-step and small-step semantics. \begin{theorem} \label{thm:small-big-semantics}
    i) \sm{(c,s)\Rightarrow (s',\textit{\textit{tr}})} implies 
	\sm{(c,s) \stackrel{\textit{\textit{tr}}}{\rightarrow^*} (\pskip,s')}.
	ii)  \sm{(c,s)\stackrel{\textit{\textit{tr}}}{\rightarrow^*}(\pskip,s')} implies    \sm{\textit{\textit{tr}}\leadsto_r \textit{\textit{tr}}'$ and $(c,s)\Rightarrow (s',\textit{\textit{tr}}')} for some $\textit{\textit{tr}}'$.
\end{theorem}

Before stating the equivalence between big-step and small-step semantics, we need some preliminary concepts.
The transitive closure of small-step semantics is defined as follows. We use
\sm{(c,s) \stackrel{\textit{tr}}{\rightarrow^*}(c',s')} to indicate that starting from process \sm{c} and state \sm{s}, a sequence of small-step transitions  results in process \sm{c'} and state \sm{s'}, and \sm{\textit{tr}} collects the events that occurred in between, ignoring any \sm{\tau} events. The formal definition is given by the following set of rules. 
\[ \prftree{(c,s)\stackrel{\epsilon}{\rightarrow^*}(c,s)} 
\quad
\prftree{(c,s)\xrightarrow{\tau} (c',s')} {(c',s')\stackrel{\textit{tr}}{\rightarrow^*}(c'',s'')}
{(c,s)\stackrel{\textit{tr}}{\rightarrow^*}(c'',s'')}
\quad
\prftree{(c,s)\xrightarrow{e} (c',s')}
{(c',s')\stackrel{\textit{tr}}{\rightarrow^*}(c'',s'')}
{(c,s)\stackrel{e^\chop \textit{tr}}{\rightarrow^*}(c'',s'')} \]

The following lemma will be useful later.
\begin{lemma}\label{lem:small-step-split}
	If $(c,s) \xrightarrow{\langle d,\vec{p},\rdy\rangle}  (c',s')$, and $0<d'<d$, then there exists $c''$ and $s''$ such that
	$(c,s) \xrightarrow{\langle d',\vec{p},\rdy\rangle}  (c'',s'')$ and
	$(c'',s'') \xrightarrow{\langle d-d',\vec{p}(\cdot+d'),\rdy\rangle}  (c',s')$.
\end{lemma}
\begin{proof}
	By case analysis on the small-step rules used to derive $(c,s)\xrightarrow{\langle d,\vec{p},\rdy\rangle} (c',s')$.
\end{proof}

The theorem going from big-step to small-step semantics is then stated as follows. For the parallel case, we use a single $\pskip$ to stand for the parallel composition of $\pskip$ programs.
\begin{theorem}[Big-step to Small-step] \label{thm:big2small}
	For any big-step relation \sm{(c,s)\Rightarrow (s',\textit{tr})}, we have the small-step relation
	\sm{(c,s) \stackrel{\textit{tr}}{\rightarrow^*} (\pskip,s')}.
\end{theorem}

\begin{proof}
	First prove the result where $c$ is a sequential program, by induction on the derivation of $(c,s)\Rightarrow (s',tr)$ using big-step semantics. We focus on the operations special to HCSP. The cases for skip, assign, sequence, conditional, internal choice, and repetition are standard.
	
	\begin{itemize}
		\item Input: there are three rules InB1, InB2 and InB3 in big-step semantics. InB1 corresponds to applying InS1, InB2 corresponds to applying InS2 followed by InS1, InB3 corresponds to applying InS3.
		
		\item Output: there are three rules OutB1, OutB2 and OutB3 in big-step semantics. OutB1 corresponds to applying OutS1, OutB2 corresponds to applying OutS2 followed by OutS1, OutB3 corresponds to applying OutS3.
		
		\item Repetition: there are two rules RepB1 and RepB2 in big-step semantics, corresponding to the cases that $c$ executes for zero or more than one times respectively.  
		
		\item Continuous: there are two rules ContB1 and ContB2 in big-step semantics. ContB1 corresponds to applying ContS1, ContB2 corresponds to applying ContS2 followed by ContS1.
		
		\item Interrupt: there are six rules in the big-step semantics. IntB1 corresponds to applying IntS3, IntB2 corresponds to applying IntS1 followed by IntS3, IntB3 corresponds to applying IntS4, IntB4 corresponds to applying IntS1 followed by IntS4. IntB5 corresponds to applying IntS2, IntB6 corresponds to applying IntS1 followed by IntS2.
	\end{itemize}
	
	Next, we prove the result when $c$ is a parallel program. Hence, we assume $(c_1,s_1)\Rightarrow (s_1',tr_1)$, $(c_2,s_2)\Rightarrow (s_2',tr_2)$ and $tr_1\|_{cs}tr_2 \Downarrow tr$. By induction, we have $(c_1,s_1) \stackrel{tr_1}{\rightarrow^*} (\pskip,s_1')$ and $(c_2,s_2) \stackrel{tr_2}{\rightarrow^*} (\pskip,s_2')$. We now induct on the derivation of $tr_1\|_{cs}tr_2 \Downarrow tr$. The cases correspond to the rules in Fig.~\ref{fig:rule-synchronization}.
	\begin{itemize}
		\item SyncIO: We have $ch\in cs$, $tr_1=\langle ch!,v\rangle^\chop tr_1'$, $tr_2=\langle ch?,v\rangle^\chop tr_2'$, $tr=\langle ch,v\rangle^\chop tr'$ and $tr_1'\|_{cs}tr_2'\Downarrow tr'$. From $(c_1,s_1)\stackrel{tr_1}{\rightarrow^*} (\pskip,s_1')$, we obtain $c_1', c_1'', s_1'', s_1'''$ such that
		\[ \begin{array}{l}
		(c_1,s_1)\stackrel{\epsilon}{\rightarrow^*}(c_1',s_1''), \quad
		(c_1',s_1'')\xrightarrow{\langle ch!,v\rangle} (c_1'',s_1''') \quad\mbox{ and}\quad
		(c_1'',s_1''')\stackrel{tr_1'}{\rightarrow^*}(\pskip, s_1').
		\end{array} \]
		 Likewise, from $(c_2,s_2)\stackrel{tr_2}{\rightarrow^*}(\pskip,s_2')$, we obtain $c_2', c_2'', s_2'', s_2'''$ such that
		\[ \begin{array}{l}
		(c_2,s_2)\stackrel{\epsilon}{\rightarrow^*}(c_2',s_2''), \quad
		(c_2',s_2'')\xrightarrow{\langle ch?,v\rangle}  (c_2'',s_2''') \quad\mbox{ and} \quad
		(c_2'',s_2''')\stackrel{tr_2'}{\rightarrow^*}(\pskip, s_2').
		\end{array} \]
		Now by applying rule ParTauS repeatedly, rule ParPairS1 and the inductive hypothesis, we obtain
		\[
		\begin{array}{l}
			(c_1\|_{cs}c_2, s_1\uplus s_2) \stackrel{\epsilon}{\rightarrow^*} (c_1'\|_{cs}c_2', s_1''\uplus s_2''), \\
			(c_1'\|_{cs}c_2', s_1''\uplus s_2'') \xrightarrow{\langle ch,v\rangle}  (c_1''\|_{cs}c_2'', s_1'''\uplus s_2'''), \\
			(c_1''\|_{cs}c_2'', s_1'''\uplus s_2''') \stackrel{tr'}{\rightarrow^*} (\pskip\|_{cs}\pskip, s_1'\uplus s_2').
		\end{array}
	 	\]
	 	They combine together to give $(c_1\|_{cs}c_2, s_1\uplus s_2) \stackrel{tr}{\rightarrow^*}(\pskip\|_{cs}\pskip, s_1'\uplus s_2')$, as desired. The other direction SyncIO' is similar, where the small-step rule ParPairS2 is used.
	 	
	 \item NoSyncIO: the proof is similar to the SyncIO case, except we only need to work on the left side. The corresponding small-step rule is ParUnpairS1. 
	 
	 \item SyncWait1: the proof is similar to the SyncIO case. The corresponding small-step rule is ParDelayS.
	 
	 \item SyncWait2: We have $d_1>d_2$, $tr_1=\langle d_1,\vec{p}_1,\rdy_1\rangle^\chop tr_1'$, $tr_2=\langle d_2,\vec{p}_2,\rdy_2\rangle^\chop tr_2'$, $tr=\langle d_2,\vec{p}_1\uplus \vec{p}_2,(\rdy_1\cup\rdy_2)-cs\rangle^\chop tr'$, $\compat(\rdy_1,\rdy_2)$ and
	 \[ \langle d_1-d_2, \vec{p}_1(\cdot+d_2),\rdy_1\rangle^\chop tr_1 \|_{cs} tr_2 \Downarrow tr \]
	 
	 From $(c,s_1) \stackrel{tr_1}{\rightarrow^*} (\pskip,s_1')$, we obtain $c',c'',s_1'',s_1'''$ such that
	 \[
	 \begin{array}{l}
	 (c,s_1) \stackrel{\epsilon}{\rightarrow^*} (c',s_1''), \quad
	 (c',s_1'') \xrightarrow{\langle d_1,\vec{p}_1,\rdy_1\rangle}  (c'',s_1''') \quad \mbox{and} \quad
	 (c'',s_1''') \stackrel{tr_1'}{\rightarrow^*} (\pskip,s_1').
	 \end{array} \]
	 By Lemma~\ref{lem:small-step-split}, we obtain $c''',s''''$ such that
	 \[ \begin{array}{l}
	 (c',s_1'') \xrightarrow{\langle d_2,\vec{p}_1,\rdy_1\rangle}  (c''',s'''') \quad\mbox{and} \quad
	 (c''',s'''') \xrightarrow{\langle d_1-d_2,\vec{p}_1(\cdot+d_2),\rdy_1\rangle}  (c'',s_1''').
	 \end{array} \]
	 The rest follows as before, by applying rule ParTauS repeatedly, rule ParDelayS, and the inductive hypothesis. 
	\end{itemize}
\end{proof}

To state the theorem going from small-step to big-step semantics, we need to define the concept of \emph{reduction} from one trace to another. We use \sm{\textit{tr} \leadsto_r \textit{tr}'} to mean that \sm{\textit{tr}'} can be obtained from \sm{\textit{tr}} by combining some of the neighboring blocks that can be joined with each other. Note reduction is not unique: there is no obligation to perform all possible joins. The formal definition is as follows:

\[ \prftree[r]{ReduceEmpty}{\epsilon \leadsto_r \epsilon}\quad
\prftree[r]{ReduceMerge}{h_1(d_1) = h_2(0)}
{
\begin{array}{ll}
     \langle d_1,\vec{p}_1,\rdy\rangle ^\chop \langle d_2,\vec{p}_2,\rdy\rangle ^\chop \textit{tr}
	\leadsto_r \\
	\qquad\langle d_1+d_2,\vec{p}_1 \cdot \vec{p}_2,\rdy \rangle ^\chop \textit{tr} 
\end{array} }\]
\[ \prftree[r]{ReduceCons}{tr_1 \leadsto_r tr_2}
{e^\chop tr_1 \leadsto_r e^\chop tr_2}
\quad
\prftree[r]{ReduceTrans}{tr_1 \leadsto_r tr_2}{tr_2 \leadsto_r tr_3}
{tr_1 \leadsto_r tr_3} \]

A key lemma states that the synchronization of traces respects the reduction relation. More precisely:

\begin{lemma}\label{lem:reduce_sync}
	Given \sm{tr_1\leadsto_r tr_1'}, \sm{tr_2\leadsto_r tr_2'} and \sm{tr_1\|_{cs} tr_2\Downarrow \textit{tr}}, then there exists \sm{\textit{tr}'} such that \sm{\textit{tr}\leadsto_r \textit{tr}'} and \sm{tr_1'\|_{cs} tr_2'\Downarrow \textit{tr}'}.
\end{lemma}

The theorem going from small-step to big-step semantics is as follows.

\begin{theorem}[Small-step to Big-step] \label{thm:small2big}
	For any small-step relation \sm{(c,s)\stackrel{\textit{tr}}{\rightarrow^*}(\pskip,s')}, there exists \sm{\textit{tr}'} such that \sm{\textit{tr}\leadsto_r \textit{tr}'} and \sm{(c,s)\Rightarrow (s',\textit{tr}')}.
\end{theorem}

\begin{proof}
First, we prove the result when $c$ is a sequential program. Induction on the derivation of $(c,s)\stackrel{tr}{\rightarrow^*}(\pskip,s')$ gives three cases. The first case corresponds to $tr=\epsilon,c=\pskip$ and $s=s'$, and the result follows immediately. In the second case, we have $(c,s)\stackrel{\tau}{\to} (c',s')$ and $(c',s')\stackrel{tr}{\to^*} (\pskip,s'')$. From the inductive hypothesis, there exists $tr'$ such that $tr\leadsto_r tr'$ and $(c',s')\Rightarrow (s'',tr')$. It then suffices to show $(c,s)\Rightarrow (s'',tr')$. The proof is by a further induction on the derivation of $(c,s)\stackrel{\tau}{\to}(c',s')$. We omit the details here.

In the third case, we have $(c,s)\stackrel{e}{\to} (c',s')$ and $(c',s')\stackrel{tr}{\to^*} (\pskip,s'')$ for some event $e\neq\tau$. From the inductive hypothesis, there exists $tr'$ such that $tr\leadsto_r tr'$ and $(c',s')\Rightarrow (s'',tr')$, and we need to show there exists some $tr''$ such that $e^\chop tr\leadsto_r tr''$ and $(c,s)\Rightarrow (s'',tr'')$. As in the second case, the proof is by a further induction on the derivation of $(c,s)\stackrel{e}{\to}(c',s')$. In some of the cases where $e$ is a wait block, it is necessary to apply the ReduceMerge rule to merge $e$ with the initial block of $tr$.

Next, we prove the result when $c$ is a parallel program. Again, induction on the derivation of $(c,s) \stackrel{tr}{\rightarrow^*} (\pskip,s')$ results in three cases. In the third case where the first step generates an event $e\neq \tau$, we need to consider each of the small-step rules ParDelayS, ParPairS1, PairPairS2, ParUnpairS1 and ParUnpairS2, making use of Lemma~\ref{lem:reduce_sync}. The details are omitted.
\end{proof}

\begin{proof}[Proof for Theorem~\ref{thm:small-big-semantics}]
Theorem~\ref{thm:small-big-semantics} can be proved directly from 
Theorem~\ref{thm:big2small} and Theorem~\ref{thm:small2big}.
\end{proof}

\section{Hybrid Hoare logic}\label{sec:hoare}
In this section, we introduce our version of hybrid Hoare logic, still denoted by HHL, including the syntax, semantics and proof system. 

\subsection{Basic Notions}
Let $\mathbb{N}$, $\mathbb{N}^+$, $\mathbb{R}$, $\mathbb{R}^+$ and $\mathbb{R}^+_0$ be respectively the set of natural, positive natural, real, positive real and non-negative real numbers. For a vector $\vec{x} \in \mathbb{R}^n$, $x_i$ refers to its $i$-th component and $|| {\vec{x}} ||$ denotes the $\ell^2$-norm. 
Let $\mathbb{R}[\vec{x}]$ be the polynomial ring in $\xx$ over the field $\mathbb{R}$. A polynomial $h \in \mathbb{R}[\xx]$ is \emph{sum-of-squares} (SOS) iff there exist polynomials $g_1, \ldots, g_k \in \mathbb{R}[\xx]$ such that $h = \sum_{i=1}^{k} g_i^2$. We denote by $\Sigma[\xx] \subset \mathbb{R}[\xx]$ the set of SOS polynomials over $\xx$. 
%

\paragraph*{Differential Dynamical Systems.}
We consider a class of continuous dynamical systems modelled by ordinary differential equations of the autonomous type:
\begin{equation}
    \label{eqn:dynamics}
    \dot{\xx} = \ff(\xx)
\end{equation}%
where $\xx \in \RR^n$ is the \emph{state} vector, $\dot{\xx}$ denotes its temporal derivative ${\rm d}\xx/{\rm d}t$, with $t \in \RR^+_0$ modelling time, and $\ff\colon \RR^n \to \RR^n$ is a polynomial \emph{flow field} (or \emph{vector field}) that governs the evolution of the system, which is \emph{local Lipschitz}. A polynomial vector field is local Lipschitz, and hence for some $T \in \RR^+ \cup \{\infty\}$, there exists a unique \emph{solution} (or \emph{trajectory}) $\sol_{\xx_0}\colon [0,T) \to \RR^n$ originating from any initial state $\xx_0 \in \RR^n$ such that (1) $\sol_{\xx_0}(0) = \xx_0$, and (2) $\forall \tau \in [0,T)\colon \frac{{\rm d}\sol_{\xx_0}}{{\rm d}t}\big\vert_{t=\tau} \!= \ff(\sol_{\xx_0}(\tau))$. We assume in the sequel that $T$ is the maximal instant up to which $\sol_{\xx_0}$ exists for all $\xx_0$.

Differential equations is a very important branch of mathematics, particularly, ordinary differential equations is well-studied in mathematics with well-established 
theories, please refer to ~\cite{gronwall1919note,gracca2008computability,walter2013ordinary} for the details. Following Platzer \cite{Platzer08}, we call   the  first-order theories  of (ordinary) differential equations \emph{FOD} in this paper. 
Clearly, 
In order to axiomatize HCSP, it is inevitable to deal with continuous evolution as well the interaction between continuous evolution and discrete jumps. So, we will use FOD   as part of our assertion logic.

\begin{definition}[Lie derivative~\textnormal{\cite{kolar1993natural}}]
	Given a vector field $\ff\colon \RR^n \to \RR^n$ over $\xx$, the \emph{Lie derivative} of a polynomial function $p(\xx)$ along $\ff$ of order $k \in \NN$, 
 written as $\mathcal{L}_{\ff}^k p\colon \RR^n \to \RR$,  is inductively defined by 
	\begin{equation*}
		\mathcal{L}_{\ff}^k p(\xx) \define
		\left\{
		\begin{array}{ll}
			p(\xx), \quad k=0,\\
			\left\langle\frac{\partial}{\partial \xx} \mathcal{L}_{\ff}^{k-1} p(\xx), \ff(\xx)\right\rangle, \quad k>0
		\end{array}
		\right.
	\end{equation*}%
	where $\langle\cdot,\cdot\rangle$ is the inner product of vectors, i.e., $\langle \uu, \vv \rangle \define \sum_{i=1}^{n}u_i v_i$ for $\uu, \vv \in \RR^n$.
\end{definition}
The Lie derivative $\mathcal{L}_{\ff}^k p(\xx)$ is essentially the $k$-th temporal derivative of the (barrier) function $p(\xx)$, and thus captures the change of $p(\xx)$ over time. 
\oomit{\color{blue} Let $I = \langle \mathcal{L}_{\ff}^0 p(\xx), \mathcal{L}_{\ff}^1 p(\xx), \cdots, \mathcal{L}_{\ff}^{k-1} p(\xx)\rangle$ for $k>0$ denote the polynomial ideal generated by the Lie derivatives of $p(\xx)$ with orders less than $k$, if $\mathcal{L}_{\ff}^k p(\xx) \in I$, then we call $k$ the order of the Lie derivatives of $p(\xx)$ along $\ff$.}

An \emph{inductive invariant} $\invt \subseteq \mathbb{R}^n$ of a dynamical system is a set of states such that all trajectories starting from  $\invt$ never transverse it. Formally, 
\begin{definition}[Inductive invariant~\textnormal{\cite{Platzer08}}]\label{def:inv}
    Given a system \eqref{eqn:dynamics}, a set $\invt \subseteq \mathbb{R}^n$ is an \emph{inductive invariant} of system \eqref{eqn:dynamics} if and only if
    \begin{equation}
        \forall \xx_0 \in \invt \ldotp \forall t \in [0, T)\colon \sol_{\xx_0}(t) \in \invt.
    \end{equation}%
\end{definition}

In the sequel, we refer to inductive invariants simply as invariants. In~\cite{LiuEmsoft}, a sufficient and necessary condition on being a polynomial invariant is proposed:
\begin{theorem}[Invariant condition~\textnormal{\cite{LiuEmsoft}}]
    \label{thm:invariantCondition}
    Given a polynomial $p \in \mathbb{R}[\xx]$, its \emph{zero sub-level set} $\{ \xx \mid p(\xx) \leq 0 \}$ is an invariant 
    of system \eqref{eqn:dynamics} if and only if \footnote{In \eqref{eqn:invariantCondition}, $\bigwedge_{j = 0}^{i-1} \mathcal{L}_{\ff}^j p = 0$ is $\true$ for $i = 0$ by default. This applies in the sequel.}
    \begin{equation}
        \label{eqn:invariantCondition}
        p \leq 0 \implies 
        \bigvee\nolimits_{i = 0}^{\LieBound} \left(
        \left(\bigwedge\nolimits_{j = 0}^{i-1} \mathcal{L}_{\ff}^j p = 0\right) 
        \land \mathcal{L}_{\ff}^i p < 0\right)
        \lor \bigwedge\nolimits_{i = 0}^{\LieBound} \mathcal{L}_{\ff}^i p = 0
    \end{equation}%
    where $\LieBound \in \NN^+$ is a completeness threshold, i.e., a finite positive integer that bounds the order of Lie derivatives, 
    which can be computed using Gr\"{o}bner bases\footnote{$\LieBound$ is the minimal $i$ such that $\mathcal{L}_{\ff}^{i+1} p$ is in the polynomial ideal generated by $\mathcal{L}_{\ff}^0 p, \mathcal{L}_{\ff}^1 p, \ldots, \mathcal{L}_{\ff}^i p$. The ideal membership can be decided via Gr\"{o}bner basis. See \cite{LiuEmsoft} for the details.}.
\end{theorem}

Theorem~\ref{thm:invariantCondition} can be extended to general semi-algebraic sets, please refer to \cite{LiuEmsoft} for the details.

\subsection{Assertion Logic and Hoare Triples}
In this subsection, we introduce our assertion logic, which implicitly contains 
FOD for dealing with ODEs. 
\subsubsection{Syntax}
We first present the syntax for terms. The language consists of terms of several types. 
{ \[ 
\begin{array}{rll}
	\textit{val} &:=& x_i ~|~ c ~|~ v + w ~|~ v \cdot w ~|~ \cdots \\
	\textit{time} &:=& d ~|~ \infty ~|~ d_1 + d_2 ~|~ d_1 - d_2 ~|~ \cdots \\
	\textit{vector} &:=& (x_1,\dots,x_n) ~|~ \vec{x} ~|~ \vec{p}(t) \\[1mm]
	\textit{state\_traj} &:=& I_{\vec{x}_0} ~|~ \vec{p}_{\vec{x}_0,\vec{e}} ~|~ \vec{p}(\cdot + d) ~|~ \vec{p}_1 \uplus \vec{p}_2 \\
	\textit{generalized\_event} &:=& \langle ch\triangleright,\mathit{val}\rangle ~|~
	\langle \mathit{time},\mathit{state\_traj},\mathit{rdy}\rangle \\
	\textit{trace} &:=& \epsilon \, | \, \mathit{generalized\_event} \,|\, \tracev \,|\,  \mathit{trace}_1 {^\chop} \textit{trace}_2 
\end{array} 
\] }
Here \sm{\mathit{val}} are terms evaluating to real numbers, including state variables \sm{x_i}, constants \sm{c}, as well as arithmetic operations. \sm{\mathit{time}} evaluates to time lengths, either a positive real number or \sm{\infty}. \sm{\mathit{vector}} evaluates to vectors. We use the special symbol \sm{\vec{x}} to denote the vector consisting of all variables in the state of a sequential process in a pre-determined order. Note that this is viewed as an abbreviation, so that substitution for a particular variable \sm{x_i} will replace the corresponding component in \sm{\vec{x}}. \sm{\mathit{state\_traj}} evaluates to solutions of ODEs, guaranteed by FOD. Here \sm{I_{\vec{x}_0}} denotes the constant state trajectory with value \sm{\vec{x}_0}, i.e., for any \sm{t}, \sm{I(t) = \vec{x}_0} (by convention, we use \sm{\vec{x}_0} for the initial values of all state variables). \sm{\vec{p}_{\vec{x}_0,\vec{e}}} denotes the trajectory of  \sm{\vec{\dot{x}}=\vec{e}} starting from \sm{\vec{x}_0} according to FOD. The second subscript \sm{\vec{e}} may be omitted if it is clear from context. \sm{\vec{p}(\cdot + d)} denotes a time shift by \sm{d} units, \sm{\vec{p}_1\uplus \vec{p}_2} denotes merging two state trajectories  for two sequential processes with disjoint sets of state variables, and \sm{\vec{p}(t)} denotes extracting 
the state at time \sm{t} from  the state trajectory $\vec{p}$. The syntax for generalized events and traces are as before. As in Hoare and He's Unifying Theories of Programming (UTP) \cite{UTP1998}, we introduce a system observational  variable, denoted by $\tracev$, to stand for the current trace of the considered process, which never occurs in any process syntactically. 

Our assertion logic is a first order logic of differential equations with generalized traces, that is an extension of FOD with predicates over generalized traces given above, denoted by \FODT. {\FODT} formulas are constructed 
from atomic formulas of the form  
$\theta_1 \unrhd \theta_2$ and atomic formulas of FOD with Boolean connectives and quantifications,  where \sm{ \unrhd \in \{=, \neq, >, \geq, <, \leq \}}. All assertions of HHL of interest are of the form $\spec{P}{c}{Q}$, still called \emph{Hoare triple}, where $P$ and $Q$ are {\FODT} formulas, and 
$c$ is a HCSP process.

Additionally, if we only allow expressions in {\FODT} polynomial, {\FODT} is an extension of Tarski algebra \cite{Tarski51} with trace predicates. If we allow 
more expressive expressions with Noetherian functions \cite{Krantz2002}, i.e., so called analytic terms, {\FODT} is called \emph{semianalytic} algebra together with trace predicates.
For the latter, it is unclear whether Theorem~\ref{thm:invariantCondition} holds.  
See \cite{Platzer20} for detailed discussions on  \emph{semianalytic} terms and formulas.

\subsubsection{Semantics} 
The terms and formulas are defined over a triple \sm{\langle s, h, \lvar \rangle}, 
where  \sm{s} is a state,
\sm{h} a trace, and \sm{\lvar} a valuation  assigning values to logical variables.
The evaluation of \sm{\textit{val}}, \sm{\textit{time}} and \textit{vector} is defined with respect to  \sm{s} and \sm{\lvar}, denoted by \sm{\seman{\textit{val}}_s^\lvar},  \sm{\seman{\textit{time}}_s^\lvar} and \sm{\seman{\textit{vector}}_s^\lvar} respectively. Their definitions are routine, so we omit them here. 
The evaluation of \sm{\textit{state\_traj}} with respect to \sm{s} and \sm{\lvar}, which returns a function mapping from time to state, is given as follows:
{\[
\begin{array}{rcl}
	\seman{I_{\vec{x}_0}}_s^\lvar(t)  &=& s[\vec{x} \mapsto \vec{x}_0] \\
	\seman{\vec{p}_{\vec{x}_0,e}}_s^\lvar(t)  &=&
	s[\vec{x} \mapsto \vec{p}(t)]  \\
	\multicolumn{3}{l}{\qquad\mbox{where $\vec{p}$ is the solution of \sm{\vec{\dot{x}}=\vec{e}} with initial state \sm{s(\vec{x}_0)}, i.e., $\sol_{\xx_0}$}}\\
	\seman{\vec{p}(\cdot + d)}_s^\lvar(t) &=& \seman{\vec{p}}_s^\lvar(t+d)\\
	\seman{\vec{p}_1 \uplus \vec{p}_2} _s^\lvar(t) &=&
	\seman{\vec{p}_1}_s^\lvar(t) \uplus \seman{\vec{p}_2} _s^\lvar(t)
\end{array}
\] }
Given a state \sm{s}, a trace \sm{h} and a valuation \sm{\lvar} of logical variables,
the semantics of trace expressions 
\sm{\nseman{e}{s}{h}}
is defined as follows:
{\[ 
\begin{array}{rcl}
	\nseman{\langle ch\triangleright, val\rangle}{s}{h}
	&=&  \langle ch\triangleright, \seman{\textit{val}}_{s}^{\lvar} \rangle  \\
	\nseman{\langle \textit{time}, \textit{state\_traj}, \textit{rdy} \rangle}{s}{h}	&=&
\langle \seman{\textit{time}}_{s}^{\lvar}, \seman{\textit{state\_traj}}_{s}^{\lvar}|_{[0, \seman{\textit{time}}_{s}^{\lvar}]}, \textit{rdy} \rangle  \\
	\nseman{\tracev}{s}{h} &=& h\\
	\nseman{\textit{trace}_1{^\chop} \mathit{trace}_2}{s}{h} &=& 
	\nseman{\textit{trace}_1}{s}{h}{^\chop}\nseman{\textit{trace}_2}{s}{h}
\end{array}
\] }
We can see that each trace expression is interpreted to a trace value defined in Sect.~\ref{sec:tracesemantics}. Especially, the state trajectory $\textit{state\_traj}$ in each continuous event 
is restricted to the time interval  $[0, \seman{\textit{time}}_{s}^{\lvar}]$, i.e. $\seman{\textit{state\_traj}}_{s}^{\lvar}|_{[0, \seman{\textit{time}}_{s}^{\lvar}]}$.  Based on the semantics of terms, the semantics of formulas can be defined as usual.

Given a (sequential or parallel) process \sm{c}, we say a Hoare triple is \emph{valid}, denoted by \sm{\vDash \spec{P}{c}{Q}}, 
if for all \sm{s_1, h_1} such that \sm{\nseman{P}{s_1}{h_1}} holds, and big-step relation \sm{(c,s_1) \Rightarrow (s_2,h_2)}, then \sm{\nseman{Q}{s_2}{{h_1}^\chop h_2}} holds.

\subsection{Proof System}
\oomit{\mycomment{Question by wsl: Should we change all occurrences of FOD in the following paragraph to first-order logic, i.e. FOL? I think there are some confusion between the use of FOL and FOD in this paper.  FOD is an abbreviation for first-order logic of differential equations, proposed by Platzer, to allow assertions in the form $\langle \dot{x} = e \rangle\ p$, where $p$ is a pure formula. It is only used for continuous completeness, so I added the explanation of FOD till that section on continuous completeness (In the reviews, two of them mentioned the lack of explaining FOD).}
\mycomment{ZNJ: It depends how to present inference rules for ODEs related constructs. If we use invariant-based rules, I agree with your point, otherwise, solution-based rules, it needs FOD. } }

A proof system for HHL is intended to derive all valid Hoare triples syntactically.
Our proof system of HHL consists of three parts: the proof system for FOD\footnote{When we consider the discrete relative completeness, {FOD} will be replaced by the first-order theory of real arithmetic.}, axioms and inference rules for timed traces and readiness, and axioms and inference rules for HCSP constructs. The first two parts form a proof system for {\FODT}.

As said above, FOD is a well-studied mathematical branch, we will not discuss the proof theory for FOD in this paper, please refer to ~\cite{gronwall1919note,gracca2008computability,walter2013ordinary} for the details. 

\subsubsection{Axioms and inference rules for traces and readiness} 
 Here we give a set of inference rules shown in Fig.~\ref{fig:elim-synchronization} for concluding properties of \sm{\textit{tr}}  from those of  \sm{tr_1} and 
 \sm{tr_2}, given a synchronization operation \sm{\mathit{tr_1}\|_{cs}\mathit{tr_2}\Downarrow \mathit{tr}}. 
   We omit obvious symmetric versions of rules. 

\begin{figure*}
{ \centering 
 \begin{eqnarray*}
&  \prftree[r]{SyncPairE}{\langle ch_1\triangleright_1,v_1\rangle^\chop tr_1 \|_{cs} \langle ch_2\triangleright_2,v_2\rangle^\chop tr_2 \Downarrow \textit{tr}}{ch_1\in cs}{ch_2\in cs}{\exists \textit{tr}'.\, ch_1=ch_2 \wedge v_1 = v_2 \wedge (\triangleright_1,\triangleright_2)\in \{(!,?),(?,!)\}\wedge \textit{tr} = \langle ch,v\rangle^\chop \textit{tr}' \wedge tr_1\|_{cs}tr_2 \Downarrow \textit{tr}'} & \\
& \prftree[r]{SyncUnpairE1}{\langle ch_1\triangleright_1,v_1\rangle^\chop tr_1 \|_{cs} \langle ch_2\triangleright_2,v_2\rangle^\chop tr_2 \Downarrow \textit{tr}}{ch_1\notin cs}{ch_2\in cs}
{\exists \textit{tr}'.\, \textit{tr}=\langle ch_1\triangleright_1,v_1\rangle^\chop \textit{tr}' \wedge tr_1\|_{cs} \langle ch_2\triangleright_2,v_2\rangle^\chop tr_2 \Downarrow \textit{tr}'} & \\
&  \prftree[r]{SyncUnpairE2}{\langle ch_1\triangleright_1,v_1\rangle^\chop tr_1 \|_{cs} \langle ch_2\triangleright_2,v_2\rangle^\chop tr_2 \Downarrow \textit{tr}}{ch_1\notin cs}{ch_2\notin cs}
{\begin{array}{ll}
		&(\exists \textit{tr}'.\, \textit{tr}=\langle ch_1\triangleright_1,v_1\rangle^\chop \textit{tr}' \wedge tr_1\|_{cs} \langle ch_2\triangleright_2,v_2\rangle^\chop tr_2 \Downarrow \textit{tr}')\,\vee \\
		&(\exists \textit{tr}'.\, \textit{tr}=\langle ch_2\triangleright_2,v_2\rangle^\chop \textit{tr}' \wedge \langle ch_1\triangleright_1,v_1\rangle^\chop tr_1\|_{cs} tr_2 \Downarrow \textit{tr}')
\end{array}}  & \\
& \prftree[r]{SyncUnpairE3}{ch\notin cs}{\langle ch\triangleright,v\rangle^\chop tr_1 \|_{cs} \langle d,\vec{p},\rdy \rangle^\chop tr_2 \Downarrow \textit{tr}}{\exists \textit{tr}'.\, \textit{tr} = \langle ch\triangleright,v\rangle^\chop \textit{tr}' \wedge tr_1\|_{cs} \langle d,\vec{p},\rdy\rangle^\chop tr_2 \Downarrow \textit{tr}'} 
& \\
&\prftree[r]{SyncUnPairE4}{ch\in cs}{\langle ch\triangleright,v\rangle^\chop tr_1 \|_{cs} \langle d,\vec{p},\rdy \rangle^\chop tr_2 \Downarrow \delta} & \\
& \prftree[r]{SyncWaitE1}{\neg \compat(\rdy_1,\rdy_2)}{\langle d_1,\vec{p}_1,\rdy_1\rangle^\chop tr_1 \|_{cs} \langle d_2,\vec{p}_2,\rdy_2\rangle^\chop tr_2 \Downarrow \delta}
& \\
&\prftree[r]{SyncWaitE2}{\langle d,\vec{p}_1,\rdy_1\rangle^\chop tr_1 \|_{cs} \langle d,\vec{p}_2,\rdy_2\rangle^\chop tr_2 \Downarrow \textit{tr}}{\compat(\rdy_1,\rdy_2)}
{\exists \textit{tr}'.\, \textit{tr} = \langle d,\vec{p}_1\uplus \vec{p}_2,(\rdy_1\cup \rdy_2)-cs\rangle ^\chop \textit{tr}' \wedge tr_1\|_{cs} tr_2 \Downarrow \textit{tr}'} & \\
& \prftree[r]{SyncWaitE3}{d_1<d_2}{\langle d_1,\vec{p}_1,\rdy_1\rangle^\chop tr_1 \|_{cs} \langle d_2,\vec{p}_2,\rdy_2\rangle^\chop tr_2 \Downarrow \textit{tr}}{\compat(\rdy_1,\rdy_2)}
{\exists \textit{tr}'.\, \textit{tr} = \langle d_1,\vec{p}_1\uplus \vec{p}_2,(\rdy_1\cup \rdy_2)-cs\rangle ^\chop \textit{tr}' \wedge tr_1\|_{cs} \langle d_2-d_1,\vec{p}_2(\cdot+d_1),\rdy_2\rangle^\chop tr_2 \Downarrow \textit{tr}'} & \\
&  
\prftree[r]{SyncEmptyE1}{\epsilon \|_{cs} \epsilon \Downarrow tr}{tr = \epsilon} \qquad
\prftree[r]{SyncEmptyE2}{\langle d,\vec{p},\rdy\rangle^\chop tr_1 \|_{cs} \epsilon \not\Downarrow tr} 
& \\
&\prftree[r]{SyncEmptyE3}{\langle ch*,v\rangle^\chop tr_1 \|_{cs} \epsilon \not\Downarrow tr}{ch\notin cs\wedge \exists tr'.\, tr=\langle ch*,v\rangle^\chop tr' \wedge tr_1 \|_{cs} \epsilon \Downarrow tr'} 
&
\end{eqnarray*} }
\caption{Inference rules for timed traces and readiness}
\label{fig:elim-synchronization}
\end{figure*}

This set of rules can be categorized by the types of initial events on the two sides. For each combination of types of initial events, there is exactly one rule that is applicable, which either produces a synchronization operation where at least one of \sm{tr_1} and \sm{tr_2} is reduced by one event, or produces a deadlock. The initial event has three cases: communication event  where the channel lies or does not lie in \sm{cs}, and continuous event. We only explain some cases because of space limit. If both sides are communication events, where the channel lies in \sm{cs}, then the two events must synchronize with each other, and they have the same channel and value (rule SyncPairE). If both sides are continuous events, then the two ready sets must be compatible (rule SyncWaitE1). Moreover,  if the two durations are equal, they can be synchronized with each other (rule SyncWaitE2); otherwise, the shorter one synchronizes with the initial part of the longer one first (rule SyncWaitE3).

Note that these rules above are essentially the same as the ones in Fig.~\ref{fig:rule-synchronization}, except that
the rules in Fig.~\ref{fig:rule-synchronization} 
 compose a synchronized trace for a parallel process from the traces of its component processes, while these  rules above decompose the trace of a parallel process  into the traces for its component processes in order to split a complicated proof obligation into several smaller ones.

\subsubsection{Axioms and inference rules for HCSP constructs}

The axioms and rules for the constructs of HCSP are presented in Fig.~\ref{fig:hoarelogic}. We explain them in sequence below. 

\begin{figure*}
{
\begin{eqnarray*}
  &   \prftree[r]{ Skip }{\spec{Q}{\pskip}{Q}}   \qquad
	\prftree[r]{ Assign }{\spec{Q[e/x]}{x:=e}{Q}} 
	& \\
 	 &\prftree[r]{ Output }{\left\{
	 	\begin{array}{ll}
	 		Q[\tracev^\chop \langle ch!, e\rangle/\tracev]\,\wedge \\
	 		\forall d>0.\,Q[\tracev^\chop \langle d, I_{\vec{x}_0}, \{ch!\}\rangle^\chop \langle ch!, e\rangle/\tracev]\,\wedge \\
	 		Q[\tracev^\chop \langle \infty, I_{\vec{x}_0}, \{ch!\}\rangle/\tracev]
	 	\end{array}
	 	\right\}
	 	\ ch!e\ \{Q\}} & \\
 	 & \prftree[r]{ Input }{
	 	\left\{
	 	\begin{array}{ll}
	 		\forall v.\, Q[v/x, \tracev^\chop \langle ch?, v\rangle/\tracev]\,\wedge \\
	 		\forall d>0.\,\forall v.\, Q[v/x, \tracev^\chop \langle d, I_{\vec{x}_0},\{ch?\}\rangle ^\chop \langle ch?,v\rangle /\tracev]\,\wedge \\
	 		Q[\tracev^\chop\langle \infty, I_{\vec{x}_0}, \{ch?\}\rangle/\tracev]
	 	\end{array}
	 	\right\}
	 	\ ch?x\ \{Q\}} & \\
 	 &\prftree[r]{ Cont }{
 		\left\{
 		\begin{array}{ll}
 			(\neg B \rightarrow Q)\,\wedge 	\forall d>0.\\
 		\,(\forall t\in[0,d).\,B[\vec{p}_{\vec{x}_0}(t)/\vec{x}]\wedge \neg B[\vec{p}_{\vec{x}_0}(d)/\vec{x}] \\
 			 \to Q[\vec{p}_{\vec{x}_0}(d)/\vec{x}, \tracev^\chop \langle d, \vec{p}_{\vec{x}_0}, \emptyset\rangle/\tracev])
 		\end{array}
 		\right\}\ \evo{x}{e}{B}\ \{Q\}} & \\[1mm] 
 	& \prftree[r]{ Int }{
 		\begin{array}{l}
 			\forall i\in I.\,\textrm{if } ch_i*=ch!e \textrm{ then} ~ \spec{Q_i}{c_i}{R} \textrm{ and }  
 		    P \rightarrow Q_i[\tracev^\chop \langle ch!,e\rangle/\tracev]\,\wedge \\
 			\qquad\qquad P \rightarrow \forall d>0.\, (\forall t\in[0,d).\, B[\vec{p}_{\vec{x}_0}(t)/\vec{x}]) \to \\
            \qquad\qquad\qquad
            Q_i[\vec{p}_{\vec{x}_0}(d)/\vec{x}, \tracev^\chop \langle d,\vec{p}_{\vec{x}_0},\{\cup_{i\in I} ch_i*\}\rangle^\chop \langle ch!,e[\vec{p}_{\vec{x}_0}(d)/\vec{x}]\rangle/\tracev] \\
 			\qquad\quad \textrm{elif } ch_i*=ch?x \textrm{ then} ~  \spec{Q_i}{c_i}{R} \textrm{ and }   P \rightarrow \forall v.\, Q_i[v/x, \tracev^\chop \langle ch?,v\rangle/\tracev]\,\wedge \\
 			\qquad\qquad P \rightarrow \forall d>0.\, \forall v.\, (\forall t\in[0,d).\, B[\vec{p}_{\vec{x}_0}(t)/\vec{x}]  \to \\
            \qquad\qquad\qquad
 			Q_i[\vec{p}_{\vec{x}_0}(d)/\vec{x},v/x,\tracev^\chop \langle d,\vec{p}_{\vec{x}_0},\{\cup_{i\in I} ch_i*\}\rangle^\chop \langle ch?,v\rangle/\tracev] \\
 			P\rightarrow(\neg B \rightarrow R)\\
 			P\rightarrow (\forall d>0.\, (\forall t\in[0,d).\, B[\vec{p}_{\vec{x}_0}(t)/\vec{x}]  \wedge \neg B[\vec{p}_{\vec{x}_0}(t)/\vec{x}] ) \to   R[\vec{p}_{\vec{x}_0}(d)/\vec{x}, \tracev^\chop \langle d, \vec{p}_{\vec{x}_0}, \{\cup_{i\in I} ch_i*\}\rangle/\tracev]) 
 		\end{array}
 	}
 	{ \spec{P}{\exempt{\evo{x}{e}{B}}{i\in I}{ch_i*}{c_i}}{R} } & \\[1mm]
 	&  \prftree[r]{ Par }
 	{\spec{P_1}{c_1}{Q_1}}{\spec{P_2}{c_2}{Q_2}}
 	{\spec{P_1[\epsilon/\tracev] \wedge P_2[\epsilon/\tracev]}
 		{c_1 \|_{cs} c_2}
 		{\exists tr_1,\, tr_2.\, Q_1[tr_1/\tracev] \wedge Q_2[tr_2/\tracev] \wedge tr_1\|_{cs}tr_2 \Downarrow \tracev}} & \\[1mm]
   &  
	\prftree[r]{Seq}{\spec{P}{c_1}{Q}}{\spec{Q}{c_2}{R}}
 	    {\spec{P}{c_1; c_2}{R}}	\quad 
 	    \prftree[r]{Cond}{\spec{P\wedge b}{c_1}{Q}}{\spec{P \wedge \neg b}{c_2}{Q}}
 	{\spec{P}{\IFE{b}{c_1}{c_2}}{Q}} 
  & \\[1mm]
   &  
 	\prftree[r]{IChoice}{\spec{P_1}{c_1}{Q}}{\spec{P_2}{c_2}{Q}}
 	{\spec{P_1 \wedge P_2}{c_1 \sqcup c_2}{Q}}
  \quad
 \prftree[r]{Rep}{\spec{P}{c}{P}}{\spec{P}{c^*}{P}}
 & \\[1mm]	 
 	& \prftree[r]{Conj}{\spec{P}{ c}{Q} }{\spec{P}{ c}{R} }
 	{\spec{P}{ c}{Q \wedge R}}
 	\quad\prftree[r]{Inv}{\freev(P) \cap (\var(c) \cup \{\tracev\})= \emptyset}{\spec{P}{c}{P}}
 	\quad
 	\prftree[r]{Conseq}{P'\rightarrow P}{Q\rightarrow Q'}{\spec{P}{c}{Q}}{\spec{P'}{c}{Q'}} &   
 	\end{eqnarray*}}  
	\caption{Axioms and inference rules for HCSP constructs}
	\label{fig:hoarelogic}  
\end{figure*}

\begin{itemize}
	 \item
The axioms   for $\pskip$ and assignment, and rules for  sequential composition, conditional  statement and internal choice are as usual. 

\item 
For communication events, we need to consider  when a communication event can happen, as it may need to wait for its dual from the environment for 
synchronization, which could be one of  three 
possibilities,  
see axioms Output and Input. These axioms also provide a way 
to compute the weakest precondition 
w.r.t. a given postcondition. 

\item The axiom for ODE \sm{\evo{x}{e}{B}} considers two cases: when the domain \sm{B} is initially false (and the process terminates immediately), or when the ODE evolves for some positive amount of time \sm{d} (axiom Cont).
In Cont, \sm{\vec{p}_{\vec{x}_0}} is the unique solution of \sm{\vec{\dot{x}}=\vec{e}} starting from \sm{\vec{x}_0}. \oomit{This rule can be directly applied if there is an explicit solution of the ODE. For ODEs without explicit solutions, We will derive more convenient rules based on differential invariants later on.}

\item 
For interrupt (rule Int),  precondition \sm{P} should imply the weakest precondition derived from each of the possibilities. Here \sm{Q_i} is a family of predicates indexed by \sm{i\in I}.

\item  
The rule for repetition is defined with the help of a loop invariant.

\item Moreover, for completeness, several  general rules including invariance, conjunction and consequence are added (rules Inv, Conj and Conseq). 
\end{itemize}

We now turn to the rule for parallel processes (rule Par). Any state \sm{s} of \sm{c_1\|_{cs}c_2} can be written in the form \sm{s_1\uplus s_2}, where \sm{s_1} and \sm{s_2} are states of \sm{c_1} and \sm{c_2}, respectively. Here \sm{P_1} and \sm{Q_1} are predicates on the state for \sm{c_1}, and \sm{P_2} and \sm{Q_2} are predicates on the state for \sm{c_2}. In the postcondition, we require that the trace of the parallel program is a synchronization of the traces of \sm{c_1} and \sm{c_2}.  

For any HCSP process \sm{c}, if \sm{\spec{P}{c}{Q}} is derived by the above inference rules, we write \sm{\vdash \spec{P}{c}{Q}}. The following theorem indicates that the proof systems given in Fig.~\ref{fig:elim-synchronization} and  Fig.~\ref{fig:hoarelogic} are sound.
\begin{theorem}[Soundness] \label{thm:soundness} 
	If \sm{\vdash \spec{P}{c}{Q}}, then \sm{\vDash \spec{P}{c}{Q}}. Furthermore, reasoning about 
	\sm{tr_1\|_{cs}tr_2 \Downarrow \textit{tr}} according to the rules of Fig.~\ref{fig:elim-synchronization} is valid according to the rules in Fig.~\ref{fig:rule-synchronization}. 
\end{theorem}

\begin{proof}[Proof for Theorem~\ref{thm:soundness}]	
	Suppose \sm{(c,s) \Rightarrow (s',h')} and \sm{\nseman{P}{s}{h}}, we need to prove \sm{\nseman{Q}{s'}{h^\chop h'}}. We show this by induction on the structure of program \sm{c}.
	
	\begin{itemize}
		\item Assign: From \sm{\nseman{Q[e/x]}{s}{h}}, it follows \sm{\nseman{Q}{s[x \mapsto e]}{h}}. It is also clear that \sm{Q[e/x]} is the weakest precondition.
		
		\item Output: Assume
		\[\begin{array}{ll}
			Q[\tracev^\chop \langle ch!, e\rangle/\tracev]\,\wedge \\
			\forall d>0.\,Q[\tracev^\chop \langle d, I_{\vec{x}_0}, \{ch!\}\rangle^\chop \langle ch!, e\rangle/\tracev]\,\wedge 
			Q[\tracev^\chop \langle \infty, I_{\vec{x}_0},\{ch!\}\rangle/\tracev]
		\end{array}\]
		holds in \sm{(s, h)}. The three parts of the conjunction correspond to the three big-step rules for output. Since the first part holds, we get \sm{\nseman{Q}{s}{h^\chop \langle ch!, s(e) \rangle}}, showing \sm{Q} holds after following the semantic rule \sm{\textrm{OutB1}}. Since the second part holds, we get \sm{\nsemans{Q}{s}{h^\chop h'}{\lvar}}, where \sm{h'= \langle d, I_s, \{ch!\}\rangle^\chop \langle ch!, s(e)\rangle}  for any \sm{d>0}, showing \sm{Q} holds after following the second semantic rule \sm{\textrm{OutB2}}. Since the third part holds, we get \sm{\nseman{Q}{s}{h^\chop h'}}, where \sm{h' = \langle \infty,I_s,\{ch!\}\rangle}, showing \sm{Q} holds after following the third semantic rule \sm{\textrm{OutB3}}. The above analysis also shows that the precondition is in fact the weakest liberal precondition.
		
		\item Input: Assume
		\[
		\begin{array}{ll}
			\forall v.\, Q[v/x, \tracev^\chop \langle ch?, v\rangle/\tracev]\,\wedge \\
			\forall d>0.\,\forall v.\, Q[v/x, \tracev^\chop \langle d, I_{\vec{x}_0},\{ch?\}\rangle ^\chop \langle ch?,v\rangle/\tracev]\,\wedge 
			Q[\tracev^\chop\langle \infty, I_{\vec{x}_0}, \{ch?\}\rangle/\tracev]
		\end{array}
		\]
		holds in \sm{(s, h)}. The three parts of the conjunction correspond to the three big-step rules for input. Since the first part holds, we get \sm{\nsemans{Q}{s[x\mapsto v]}{h^\chop \langle ch?, v \rangle}{\lvar}} for any \sm{v}, showing \sm{Q} holds after following the semantic rule \sm{\textrm{InB1}}. Since the second part holds, we get \sm{\nsemans{Q}{s[x\mapsto v]}{h^\chop h' }{\lvar}}, where \sm{h' =\langle d, I_s, \{ch?\}\rangle^\chop \langle ch?, v \rangle}  for any \sm{d>0} and \sm{v}, showing \sm{Q} holds after following the semantic rule \sm{\textrm{InB2}}. Since the third part holds, we get \sm{\nseman{Q}{s}{h^\chop h'}}, where \sm{h' = \langle \infty,I_s,\{ch?\}\rangle}, showing \sm{Q} holds after following the semantic rule \sm{\textrm{InB3}}. The above analysis also shows that the precondition is in fact the weakest liberal precondition.
		
		\item Sequence: By induction, we have \sm{\vDash\spec{Q}{c_2}{R}}. According to the semantics rule \sm{\textrm{SeqB}}, there must exist \sm{(s_1, h_1)} such that \sm{(c_1,s) \Rightarrow (s_1,h_1)} and \sm{(c_2,s_1) \Rightarrow (s',h_2)} and \sm{h'=h_1 {^\chop} h_2}. The Hoare triple for \sm{c_1} gives \sm{\nseman{Q}{s_1}{h^\chop h_1}}, then the Hoare triple for \sm{c_2} gives
		\sm{\nseman{R}{s'}{(h^\chop h_1)^\chop h_2}}, which is equal to \sm{\nseman{R}{s'}{h^\chop h'}}.
		
		\item The proofs for conditional rule and internal choice are as usual.
		
		\item Repetition: By induction, we have \sm{\vDash \spec{P}{c}{P}}. According to the operational semantics, there are two cases. The first case is \sm{h'=\epsilon} (rule \sm{\textrm{RepB1}}). Then \sm{\nseman{P}{s}{h}} holds directly. In the second case, there exist \sm{s_1, h_1} and \sm{h_2} such that
		\[(c, s) \Rightarrow (s_1, h_1), (c^*, s_1) \Rightarrow (s', h_2) \] 
		and \sm{h'=h_1 {^\chop} h_2} (rule \sm{\textrm{RepB2}}). From the Hoare triple for \sm{c}, we have \sm{\nseman{P}{s_1}{h^\chop h_1}}; then by induction on the number of iterations, we get the \sm{\nseman{P}{s'}{(h^\chop h_1)^\chop h_2}}, which is 
		equal to \sm{\nseman{P}{s'}{h^\chop h'}}.
	
	\oomit{	\item  Recursion: The proof is standard. Assume $P$ holds in $(s, h)$. For recursion, if it terminates, there must be a finite recursive depth. The fact holds for depth 0 trivially.  Suppose the fact holds for depth n, then $\vDash \spec{P}{\mu X. c}{Q}$ holds for depth n. We need to prove the fact holds for depth $n+1$. As $\vDash \spec{P}{\mu X. c}{Q}$, then
		by assumption,  \sm{\vDash \spec{P}{c[\mu X. c/X]}{Q}}  is true, which corresponds to exactly  the $n+1$-th call. According to the semantic rule RecB,  \sm{(c[\mu X. c/X], s) \Rightarrow (s',h')}, then \sm{\nseman{Q}{s'}{h^\chop h'}} is proved by induction. Similarly, we can prove that the weakest liberal precondition defined by unfolding the body is also sound. }
		
		\item Continuous: Assume
		\[  
		\begin{array}{ll}
			(\neg B \rightarrow Q)\,\wedge \\
			\forall d>0.\,(\forall t\in[0,d).\,B[\vec{p}_{\vec{x}_0}(t)/\vec{x}])\wedge \neg B[\vec{p}_{\vec{x}_0}(d)/\vec{x}] \to 
			Q[\vec{p}_{\vec{x}_0}(d)/\vec{x}, \tracev^\chop \langle d, \vec{p}_{\vec{x}_0}, \emptyset\rangle/\tracev]
		\end{array}
		\]
		holds in \sm{(s, h)}. There are two parts of the conjunction, corresponding to the evaluation of \sm{B} on \sm{s}.
		If \sm{B} is false in \sm{s}, according to the operational semantics (rule \sm{\textrm{ContB1}}), \sm{s' = s} and \sm{h'} is \sm{\varepsilon}. From the first part of the conjunction, we get \sm{\nseman{Q}{s}{h}}, which is equal to \sm{\nseman{Q}{s'}{h^\chop h'}}, as desired.
		
		Now suppose \sm{B} is true in \sm{s}, according to the operational semantics (rule \sm{\textrm{ContB2}}), suppose \sm{d>0} and the solution \sm{\vec{p}} of the ODE starting from \sm{s} satisfies \sm{\forall t\in[0,d).\,s[\vec{x}\mapsto \vec{p}_{\vec{x}_0}(t)](B)} and \sm{s[\vec{x}\mapsto \vec{p}_{\vec{x}_0}(d)](\neg B)}, then the final state and trace are \sm{s'=s[\vec{x} \mapsto \vec{p}_{\vec{x}_0}(d)]}  and \sm{h'= \langle d, \vec{p}, \emptyset\rangle}. From the second part of the conjunction, with \sm{\seman{\vec{p}_{\vec{x}_0}}_{s,h}=\vec{p}}, we get \sm{\nsemans{Q}{s'}{h^\chop \langle d, \vec{p}, \emptyset\rangle}{\lvar}}, as desired. The above analysis also shows that the precondition is in fact the weakest liberal precondition.
		
		\item Interrupt: There are four parts in the precondition of the rule \sm{\textrm{Int}}. The first two parts are for interrupt by output and input communication, respectively. The third part is for the case that violation of \sm{B} holds initially, and the fourth part is for the case that after some time \sm{d>0}, \sm{B} violates. Here we give the proof for the case of interrupt by output communication. The input case is also similar.
		
		Assume the semantic rule \sm{\textrm{IntB2}} is applied, suppose \sm{d>0} and the solution \sm{\vec{p}} of the ODE starting from \sm{s} satisfies \sm{\forall t\in[0,d).\,s[\vec{x}\mapsto \vec{p}(t)](B)}, and there exists \sm{ch!e \in ch_i*} and \sm{(c_i,s[\vec{x} \mapsto \vec{p}(d)])\Rightarrow (s_2,h_2)}. Then the final state and trace of the interrupt command is
		\[ (s_2, \langle d, \vec{p}, \{\cup_{i\in I} ch_i*\}\rangle^\chop \langle ch!,s[\vec{x} \mapsto \vec{p}(d)](e)\rangle ^\chop h_2). \]
		From the assumption 
		{
		\[
		\begin{array}{ll}
			P \rightarrow \forall d>0.\, (\forall t\in[0,d).\, B[\vec{p}_{\vec{x}_0}(t)/\vec{x}]) \to \\
			\quad Q_i[\vec{p}_{\vec{x}_0}(d)/\vec{x}, \tracev^\chop \langle d,\vec{p}_{\vec{x}_0},\{\cup_{i\in I} ch_i*\}\rangle^\chop \langle ch!,e[\vec{p}_{\vec{x}_0}(d)/\vec{x}]\rangle/\tracev] 
		\end{array}
		\]}
		
		\noindent plus that \sm{P} holds for \sm{(s, h)}, we have the right side of the entailment holds for \sm{(s, h)}. Then the assumptions on \sm{d} and \sm{\vec{p}} gives
		{
		\[Q_i[\vec{p}_{\vec{x}_0}(d)/\vec{x}, \tracev^\chop \langle d,\vec{p}_{\vec{x}_0},\{\cup_{i\in I} ch_i*\}\rangle^\chop \langle ch!,e[\vec{p}_{\vec{x}_0}(d)/\vec{x}]\rangle/\tracev] \]
		}
		
		\noindent must hold for \sm{(s, h)}. Thus
		\sm{Q_i} must hold for \sm{s[\vec{x} \mapsto \vec{p}_{\vec{x}_0}(d)], h_1)} with 
		{
		\[ h_1 = h ^\chop \langle d_0,\vec{p},\{\cup_{i\in I} ch_i*\}\rangle^\chop \langle ch!,s[\vec{x} \mapsto \vec{p}(d)](e))\rangle\]
    	}
    	
		\noindent By the inductive assumption on \sm{c_i}, we have \sm{\nseman{R}{s_2}{h_1 {^\chop} h_2}}, which is 
		\[ \nseman{R}{s_2}{h ^\chop \langle d_0,\vec{p},\{\cup_{i\in I} ch_i*\}\rangle^\chop \langle ch!,s[\vec{x} \mapsto \vec{p}(d)](e)\rangle^\chop h_2} \]
		The above proof shows the case where the output interrupt occurs after time \sm{d>0}. There is another simpler case without waiting time, we omit the details here.
		
		
		\item Parallel composition: Assume \sm{P_1[\epsilon/\tracev] \wedge P_2[\epsilon/\tracev]} holds in \sm{(s_1\uplus s_2, h)}, 
		and \sm{(c_1 \|_{cs} c_2, s_1\uplus s_2) \Rightarrow (s_1'\uplus s_2', h')}, so that there exist $h_1$ and $h_2$ such that \sm{(c_1, s_1) \Rightarrow (s_1', h_1)$, $(c_2, s_2) \Rightarrow (s_2', h_2)} and \sm{h_1\|_{cs}h_2 \Downarrow h'} hold, according to the semantic rule \sm{\textrm{ParB}}. We need to prove 
		\[ \exists tr_1\, tr_2.\, Q_1[tr_1/\tracev] \wedge Q_2[tr_2/\tracev] \wedge tr_1\|_{cs}tr_2 \Downarrow \tracev \]
		holds for \sm{(s_1'\uplus s_2', h^\chop h')}. From the assumption that \sm{P_1[\epsilon/\tracev] \wedge P_2[\epsilon/\tracev]} holds in \sm{(s_1\uplus s_2, h)}, we have \sm{h=\varepsilon}, \sm{\nseman{P_1}{s_1}{h}} and \sm{\nseman{P_2}{s_2}{h}}. By induction on \sm{c_1} and \sm{c_2}, we get \sm{\nseman{Q_1}{s'_1}{h_1}} and \sm{\nseman{Q_2}{s'_2}{h_2}}. On the other hand, we also have \sm{ \seman{tr_1}_{s'_1}^\lvar = h_1} and
		\sm{\seman{tr_2}_{s'_2}^\lvar = h_2}.

		Thus the existence condition holds for \sm{(s_1'\uplus s_2', h')} by taking \sm{tr_1} and \sm{tr_2} to be \sm{h_1, h_2}, respectively. This analysis also shows that the postcondition is in fact the strongest postcondition.
		
		\item The proof for the soundness of the rest rules are  as usual.
		
	\end{itemize}
	
We now prove that each rule in Fig.~\ref{fig:elim-synchronization} is valid according to the rules of Fig.~\ref{fig:rule-synchronization}. Selection of the proofs of the rules are given due to similarity. 
\begin{itemize}
\item Rule SyncPairE: Since \sm{ch_1\in cs} and \sm{ch_2\in cs}, the assumption can only be derived from rule SyncIO, and the result follows.

\
item Rule SyncUnpairE1: Since \sm{ch_1\notin cs}, the assumption cannot be derived from rule SyncIO, so only NoSyncIO can be used, and the result follows. 
Derivation for its symmetric counterpart is similar.

\item Rule SyncUnpairE2: the assumption can be derived using  NoSyncIO or its symmetric case. These two cases correspond to the two cases of the disjunction, respectively.

\item Rule SyncUnpairE3: only rule SyncUnpairE1 can be used to derive the assumption, and the result follows. Derivation for its symmetric counterpart   is similar.

\item Rule SyncUnpairE4: the negation of the conclusion cannot be derived using any rule. Note SyncUnpairE1 cannot be used since \sm{ch \in cs}.

\item Rule SyncWaitE1: this rule states that two processes cannot be waiting for two sides of the same communication at the same time. The negation of the conclusion can be derived using only one of SyncWait1, SyncWait2 and its symmetric case. However, all these rules require the condition \sm{\compat(\rdy_1,\rdy_2)}.

\item Rules SyncWaitE2, SyncWaitE3 and its symmetric case: the assumptions of the three rules only be derived from SyncWait1, SyncWait2 and its symmetric case, respectively, so the result follows.

\item Rule SyncEmpE1: the assumption of the rule can be derived only using SyncEmpty3.

\item Rule SyncEmpE2: there is no introduction rule that can derive the negation of the conclusion.

\item Rule SyncEmpE3: the only introduction rule that can derive the assumption is NoSyncIO. This rule requires that \sm{ch \notin cs}, so the result follows. Derivation of its symmetric counterpart   is similar.
\end{itemize}
\end{proof}

\subsection{Incompleteness and Undecidability}
Obviously, \sm{\models \{\top\} ~ {\textbf{while} ~B ~ \textbf{ do } ~S} ~ \{\bot\}} 
if and only if
\sm{\textbf{while} ~B ~ \textbf{do} ~S} does not terminate. As argued in \cite{Cook78}, in order to 
specify termination, an assertion logic should be at least as expressive as Peano arithmetic, which is not 
complete according to G\"{o}del's Incompleteness Theorem \cite{godel1931}. So, the proof system of HHL is 
not complete. Moreover, its validity is not decidable either, even not semi-decidable, as \emph{multiple-path polynomial programs} (MPP), whose termination problem is even not semi-decidable \cite{Manna05},  can be easily modelled with HCSP.

\begin{theorem}[Incompleteness and Undecidability]
The proof system of HHL is incomplete, and the validity of HHL is undecidable, even not semi-decidable. 
\end{theorem}

Additionally, as we pointed out before, in order to deal with communication and concurrency, we introduce \emph{generalized events}, 
\emph{traces} and \emph{trace synchronization}. Thus, the execution of 
a process may start with some history trace, but the following 
theorem indicates that our proof system can guarantee any execution of a process itself is indeed independent of any history trace, which is in accordance with  the healthiness condition given in UTP \cite{UTP1998}.
\begin{theorem} \label{thm:trace-independence}
If $\vdash \spec{P\wedge \tracev=\epsilon}{c}{Q\wedge \tracev=h} $, then for any trace $h'$ with $\freev(h') \cap \wvar(c) = \emptyset$, we have 
 $\vdash \spec{P\wedge \tracev=h'}{c}{Q\wedge \tracev=h'^\chop h} $, where $\wvar(c)$ stands for the variables that are updated  by $c$. 
\end{theorem}
\begin{proof}[Proof for Theorem~\ref{thm:trace-independence}]
We give a proof sketch for this theorem by structural induction on $c$. For all the cases, suppose $\vdash \spec{P\wedge \tracev=\epsilon}{c}{Q\wedge \tracev=h}$ holds, we need to prove $\vdash \spec{P\wedge \tracev=h'}{c}{Q\wedge \tracev=h'^\chop h} $ holds when $\freev(h') \cap \wvar(c) = \emptyset$. Next we use $\equiv$ to represent that two assertions are equivalent. 

\begin{itemize}
    \item Rule Skip: we have $P \equiv Q$ and $h = \epsilon$, the fact holds trivially. 
    
    \item Rule Assign: we have $(P\wedge \tracev=\epsilon) \equiv ((Q\wedge \tracev=h)[e/x])$, then $P\equiv Q[e/x]$ and $h = \epsilon$. By applying rule Assign, $\spec{P\wedge \tracev=h'}{c}{Q\wedge \tracev=h'^\chop h}$  holds when $x$ does not occur in $h'$. This is guaranteed by the restriction $\freev(h') \cap \wvar(c) = \emptyset$. 
    
    \item Rule Output: According to the rule, $P\wedge \tracev=\epsilon$ is equivalent to 
    \[\begin{array}{ll}
			(Q\wedge \tracev=h)[\tracev^\chop \langle ch!, e\rangle/\tracev]\,\wedge \\
			\forall d>0.\,(Q\wedge \tracev=h)[\tracev^\chop \langle d, I_{\vec{x}_0}, \{ch!\}\rangle^\chop \langle ch!, e\rangle/\tracev]\,\wedge 
			(Q\wedge \tracev=h)[\tracev^\chop \langle \infty, I_{\vec{x}_0},\{ch!\}\rangle/\tracev]
		\end{array}\]
		Then by replacing $h$ by $h'^\chop h$ in the above formula, and denoting the resulting formula as $Pre'$, then we need to prove that $Pre'$ is equivalent to $P \wedge \tracev=h'$. The proof is given below. In fact, there are three cases for the conjunction depending on whether and when communication occurs. If the first case occurs, we have $(P\wedge \tracev=\epsilon) \equiv ((Q\wedge \tracev=h)[\tracev^\chop\langle ch!, e\rangle/\tracev])$, then $P \equiv Q$ and $h \equiv \langle ch!, e\rangle$. Then by applying the same rule for postcondition $Q\wedge \tracev=h'^\chop h$, we get the precondition $P \equiv \tracev=h'$, the fact is proved.
    If the second case occurs, we have  $(P\wedge \tracev=\epsilon) \equiv \forall d>0. ((Q\wedge \tracev=h)[\tracev^\chop \langle d, I_{\vec{x}_0}, \{ch!\}\rangle^\chop \langle ch!, e\rangle/\tracev])$, then  $P \equiv Q$ and $h \equiv \langle d, I_{\vec{x}_0}, \{ch!\}\rangle^\chop \langle ch!, e\rangle$ for any $d>0$. Then by applying the same rule for postcondition $Q\wedge \tracev=h'^\chop h$, we get the precondition $P \equiv \tracev=h'$, the fact is proved.
    The third case can be proved similarly and we omit it. 
    
    \item Rule Input: The proof can be given by combining the proofs for output and assignment. We omit the details here.
    
    \item Rule Cont: According to the rule, $P \wedge \tracev = \epsilon$ is equivalent to 
    \[
    \begin{array}{ll}
 			(\neg B \rightarrow (Q\wedge \tracev=h))\,\wedge 
 			\forall d>0.\,(\forall t\in[0,d).\,B[\vec{p}_{\vec{x}_0}(t)/\vec{x}])\wedge \neg B[\vec{p}_{\vec{x}_0}(d)/\vec{x}] \to \\
 			\qquad (Q\wedge \tracev=h)[\vec{p}_{\vec{x}_0}(d)/\vec{x}, \tracev^\chop \langle d, \vec{p}_{\vec{x}_0}, \emptyset\rangle/\tracev]
 		\end{array}
\]
There are two cases depending on whether $B$ holds or not initially. If $B$ does not hold, we have $h=\epsilon$. By replacing $h$ by $h'^\chop h$, it is equivalent to $P \wedge \tracev=h'$, which completes the proof. Otherwise if $B$ holds, then $h=\langle d, \vec{p}_{\vec{x}_0}, \emptyset\rangle$. By replacing 
$h$ by $h'^\chop h$, it is equivalent to $P \wedge \tracev=h'$, as $\vec{x}$ do not occur in $h'$.   

\item Rule Int: There are six cases for the interrupt. Consider the two cases for output interrupt, some $c_i$ executes after the communication occurs. Then by induction, the trace history independence holds for $c_i$: if we have $\spec{P_c \wedge \tracev = h_c}{c_i}{Q\wedge \tracev=h}$  for some $h_c$, then for any $h'$ that the variables of $c_i$ do not occur in, there must be $\spec{P_c \wedge \tracev = h'^\chop h_c}{c_i}{Q\wedge \tracev=h'^\chop h}$. Continuing the proof by considering the communication, we can obtain 
$\spec{P \wedge \tracev = h'}{c}{Q\wedge \tracev=h'^\chop h}$, as $\vec{x}$ does not occur in $h'$.

\item Rule Par: For parallel composition $c_1\|c_2$, the initial traces for both $c_1$ and $c_2$ are always $\epsilon$, so the fact holds trivially. 

\item Rule Seq: By induction, for $c_1$ and $c_2$, we can get the two facts:
$\spec{P\wedge \tracev=h'}{c_1}{Q_m \wedge \tracev=h'^\chop h_m}$ and 
$\spec{Q_m \wedge \tracev=h'^\chop h_m}{c_2}{Q \wedge \tracev=h'^\chop h_m ^\chop h_n}$ such that $h =h_m ^\chop h_n$. The fact holds by applying Rule Seq.

\item  Rules Cond, IChoice, Repetition, Conj, Inv, Conseq can be proved easily by induction.  

\end{itemize}
    
\end{proof}

\subsection{Differential Invariant Rules}
\label{subsec:dinvariant}
 Axioms and rules for ODEs and communication interruptions are based on  explicit solutions of ODEs in the previous subsection. 
 According to FOD, 
 explicit solutions to many ODEs do not exist, even exist, it is not easy to manipulate as they are too complicated. So, in the literature, it is common to use \emph{(differential) invariant} to specify and reason about continuous evolutions \cite{LLQZ10,LiuEmsoft,Platzer20}. 
 So,  we also provide  a set of differential invariant rules 
 as alternatives  below in order to provide a practical way to cope with continuous evolution, which are similar to \cite{LiuEmsoft,Platzer20}. 

To the end,
we introduce the following notation first. 

{ \[ 
\textsf{trInv}(d,\textit{Inv},\rdy) \define \exists\vec{p}.\, \tracev = \langle d,\vec{p},\rdy\rangle \wedge \forall \tau\in[0,d].\, \textit{Inv}[\vec{p}(\tau)/\vec{x}]
\] }

\noindent Here \textit{Inv} is a Boolean formula on states, and the assertion states that \textit{Inv} is satisfied along the entire trajectory, we omit $\rdy$ if it is an empty set.

\cite{Platzer20} proposes a complete version of differential invariant rule in terms of higher-order Lie derivatives, which is quite similar to 
Theorem~\ref{thm:invariantCondition}, adapted to the case of HCSP as follows (below $\closeb(B)$ stands for the closure of $B$ including $B$ and its boundary):
\oomit{The \emph{Lie derivative} of some term $p$ along a given ODE $\vec{\dot{x}} = \vec{e}$, written as $\mathcal{L}^{1}_{\vec{e}}(\textit{p})$,  is defined as $\sum_{i=1}^{i=n} \frac{\partial{p}}{\partial{x_i}} \cdot e_i$ where $n$ is the dimension of $\vec{x}$,  and recursively, the $(i+1)$-th higher Lie derivative of $p$ along the ODE, written as $\mathcal{L}^{i+1}_{\vec{e}}(\textit{p})$,  is defined as $\mathcal{L}^{1}_{\vec{e}}(\mathcal{L}^{i}_{\vec{e}}(p))$. 
In order to guarantee completeness, \cite{Platzer20} proposes the extended term conditions:  $C^{\infty}$ smoothness (i.e. have partial derivatives of any order), existence of  syntactic partial derivative representation, and the most critical one, 
which is, for any term $p$ and an ODE $\dot{x} = e$, there must exist $N>0$ and $N$ terms $g_i$ such that
the higher Lie derivatives of $p$ along the ODE satisfy $\mathcal{L}^{N}_{\vec{e}}(\textit{p})= \sum_{i=0}^{i=N-1} g_i \cdot \mathcal{L}^{i}_{\vec{e}}(\textit{p})$. As a result of the last condition, all the Lie derivatives  higher than $N$
 can also be written as sums over the Lie derivatives lower than $N$ with appropriate cofactors.  Polynomial terms obviously satisfy the extended term conditions. We adapt the complete invariant rule of \cite{Platzer20} to our case and present it below. }
 
 {
 \[ 
\prftree[r]{}{B \wedge \textit{Inv} \wedge \mathcal{L}^{*}_{\vec{e}}(B) \rightarrow \mathcal{L}^{*}_{\vec{e}}(\textit{Inv}) \quad B \wedge \neg \textit{Inv} \wedge \mathcal{L}^{*}_{-\vec{e}}(B) \rightarrow \mathcal{L}^{*}_{-\vec{e}}(\neg \textit{Inv})}{\spec{\textit{Inv}\wedge P}{\evo{x}{e}{B}}{\closeb(B) \wedge \closeb(\neg B) \wedge \textit{Inv}\wedge P\joinop \textsf{trInv}(d,\textit{Inv})}} \]
}

\noindent where $B\define \bigvee_i (\bigwedge_{j}p_{ij} \geq 0)$ and $\textit{Inv} \define \bigvee_i (\bigwedge_{j}q_{ij} > 0 )$ are both semi-analytic formulas,
with $p_{ij}$ and $q_{ij}$ being  the analytic  terms  with orders of Lie derivatives less than upper bound $N$, 
and 
{ \[
\begin{array}{rcl}
\mathcal{L}^{*}_{\vec{e}}(B) & \define  & \bigvee_i (\bigwedge_{j}\mathcal{L}^{*}_{\vec{e}}(p_{ij}) \geq 0 )\\
 \mathcal{L}^{*}_{\vec{e}}(\textit{Inv}) & \define  & \bigvee_i ( \bigwedge_{j}\mathcal{L}^{*}_{\vec{e}}(q_{ij}) > 0 )\\
  \mathcal{L}^{*}_{\vec{e}}(q) > 0 & \define & \left(\begin{array}{l}
   (q \geq 0) \wedge (q = 0 \rightarrow \mathcal{L}^{}_{\vec{e}}(q) \geq 0) \wedge \\
    (q = 0 \wedge \mathcal{L}^{}_{\vec{e}}(q) = 0
  \ \rightarrow \mathcal{L}^{2}_{\vec{e}}(q) \geq 0) \wedge \\
  \cdots  \wedge (q = 0 \wedge \cdots \wedge \mathcal{L}^{N-2}_{\vec{e}}(q) = 0 \rightarrow \mathcal{L}^{N-1}_{\vec{e}}(q) > 0) 
  \end{array} \right) \\
  \mathcal{L}^{*}_{\vec{e}}(q) \geq  0  & \define & (\mathcal{L}^{*}_{\vec{e}}(q) > 0) \vee (\bigwedge_{i=0}^{N-1}\mathcal{L}^{i}_{\vec{e}}(q)=0)\\
 \end{array}
\] }

The above rule is complete for semi-analytic invariant,  in the sense that $\textit{Inv}$  is a semi-analytic invariant for the ODE, iff the arithmetic premise holds. The premise states that, for  the positive case  of $\textit{Inv}$ (i.e. $q_{ij}$), there must exist some $i$ such that $i < N$ and the $i$-th Lie derivative   is greater than 0 and the lower Lie derivatives than $i$ are 0, and for the non-negative case (i.e. $p_{ij}$), it is weaker that all the Lie derivatives less than $N$ can be 0. Moreover,  the Lie derivatives of $\neg \textit{Inv}$ with respect to the backward ODE $\vec{\dot{x}} = -\vec{e}$ have the similar constraint.  At termination, the escaping point must belong to the boundary of domain $B$ (defined by the conjunction of the closures of $B$ and $\neg B$) and also satisfies $\textit{Inv}$.

\begin{theorem} 
\label{theorem:completeODE}
For semi-analytic formulas $\textit{Inv}$ and $Q$, $\spec{\textit{Inv}\wedge P}{\evo{x}{e}{B}}{Q \wedge P\joinop \textsf{trInv}(d,\textit{Inv})}$ can be derived from the proof system of HHL, iff the premise conditions on $\textit{Inv}$ as shown in the above differential invariant rule holds, and $\closeb(B) \wedge \closeb(\neg B) \wedge \textit{Inv} \rightarrow Q$ holds. 
\end{theorem}


\begin{proof}The specification defines that $\textit{Inv}$ holds initially, and during the execution of ODE, $\textit{Inv}$ is maintained as an invariant, indicated by $\textsf{trInv}(d,\textit{Inv})$; moreover, the postcondition $Q$ holds for the final state if $\closeb(B) \wedge \closeb(\neg B) \wedge \textit{Inv} \rightarrow Q$ holds. For the first part, it is equivalent to the corresponding specification of~\cite{Platzer20}, and we can directly inherit the proof of ~\cite{Platzer20} here.  The proof of ~\cite{Platzer20} is given based on  the lemmas on continuous existence, uniqueness, and differential adjoints etc. All these lemmas also hold in our case as we require that all ODEs satisfy the local Lipschitz condition. Thus the proof of ~\cite{Platzer20} still holds for our case. We can get the fact that $Inv$ holds for the whole evolution, including the final state at termination. 

For the second part, according to the semantics of $\evo{x}{e}{B}$, when it terminates, the final state must be at the boundary of $B$, thus both $\closeb{(B)}$ and $\closeb{(\neg B)}$ hold for the final state. Plus the fact that $Inv$ holds for the final state, from $\closeb(B) \wedge \closeb(\neg B) \wedge \textit{Inv} \rightarrow Q$, $Q$ holds for the final state. 
\end{proof}

 For continuous interrupt, the ODE part with explicit solutions can be replaced by differential invariants similarly. 


{
 \[  
\prftree[r]{}
{
\begin{array}{c}
     B \wedge \textit{Inv} \wedge \mathcal{L}^{*}_{\vec{e}}(B) \rightarrow \mathcal{L}^{*}_{\vec{e}}(\textit{Inv}) \quad B \wedge \neg \textit{Inv} \wedge \mathcal{L}^{*}_{-\vec{e}}(B) \rightarrow \mathcal{L}^{*}_{-\vec{e}}(\neg \textit{Inv})  \\
     \closeb(B) \wedge \closeb(\neg B) \wedge \textit{Inv}\wedge P\joinop \textsf{trInv}(d,\textit{Inv},\{\cup_{i\in I} ch_i*\})\rightarrow R \\
     ch_i*=ch_i!e\rightarrow\spec{\closeb(B) \wedge \textit{Inv}\wedge P\joinop \textsf{trInvOut}(d,\textit{Inv},\{\cup_{i\in I} ch_i*\},ch_i*,e)}{c_i}{R}\\
     ch_i*=ch_i?x\rightarrow\spec{\exists\,v.\,(\closeb(B) \wedge \textit{Inv})[v/x]\wedge P\joinop \textsf{trInvIn}(d,\textit{Inv},\{\cup_{i\in I} ch_i*\},ch_i*,v)}{c_i}{R}
\end{array}
}
{\spec{\textit{Inv}\wedge P}{\exempt{\evo{x}{e}{B}}{i\in I}{ch_i*}{c_i}}{R}} \]
}

where $\textsf{trInvOut}$ and $\textsf{trInvIn}$ are defined as:
{
\begin{align*}
    \textsf{trInvOut}
    &(d,\textit{Inv},\rdy,ch,e) \define \exists\vec{p}.\, \tracev = \langle d,\vec{p},\rdy\rangle^\chop\langle ch!,e[\vec{p}(\tau)/\vec{x}]\rangle\wedge \forall \tau\in[0,d].\, \textit{Inv}[\vec{p}(\tau)/\vec{x}]\\
    &\textsf{trInvIn}(d,\textit{Inv},\rdy,ch,v) \define \exists\vec{p}.\, \tracev = \langle d,\vec{p},\rdy\rangle^\chop\langle ch?,v\rangle\wedge \forall \tau\in[0,d].\, \textit{Inv}[\vec{p}(\tau)/\vec{x}]
\end{align*}
}
\paragraph{Some Derived Rules}
In the following, we discuss some derived differential rules, which could be applied to prove some complicated properties on continuous evolution more efficiently, as 
they provide sufficient conditions for being an invariant of an ODE, but not necessary.

The following differential invariant rule  says that whenever the Lie derivative of an expression \textit{q} w.r.t. the ODE within domain $B$ is zero, then $\textit{q}=c$ keeps invariant. 
\[ 
\prftree[r]{DiffInv}{B \rightarrow \mathcal{L}_{\vec{e}}(\textit{q})=0}
{\spec{\textit{q}=c\wedge P}{\evo{x}{e}{B}}{\textit{q}=c\wedge P\joinop \textsf{trInv}(d,\textit{q}=c)}}
\]
here $\mathcal{L}_{\vec{e}}(\textit{q})$ denotes the Lie derivative of \textit{q} w.r.t. the vector field $\vec{e}$. 
This rule could be useful when the differential equation cannot be solved exactly, and all we need is to prove some invariant property.
Moreover, we prove both the positive and negative cases for the above invariant rules, by changing the premise to be $\mathcal{L}_{\vec{e}}(\textit{q}) \geq 0$, and  the conclusion to be $q \geq c$ or $q > c$; symmetrically, the premise $\mathcal{L}_{\vec{e}}(\textit{q}) \leq 0$, with the conclusion  $\textit{q} \leq c$ or $\textit{q} < c$.

Likewise, we can prove a version of Darboux equality rule~\cite{Platzer20}:
\[ 
\prftree[r]{Dbx}{B \rightarrow \mathcal{L}_{\vec{e}}(\textit{q})=g\cdot \textit{q}}{\spec{\textit{q}=0\wedge P}{\evo{x}{e}{B}}{\textit{q}=0\wedge P\joinop \textsf{trInv}(d,\textit{q}=0)}}
\]
where $g$ is a continuous function. The intuition is that, when the first Lie derivative of $\textit{q}$ is a product between a continuous cofactor $g$ and $\textit{q}$, then its all higher Lie derivatives can also be written as a product between some cofactor and itself. Thus, when $\textit{q}$ is 0 initially, all its derivatives are 0, $\textit{q} = 0$ will stay invariant along the evolution. Similarly, we have also proved the positive and negative cases for Darboux inequalities and here will not list them. 

We can also prove the invariant property of ODE with the idea of barrier certificate \cite{DBLP:conf/hybrid/PrajnaJ04}:
{ \[ \hspace*{-1mm} 
\prftree[r]{Dbarrier}{B \rightarrow \textit{q}=0 \rightarrow \mathcal{L}_{\vec{e}}(\textit{q})<0}
{\spec{\textit{q}\le 0\wedge P}{\evo{x}{e}{B}}{\textit{q}\le 0\wedge P\joinop \textsf{trInv}(d,\textit{q}\le 0)}}
\] }
Whenever $\textit{q}$ reaches $0$ along the trajectory, the negative 
Lie derivatives push $\textit{q}$ to decrease, thus it can never exceed $0$ within $B$. The case with $\textit{q} \ge 0 $  can be proved in a similar way.

\subsection{Discussion on Partial Correctness, Total Correctness, 
Deadlock, Livelock and Invariants}
\emph{Partial correctness vs total correctness.} 

If we focus on partial correctness, our proof system needs to add  the following inference rule for deadlock:
\begin{eqnarray*}
& \prftree[r]{Deadlock}{\spec{P}{c}{\tracev=\delta}}{Q \mbox{ is any formula of \FODT}}{}{\spec{P}{c}{Q}} & 
\end{eqnarray*}
Rule Deadlock says that  a deadlocked behaviour can imply any property, as it never terminates. 
If we investigate total correctness, we need a rule for variant  as in classical Hoare logic. We will address this issue together with \emph{livelock} in future work. 

\noindent 
\emph{Deadlock and livelock.} \emph{Deadlock} can be handled in our proof system with the rules on trace synchronization and the above rule, but {\em livelock} is not considered. {\em Livelock} could be handled by proving/disproving its existence/absence  like program termination analysis, or 
by allowing recording internal action in traces and checking whether 
there are infinitely many internal actions in a finite time horizon. 
As said above, we will address this issue in future work. 

\noindent 
\emph{Loop invariant and differential invariant.}
Reasoning about repetition needs invariants. Just as discussed in Sect.~\ref{subsec:dinvariant}, although continuous evolution can be reasoned about by explicitly using  its solution as indicated in rule Cont, differential invariants can 
ease the reasoning very much as obtaining a solution of an ODE is mathematically  difficult. 
As in classical Hoare logic, invariant generation plays a central role in deductive verification of HSs. 
But in HSs, one has to consider to synthesize \emph{global invariants} (for loops and recursions) and 
local (differential) invariants (for ODEs, as shown in Sect.~\ref{subsec:dinvariant}) simultaneously.  As discussed in \cite{PlatzerClarke08}, synthesizing global invariants can be achieved by combining invariant generation techniques for discrete programs and differential invariant generation techniques for ODEs. As we will see in the case study section, we verify the second case study using both the notions of global invariants and local differential invariants. 
In the literature, there are various works on differential invariant synthesis for dynamical systems. 
\cite{LiuEmsoft} gave a necessary and sufficient condition for a semi-algebraic set to be a differential invariant of a polynomial system. Based on which, \cite{WangCAV2021} proposed an efficient approach for synthesizing semi-algebraic invariants for polynomial dynamical systems by exploiting 
difference of convex programming.   \cite{Platzer20} presented  a complete axiomatic system for reasoning about differential invariants based on a similar condition to \cite{LiuEmsoft}.

Alternatively, reasoning about continuous evolution  can be conducted  by discretization, e.g., \cite{TOSEM20} presented  a set of refinement rules to discretize HCSP, further refined discretized HCSP  to SystemC. 
\cite{LP16} presented a refinement logic 
which investigates the inverse direction to reduce verification of discrete systems to verification of hybrid systems.

\section{ \hspace*{-4mm} Continuous relative completeness}\label{sec:completeness}

In this section, we show continuous relative completeness of HHL w.r.t. FOD. 
This is done in two steps. First, we show that the proof system is complete if all weakest liberal preconditions/strongest postconditions can be expressed in FOD and all valid entailments between predicates can be proved. Second, we 
show that all required predicates can be expressed as formulas of FOD. Following the form of the proof system given above, weakest liberal preconditions are used in the sequential case, and strongest postconditions are used in the parallel case.

\subsection{Weakest Liberal Preconditions and Strongest Postconditions}

For the sequential case, given a process \sm{c} and postcondition \sm{Q}, the weakest liberal precondition \sm{\wp(c,Q)} is a predicate on state and trace pairs, defined as:
{ \[\wp(c,Q) = \{(s, tr)\,|\,\forall s', tr'. (c, s) \Rightarrow (s', tr') \rightarrow Q(s', tr^\chop tr')\}\]}
Thus, the computation of weakest liberal preconditions is straightforward from the definition of big-step semantics. Most of them correspond directly to the Hoare rules in Fig.~\ref{fig:hoarelogic}. 
The only $\wp$ rule that does not allow direct computation is that for repetition. Instead it satisfies the following recursive equation:
\[\wp(c^*, Q) = \wp(c; c^*, Q)\]
The recursive  equation is not solvable in general, but it provides a way to approximate $\wp(c^*,Q)$ according to provided invariants for $c^*$. So,  as in the verification of programs with classical Hoare logic, invariant generation plays a central role in the verification of HSs with HHL. The computation of weakest liberal precondition is given in Fig.~\ref{fig:wp}. Justification of this computation is given as part of the soundness proof in previous section.
\begin{figure*}
{
    \[ \wp(\pskip, Q) = Q \quad \wp(x := e, Q) = Q[e/x] \]	
	\[ \begin{array}{ll} \wp(ch!e, Q) =
		&Q[\tracev^\chop \langle ch!, e\rangle/\tracev]\,\wedge \forall d>0.\,Q[\tracev^\chop \langle d, I_{\vec{x}_0}, \{ch!\}\rangle^\chop \langle ch!, e\rangle/\tracev]\,\wedge  Q[\tracev^\chop \langle \infty,I_{\vec{x}_0},\{ch!\}\rangle/\tracev]
	\end{array}
	\]	
	\[ \begin{array}{ll} \wp(ch?x, Q) =
		&\forall v.\, Q[v/x, \tracev^\chop \langle ch?, v\rangle/\tracev]\,\wedge \\
        &\forall d>0.\,\forall v.\, Q[v/x, \tracev^\chop \langle d,I_{\vec{x}_0},\{ch?\}\rangle ^\chop \langle ch?,v\rangle/\tracev]\,\wedge  Q[\tracev^\chop \langle \infty,I_{\vec{x}_0},\{ch?\}\rangle/\tracev]
	\end{array}
	\]
	\[ \wp(c1; c2, Q) = \wp(c1, \wp(c2, Q)) \]
	\[ \wp(\IFE{b}{c_1}{c_2}, Q) = \IFE{b}{\wp(c_1,Q)}{\wp(c_2,Q)} \]
	\[ \wp(c_1\sqcup c_2, Q) = \wp(c_1, Q) \wedge \wp(c_2, Q) \]
	\[
	\begin{array}{ll}
		\wp(\evo{x}{e}{B},Q) = &(\neg B \rightarrow Q)\,\wedge  \forall d>0.\,(\forall t\in[0,d).\,B[\vec{p}_{\vec{x}_0}(t)/\vec{x}])\wedge \\
        &\neg B[\vec{p}_{\vec{x}_0}(d)/\vec{x}] \to  Q[\vec{p}_{\vec{x}_0}(d)/\vec{x}, \tracev^\chop \langle d, \vec{p}_{\vec{x}_0}, \emptyset\rangle/\tracev]
	\end{array}
	\]
	\oomit{
	\[
	\begin{array}{ll}
		&\wp(\external{i\in I}{ch_i*}{c_i}, R) = \\
		&\qquad (\forall i\in I.\, \textrm{if } ch_i*=ch!e \textrm{ then} \\
		&\qquad\qquad\qquad \wp(c_i, R)[\tracev^\chop \langle ch!,e\rangle/\tracev]\,\wedge \\
		&\qquad\qquad\qquad \forall d>0.\, \wp(c_i, R)[\tracev^\chop \langle d,I_{\vec{x}_0},\{\cup_{i\in I} ch_i*\}\rangle^\chop \langle ch!,e\rangle/\tracev] \\
		&\qquad\qquad\quad\ \textrm{elif } ch_i*=ch?x \textrm{ then} \\
		&\qquad\qquad\qquad \forall v.\, \wp(c_i, R)[v/x, \tracev^\chop \langle ch?,v\rangle/\tracev]\,\wedge \\
		&\qquad\qquad\qquad \forall d>0.\, \forall v.\, \wp(c_i, R)[v/x, \tracev^\chop \langle d,I_{\vec{x}_0},\{\cup_{i\in I} ch_i*\}\rangle^\chop \langle ch?,v\rangle/\tracev])\,\wedge \\
		&\qquad R[\tracev^\chop\langle \infty,I_{\vec{x}_0},\{\cup_{i\in I} ch_i*\}\rangle/\tracev]
	\end{array}
	\]}
	\[
	\begin{array}{ll}
		&\wp(\exempt{\evo{s}{e}{B}}{i\in I}{ch_i*}{c_i}, R) = \\
		&\qquad (\forall i\in I.\, \textrm{if } ch_i*=ch!e \textrm{ then }   \wp(c_i, R)[\tracev^\chop \langle ch!,e\rangle/\tracev]\,\wedge \\
		&\qquad\qquad\qquad \forall d>0.\, (\forall t\in[0,d).\, B[\vec{p}_{\vec{x}_0}(t)/\vec{x}]) \to\\
        &\qquad\qquad\qquad\qquad
        \wp(c_i, R)[\vec{p}_{\vec{x}_0}(d)/\vec{x}, \tracev^\chop \langle d,\vec{p}_{\vec{x}_0},\{\cup_{i\in I} ch_i*\}\rangle^\chop \langle ch!,e[\vec{p}_{\vec{x}_0}(d)/\vec{x}]\rangle/\tracev] \\
		&\qquad\qquad\quad\ \textrm{elif } ch_i*=ch?x \textrm{ then }   \forall v.\, \wp(c_i, R)[v/x, \tracev^\chop \langle ch?,v\rangle /\tracev]\,\wedge \\
		&\qquad\qquad\qquad \forall d>0.\, (\forall t\in[0,d).\,
		B[\vec{p}_{\vec{x}_0}(t)/\vec{x}]) \to\\
        &\qquad\qquad\qquad\qquad\wp(c_i, R)[\vec{p}_{\vec{x}_0}(d)/\vec{x}, v/x, \tracev^\chop \langle d,\vec{p}_{\vec{x}_0},\{\cup_{i\in I} ch_i*\}\rangle^\chop \langle ch?,v\rangle/\tracev])\, \wedge \\
		&\qquad (\neg B \rightarrow R)\, \wedge   \forall d>0.\,(\forall t\in[0,d).\,B[\vec{p}_{\vec{x}_0}(t)/\vec{x}])\wedge \\
        &\qquad\neg B[\vec{p}_{\vec{x}_0}(d)/\vec{x}] \to  R[\vec{p}_{\vec{x}_0}(d)/\vec{x}, \tracev^\chop \langle d, \vec{p}_{\vec{x}_0}, \{\cup_{i\in I} ch_i*\}\rangle/\tracev]
	\end{array}
	\]
	\caption{The weakest liberal preconditions for HCSP processes}
	\label{fig:wp}} 
\end{figure*}
Regarding \sm{\wp}, we have the following result.
\begin{lemma}\label{lem:wp-complete}
	For any sequential process \sm{c},  \sm{\vdash \{\wp(c,Q)\}\ c\ \{Q\}}.
\end{lemma}
\begin{proof}[Proof for Lemma~\ref{lem:wp-complete}]
	The proof is by induction on the structure of program \sm{c}. For most statements, the result follows directly by comparing the $\wp$-rule with the corresponding Hoare rule. We explain the more interesting cases in detail.
	
	
	For the case of repetition, we wish to prove \sm{\vdash \{\wp(c^*,Q)\}\ c^*\ \{Q\}}, given the inductive assumption \sm{\vdash \{\wp(c,Q')\}\ c\ \{Q'\}} for any \sm{Q'}. For this, we make use of the following property of \sm{\wp(c^*, Q)}:
	 $ \wp(c^*, Q) \rightarrow Q$, 
	which follows from the equation satisfied by \sm{\wp(c^*,Q)}. This allows us to reduce the goal to proving \sm{\vdash \{\wp(c^*,Q)\}\ c^*\ \{\wp(c^*,Q)\}}, and using the repetition rule, to proving \sm{\vdash \{\wp(c^*,Q)\}\ c\ \{\wp(c^*,Q) \}}.
	
	By the inductive assumption, we have \sm{\vdash \{\wp(c,\wp(c^*,Q))\}\ c\ \{\wp(c^*,Q)\}}. Hence, it suffices to show
	\[ \wp(c^*,Q) \rightarrow \wp(c,\wp(c^*,Q)), \]
	which also follows from the equation satisfied by \sm{\wp(c^*,Q)}.
	
	Next, we consider the case of interrupt. We need to show
	\[ \begin{array}{ll}
	\vdash \{\wp(\exempt{\evo{x}{e}{B}}{i\in I}{ch_i*}{c_i}, R)\}  \\
	\qquad \exempt{\evo{x}{e}{B}}{i\in I}{ch_i*}{c_i}\\
	\quad \{R\} \end{array} \]
	given the inductive assumption that \sm{\vdash \{\wp(c_i,R)\}\ c_i\ \{R\}} for any index \sm{i}. Apply the rule Interrupt, with the indexed family of assertions \sm{Q_i} given by \sm{\wp(c_i,R)}. By the inductive hypothesis, each Hoare triple in the assumption is provable. Moreover, each entailment in the assumption holds by the definition of \sm{\wp(\exempt{\evo{x}{e}{B}}{i\in I}{ch_i*}{c_i}, R)}. This finishes the proof for the interrupt case.
\end{proof}
Next, we consider the case of parallel processes. Here we make use of strongest postconditions. Given a parallel process \sm{c} w.r.t. a given precondition \sm{P} on the global state, the strongest postcondition \sm{\sp(c,P)} is a predicate on state and trace, such that \sm{\sp(c,P)(s',\textit{tr}')} holds iff  there exists \sm{s} satisfying \sm{P}, such that it is possible to go from \sm{(c,s)} to \sm{(s',\textit{tr}')} under big-step semantics.

The strongest postcondition can be recursively computed for preconditions $P$ in the form of conjunctions of predicates on individual processes. For a single process, it is equivalent to the strongest postcondition for sequential processes. For the parallel composition of two processes, { 
 define $P_1 \uplus P_2$  by $(P_1\uplus P_2)(s_1\uplus s_2)=P_1(s_1)\wedge P_2(s_2)$,} then we have:
{ \[ 
\begin{array}{l}
	\sp(c_1\|_{cs} c_2, P_1 \uplus P_2)(s_1\uplus s_2, \textit{tr}) = \\
	(\exists tr_1\ tr_2.\, \sp(c_1,P_1)[tr_1/\textit{tr}] \wedge \sp(c_2,P_2)[tr_2/\textit{tr}] \wedge tr_1\|_{cs}tr_2 \Downarrow \textit{tr})
\end{array} \]}
From this, we get the following lemma:
\begin{lemma}\label{lem:sp-parallel-complete}
	For any parallel process \sm{c}, and precondition \sm{P} in the form of conjunction of predicates on individual processes of \sm{c}, we have \sm{\vdash \{P\}\ c\ \{\sp(c,P)\}}.
\end{lemma}
\begin{proof}[Proof for Lemma~\ref{lem:sp-parallel-complete}]
	The proof is by induction on the structure of $c$. For the base case of sequential processes, this follows from Lemma~\ref{lem:wp-complete} and the definition of $\wp$. For the parallel composition of two processes, this follows from the rule (Par) and the computation of $\sp(c_1\|_{cs}c_2, P_1\uplus P_2)$.
\end{proof}

From Lemma~\ref{lem:wp-complete} and Lemma~\ref{lem:sp-parallel-complete}, we get the following theorem, under the assumption of expressibility of predicates and provability of entailments in the underlying logical system.
\begin{theorem}\label{thm:continuous-complete}
Every valid HHL goal \sm{\{P\}\ c\ \{Q\}} is provable in the above system given an oracle for \text{FOD}.
\end{theorem}
\begin{proof}[Proof for Theorem~\ref{thm:continuous-complete}]
	For sequential case, if $\{P\}\ c\ \{Q\}$ is valid, $P\rightarrow \wp(c,P)$ holds, and by assumption provable. Combining with Lemma~\ref{lem:wp-complete}, we have $\{P\}\ c\ \{Q\}$ is provable.  For parallel case, if $\{P\}\ c\ \{Q\}$ is valid, then $\sp(c,P)\rightarrow Q$ holds, and by assumption provable. Combining with Lemma \ref{lem:sp-parallel-complete}, we have $\{P\}\ c\ \{Q\}$ is provable.
\end{proof}

\subsection{Expressing Predicates in FOD}
By Theorem~\ref{thm:continuous-complete}, in order to prove the continuous relative completeness, the only remaining step is to  show expressibility of traces and trace
assertions in FOD. We follow Platzer's approach in~\cite{Platzer12a} to encode  trace-based assertions in FOD using the standard Gödelisation technique of Cook. 
For  simplicity, we use  \sm{\vec{x}_0 \langle \vec{\dot{x}}=\vec{e} \rangle \vec{x}_1 } to mean  
that starting from vector \sm{\vec{x}_0}, following the differential equation \sm{\vec{\dot{x}}=\vec{e}}, \sm{\vec{x}_1} can be reached.

\subsubsection{Encoding traces}

First, we discuss how to encode traces that appear in the previous sections in FOD. Using the $\mathbb{R}$-G\"odel encoding in \cite{Platzer12a}, it is possible to encode any sequence of real numbers of fixed length as a single real number. The basic idea (for two real numbers) is as follows. Suppose real numbers $a$ and $b$ are written as $a_0.a_1a_2\dots$ and $b_0.b_1b_2\dots$ in binary form, then the pair $(a,b)$ can be represented as $a_0b_0.a_1b_1a_2b_2\dots$. In fact, we can extend this encoding to a sequence \sm{\vec{y}} of real numbers with arbitrary length, by encoding \sm{\vec{y}} as the $\mathbb{R}$-encoding of the pair \sm{(l,y)}, where \sm{l} stores the length of \sm{\vec{y}}, and \sm{y} is the $\mathbb{R}$-G\"odel encoding of \sm{\vec{y}}. From now on, we will make implicit use of this encoding, allowing us to quantify over sequences of real numbers of arbitrary length, and adding to the language the function \sm{\len(x)} for the length of sequence \sm{x}, and \sm{x_i} (with \sm{1\le i\le \len(x)}) for the \sm{i^{th}} component of \sm{x}.

Given an HCSP process, we can fix a mapping from the channel names appearing in the process to natural numbers. Hence a communication event  of the form \sm{\langle ch\triangleright,v\rangle} can be encoded as a real number. Encoding a continuous event 
of the form \sm{\langle d,\vec{p},\rdy\rangle} requires more care. The ready set can be encoded as a natural number since the total number of channels is finite. The main problem is how to encode the state trajectory  \sm{\vec{p}} in the continuous event.
We make the restriction that any state trajectory 
\sm{\vec{p}} appearing in a continuous event of the trace must be either constant or a solution of an ODE appearing in the HCSP process. This restriction is reasonable since any other state trajectory 
cannot possibly appear in the behavior of the process. We number the ODEs appearing in the process as \sm{\langle \vec{\dot{x}}=\vec{e_1}\rangle,\dots,\langle\vec{\dot{x}}=\vec{e_k}\rangle}, and let  \sm{\langle\vec{\dot{x}}=\vec{e_0}\rangle} be the ODE \sm{\langle \vec{\dot{x}}=0\rangle} (for the case of constant state trajectories).

First, we consider the sequential case. Then a state trajectory 
\sm{\vec{p}} in a continuous event  can be encoded as a triple \sm{(d,\vec{p}_0,i)}, where \sm{d} is its duration, \sm{\vec{p}_0} is the initial state, and \sm{0\le i\le k} is the index of the differential equation satisfied by \sm{\vec{p}}. We can then define \sm{\vec{p}_{\varsigma}} for $\varsigma<d$, the state of the state trajectory 
at time $\varsigma$, as the unique state satisfying the FOD formula 
\[ (\vec{x}=\vec{p}_0 \wedge t=0)\langle \vec{\dot{x}}=\vec{e}_i, t'=1 \rangle (\vec{x}=\vec{p}_{\varsigma} \wedge t=\varsigma).  \]
For the parallel case, the state is eventually divided into component states for sequential processes, so that each component state follows one of the ODEs \sm{\vec{\dot{x}}=\vec{e}_i}. Hence, it can be encoded as a binary tree where each leaf node contains a tuple of the form \sm{(d,\vec{p}_0,i)}. For example, if we have a parallel of two sequential processes, with a path starting from state \sm{\vec{p}_0} and following ODE \sm{i} on the left, and starting from state \sm{\vec{q}_0} and following ODE \sm{j} on the right, then this state trajectory  is encoded as \sm{((d,\vec{p}_0,i),(d,\vec{q}_0,j))}.

With this encoding, it is clear that given representations of state trajectories  \sm{\vec{p}_1} and \sm{\vec{p}_2}, the state trajectory  \sm{\vec{p}_1\uplus \vec{p}_2} can be represented. For the purpose of encoding synchronization below, we also need to encode \sm{\vec{p}(\cdot + d')}. With \sm{\vec{p}} given as \sm{(d,\vec{p}_0,i)}, this is simply \sm{(d-d',\vec{p}_{d'},i)}.

Hence, we can encode any general event  as a single real number. Then a trace can also be encoded as a single real number. So, we can also encode operations on traces such as the join operation.

\subsubsection{Encoding synchronization} A key relation that needs to be encoded is the synchronization relation \sm{tr_1\|_{cs} tr_2 \Downarrow \textit{tr}}. We note that each of the synchronization rules (except SyncEmpty) produces an extra general event  in \sm{\textit{tr}}. Hence, the derivation of \sm{tr_1\|_{cs}tr_2 \Downarrow \textit{tr}} consists of exactly \sm{\len(\textit{tr})} steps plus a final SyncEmpty step. We encode the relation by encoding the entire derivation using two sequences of traces \sm{T_1} and \sm{T_2}, intended to represent intermediate traces for \sm{tr_1} and \sm{tr_2}:
{ \[ 
\begin{array}{ll}
	tr_1\|_{cs}tr_2\Downarrow \textit{tr} \equiv \exists T_1,\,T_2.
	 \len(T_1) = \len(\textit{tr}_1) + 1 \wedge \,  \\ \  \len(T_2) = \len(\textit{tr}_2) + 1
	\wedge T_{11} = tr_1 \wedge T_{21} = tr_2 \wedge T_{1,k+1} = T_{2,k+1}=\epsilon\, \\
	\qquad  \wedge\forall 1\le i\le k.\,
	\textit{Step}_{cs}(T_{1i},T_{1,i+1},T_{2i},T_{2,i+1},tr[i])
\end{array}
\] }
where \sm{	\textit{Step}_{cs}(tr_1,tr_1',tr_2,tr_2',e)} means one step in the derivation of the synchronization relation, reducing $tr_1$ to $tr_1'$ and $tr_2$ to $tr_2'$, and producing event $e$. It is obtained by encoding the definition of synchronization in Fig.~\ref{fig:rule-synchronization}.

\subsubsection{Encoding the predicates}

We now show that using the above encoding, each of the weakest liberal precondition formulas can be rewritten in the language of FOD. For the sequential case, the only tricky case is encoding the continuous evolution using the method above. As an example, consider the second conjunct of the weakest liberal precondition for ODEs:
{ \[ 
\begin{array}{ll}
\forall d>0.\, (\forall t\in[0,d).\, B[\vec{p}_{\vec{x}_0}(t)/\vec{x}] \wedge \neg B[\vec{p}_{\vec{x}_0}(d)/\vec{x}]) \to \\
\qquad
Q[\vec{p}_{\vec{x}_0}(d)/\vec{x}, \tracev^\chop\langle d,\vec{p}_{\vec{x}_0},\emptyset\rangle/\tracev].
\end{array}
\] }
Suppose the differential equation \sm{\vec{\dot{x}}=\vec{e}} is numbered as \sm{\vec{\dot{x}}=\vec{e}_i}, then the condition can be written equivalently as follows, unfolding the encoding of \sm{\vec{p}_{\vec{x}_0}} as \sm{(d,\vec{x}_0,i)}:
\[ 
\begin{array}{l}
\forall d>0.\, (\forall t\in[0,d).\, B[(d,\vec{x}_0,i)_t/\vec{x}] \wedge \neg B[(d,\vec{x}_0,i)_d/\vec{x}])\to \, \\ \qquad
Q[(d,\vec{x}_0,i)_d/\vec{x}, \tracev^\chop\langle d, (d,\vec{x}_0,i),\emptyset\rangle/\tracev].
\end{array}\]
For repetition, define by induction on natural numbers:
{  \[\begin{array}{rl}
	\wp_0(c^*,Q) = true, & 
	\wp_{k+1}(c^*,Q) = \wp_k(c; c^*, Q).
\end{array}\] }
Then we can write:
$ \wp(c^*,Q) = (\exists k >0.\, \wp_k(c^*, Q)) $.

This concludes the case of sequential programs. For the case of parallel programs, the only extra relation is the synchronization relation, whose expressibility is shown in the previous subsection.

\section{Discrete relative completeness}
\label{sec:discretecomplete}
In this section, we prove the discrete relative completeness of our proof system in the sense that all continuous evolutions can be approximated by discrete actions with arbitrary precision, which is similar to the discrete relative completeness for \dL\ in \cite{Platzer12a}.

 We say two generalized events  are within distance $\epsilon$ if 
\begin{itemize}
	\item  they are both wait events, of the form \sm{\langle d_1,\vec{p}_1,\rdy_1\rangle} and \sm{\langle d_2,\vec{p}_2,\rdy_2\rangle} respectively, and \sm{|d_1-d_2|<\epsilon}, \sm{\|\vec{p}_1(t)-\vec{p}_2(t)\|<\epsilon} for all \sm{0<t<\min(d_1,d_2)}, and \sm{\rdy_1=\rdy_2}; or
	\item  they are both communication events, of the form \sm{\langle ch_1\triangleright_1,v_1\rangle} and \sm{\langle ch_2\triangleright_2,v_2\rangle}, and \sm{ch_1\triangleright_1=ch_2\triangleright_2} and \sm{|v_1-v_2|<\epsilon}.
\end{itemize}
Two traces \sm{tr_1} and \sm{tr_2} are within distance $\epsilon$, denoted by \sm{d(tr_1, tr_2) < \epsilon}, if they contain the same number of generalized events, and each pair of corresponding generalized events are within $\epsilon$. Two states \sm{s_1} and \sm{s_2} are within distance $\epsilon$, denoted by \sm{d(s_1, s_2) < \epsilon},  if \sm{|s_1(x) - s_2(x)| < \epsilon} for all variables \sm{x}. The \sm{\epsilon$-neighborhood $\mathcal{U}_{\epsilon}(s, \textit{tr})} of a pair \sm{(s, \textit{tr})} of state and trace is defined as:
{  \[\mathcal{U}_{\epsilon}(s, \textit{tr}) \define \{(s', \textit{tr}') ~|~ d(s, s') < \epsilon \wedge  d(\textit{tr},\textit{tr}') < \epsilon\}. \]}
Similarly, we define the $\epsilon$-neighborhood of a state \sm{s}. The $-\epsilon$ neighborhood of a predicate \sm{Q} on pairs of state and trace is the set of pairs for which their $\epsilon$-neighborhoods satisfy \sm{Q}. That is,
{  \[ \mathcal{U}_{-\epsilon}(Q) \define \{(s, \textit{tr}) ~|~ \forall (s', \textit{tr}') \in \mathcal{U}_{\epsilon}(s, \textit{tr}).\, Q(s', \textit{tr}')\}. \]}
Similarly, we define  \sm{\mathcal{U}_{-\epsilon}(B)} for a predicate \sm{B} on states only. A predicate \sm{Q} is \emph{open} if for any state \sm{s} and trace \sm{\textit{tr}} satisfying \sm{Q}, there exists $\epsilon>0$ such that \sm{Q(s', \textit{tr}')} holds for all \sm{(s', \textit{tr}') \in \mathcal{U}_{\epsilon}(s, \textit{tr})}.  For discrete relative completeness, we  consider the case for \emph{open} pre- and post-conditions first, then show  the general case in  Theorem~\ref{thm:dicreteRC}.

Next, we define the Euler approximation to the solution \sm{\vec{p}_{\vec{x}_0}} of an ODE \sm{\vec{\dot{x}}=\vec{e}} w.r.t. an initial vector \sm{\vec{x}_0}. Given a step size \sm{h>0}, a \emph{discrete solution} starting at \sm{\vec{x}_0} is a sequence \sm{(\vec{x}_0,\dots,\vec{x}_n,\dots)} with   \sm{\vec{x}_{i+1} = \vec{x}_i + h\cdot \vec{e}(\vec{x}_i)}, for \sm{i=0,1,\ldots}. We define the continuous approximation by joining the discrete points with straight lines. Define a function \sm{\vec{f}_{\vec{x}_0,h}:\mathbb{R}^+\to S} by \sm{\vec{f}_{\vec{x}_0,h}((i+t)h) = \vec{x}_i + (\vec{x}_{i+1}-\vec{x}_i)t} for any integer \sm{i\ge 0} and fractional part \sm{0\le t<1}.

\subsection{Discretization Rule for Continuous Evolution}
We first  define assertions equivalent to those appearing in the rule for ODEs
with the assumption that  \sm{B} is \emph{open}. This assumption will be dropped later when proving the final theorem for discrete completeness. 
The original precondition on the continuous solution is
{   \[ 
\begin{array}{l}
	\forall d>0.\, (\forall t\in[0,d).\, B[\vec{p}_{\vec{x}_0}(t)/\vec{x}]) \wedge \neg B[\vec{p}_{\vec{x}_0}(d)/\vec{x}] \to \\
	 \qquad 
	Q[\vec{p}_{\vec{x}_0}(d)/\vec{x}, \tracev^\chop \langle d, \vec{p}_{\vec{x}_0}, \emptyset\rangle/\tracev]
\end{array} 
\tag{CP}\] }
Note that  there is at most one value of \sm{d} for which the precondition holds, which is the maximal duration that  the ODE can continuously evolve subject to \sm{B}. \sm{d} may not exist in case the ODE has an infinite-time-horizon trajectory on which \sm{B} always holds. 

The discrete version of the assertion, without mentioning solutions to ODEs, is as follows:
{  \[ 
\begin{array}{ll}
	\forall T>0.\, (\forall 0\le t<T.\, \exists \epsilon_0>0.\, \forall 0<\epsilon<\epsilon_0.\,  \\ \qquad \quad
	\exists h_0>0.\, \forall 0<h<h_0.\, \vec{f}_{\vec{x}_0,h}(t)\in \mathcal{U}_{-\epsilon}(B)) \to \\
 	 (\exists \epsilon_0>0.\, \forall 0<\epsilon<\epsilon_0.\, \exists h_0>0.\, \forall 0<h<h_0.\, 
	  
	\\ \hspace*{-7mm} \vec{f}_{\vec{x}_0,h}(T) \in \neg \mathcal{U}_{-\epsilon}(B) \to
	(\vec{f}_{\vec{x}_0,h}(T), \tracev^\chop \langle T, \vec{f}_{\vec{x}_0,h}, \emptyset \rangle) \in \mathcal{U}_{-\epsilon}(Q))
	\tag{DP}	
\end{array} \] }
where \sm{\vec{f}_{\vec{x}_0,h}} is the continuous approximation defined above. 
{ 
As stated by (DP), the continuous approximation  of the solution at time $t$, i.e. $\vec{f}_{\vec{x}_0,h}(t)$,  is within $B$ and moreover the distance between it  and the boundary of $B$ must be at least $\epsilon$, and if  the continuous approximation  of the solution at exiting  time $T$ is no longer within the $-\epsilon$ neighborhood of $B$, then  the corresponding state and trace pair at time $T$ must be within $Q$ and the distance between it and the boundary of postcondition $Q$ must be at least  $\epsilon$.}

To justify the equivalence between the predicates CP and DP, we first need to estimate the global error of Euler approximations. According to~\cite{Platzer12a}, we have the following theorem.

\begin{theorem} \label{thm:globalerror}
	Let \sm{\vec{p}(t)} be a solution of the initial value problem \sm{\vec{\dot{x}} = \vec{e}}, \sm{\vec{p}(0)=\vec{x}_0} on the time interval \sm{[0, T]}.
	Let \sm{L} be the Lipschitz constant of the ODE \sm{\vec{\dot{x}}=\vec{e}}, that is, for any compact set \sm{S} of $\mathbb{R}^n$, \sm{\|\vec{e}(\vec{y}_1) - \vec{e}(\vec{y}_2)\| \leq L\|\vec{y}_1 - \vec{y}_2\|} for all \sm{\vec{y}_1, \vec{y}_2 \in S}. Then there exists a step size \sm{h_e > 0} such that for all \sm{0<h\leq h_e} and all \sm{n} with \sm{nh \leq T}, the global discretization error between \sm{\vec{p}(nh)}  and the discrete solution
	\sm{\vec{x}_{n} = \vec{x}_{n-1} + h \cdot \vec{e}(\vec{x}_{n-1})} satisfies 
{  \[
		\|\vec{p}(nh) - \vec{x}_{n}\| \leq  \frac{h}{2}\max_{\theta \in [0, T]} \left\| \frac{d^2 \vec{p}}{dt^2}(\theta) \right\|\fracN{e^{LT} - 1}{L}.
  \] }
\end{theorem}

Using Theorem~\ref{thm:globalerror}, we can always find a sufficiently small time step $h$ such that the error between the exact solution and the discrete approximations is arbitrarily small. 
Next, we state a version of Theorem~\ref{thm:globalerror} that is better suited for the following proof.
\begin{lemma}
	Let $\vec{p}(t)$ be a solution to the initial value problem $\vec{\dot{x}}=\vec{e}$, $\vec{p}(0)=\vec{x}_0$ on the time interval $[0,T]$, and let $L$ be the Lipshitz constant as before. Given any $\epsilon>0$, there exists $h_0>0$ such that for all $0<h<h_0$, the difference between the actual solution $\vec{p}$ and its continuous approximation $\vec{f}_h$ is at most $\epsilon$ on the interval $[0,T]$.
	\label{lem:globalerror}
\end{lemma}
\begin{proof}
	First, from Theorem~\ref{thm:globalerror}, take $h_1$ such that for all $0<h<h_1$ and all $n$ with $nh\le T$, we get
	\begin{equation*}
		\|\vec{p}(nh) - \vec{x}_{n}\| \leq \frac{h}{2}\max_{\theta \in [0, T]} \left\| \frac{d^2 \vec{p}}{dt^2}(\theta) \right\|\fracN{e^{LT} - 1}{L} < \frac{\epsilon}{3}.
	\end{equation*}
	This bounds the difference between the actual solution and the discrete approximation. For the continuous approximation, we further need to consider the intermediate points between $nh$ and $(n+1)h$. Hence, take $t\in(nh, (n+1)h)$, by the mean value theorem, there is a $\theta\in(nh,t)$ such that
	\[ \| \vec{p}(t) - \vec{p}(nh)\| = (t-nh)\cdot \left\|\frac{d \vec{p}}{d t}(\theta)\right\|\]
	Since $\|\frac{d \vec{p}}{d t}(\theta)\|$ is bounded along the path $\vec{p}$, we can take $h_2$ sufficiently small so that for any $0<h<h_2$, $\|\vec{p}(t)-\vec{p}(nh)\|$ is bounded above by $\frac{\epsilon}{3}$. Likewise, since
	\[ \|\vec{f}_h(t) - \vec{x}_n \| = \|\vec{f}_h(t) - \vec{f}_h(nh) \| = (t-nh)\cdot \left\|\frac{d\vec{p}}{d t}(nh)\right\|, \]
    we have any $0<h<h_2$, also $\|\vec{f}_h(t)-\vec{x}_n \|$ is bounded above by $\frac{\epsilon}{3}$. Take $h_0=\min(h_1,h_2)$ and combining, we have for any $0<h<h_0$,
	\[ 
	\begin{array}{lll}
		 &\|\vec{p}(t) - \vec{f}_h(t)\| \\
		&\qquad\le\|\vec{p}(t) - \vec{p}(nh)\| + \|\vec{p}(nh) - \vec{x}_n\| +
		\|\vec{f}_h(t)-\vec{x}_n\| \\
		 &\qquad < \epsilon, 
	\end{array}\]
	as desired.
\end{proof}
\begin{theorem} \label{thm:CP-DP} 
\sm{\models \text{CP} \leftrightarrow \text{DP}}.    
\end{theorem}
\begin{proof}[Proof for Theorem~\ref{thm:CP-DP}]
(CP) $\rightarrow$ (DP): Assume (CP) is true in state $s$ and trace $\textit{tr}$. This means for any $d>0$, suppose $\forall t\in[0,d).\,  B[\vec{p}_{\vec{x}_0}(t)/\vec{x}]$ and $\neg B[\vec{p}_{\vec{x}_0}(d)/\vec{x}]$ hold, where  $\vec{p}_{\vec{x}_0}$ is the unique solution of $\vec{\dot{s}}=\vec{e}$ with the initial value $\vec{x}_0 = s(\vec{x})$, then\\
$\nseman{Q}{s[\vec{x} \mapsto \vec{p}_{\vec{x}_0}(d)]}{\textit{tr}^\chop \langle d, \vec{p}_{\vec{x}_0}, \emptyset\rangle}$ also holds. For ease of presentation, we will abbreviate the latter to  $Q(\vec{p}_{\vec{x}_0}(d),\textit{tr}^\chop \langle d, \vec{p}_{\vec{x}_0}, \emptyset\rangle)$ below.

Fix $T>0$ in (DP), there are three cases depending on whether the assumptions $\neg B[\vec{p}_{\vec{x}_0}(t)/\vec{x}]$ and $\forall t\in[0,T).\, B[\vec{p}_{\vec{x}_0}(t)/\vec{x}]$ hold or not.

If both assumptions hold, then we get $Q(\vec{p}_{\vec{x}_0}(d),\textit{tr}^\chop \langle d, \vec{p}_{\vec{x}_0}, \emptyset\rangle)$ holds.  From the assumption that $Q$ is open, we can take $\epsilon_1>0$ such that $Q(s',\textit{tr}')$ for all $(s',\textit{tr}')\in \mathcal{U}_{\epsilon_1}(\vec{p}_{\vec{x}_0}(d), \textit{tr}^\chop \langle d,\vec{p}_{\vec{x}_0},\emptyset\rangle)$.

Take $\epsilon_0=\frac{\epsilon_1}{2}$. Then, by Lemma~\ref{lem:globalerror}, there exists $h_0>0$ such that for all $0<h<h_0$, the distance between $\vec{p}_{\vec{x}_0}(t)$ and $\vec{f}_{\vec{x}_0,h}(t)$ is bounded above by $\epsilon_0$ along the interval $t\in[0,T]$. With this choice of $\epsilon_0$ and $h_0$ in the conclusion of (DP), we get for any $\epsilon<\epsilon_0$ and $h<h_0$, the distance between $(\vec{f}_{\vec{x}_0,h}(T),\textit{tr}^\chop\langle T,\vec{f}_{\vec{x}_0,h},\emptyset\rangle)$ and $(\vec{p}_{\vec{x}_0}(T), \textit{tr}^\chop\langle d,\vec{p}_{\vec{x}_0},\emptyset\rangle)$ is at most $\epsilon_0$. Then, any pair within distance $\epsilon$ of $(\vec{f}_{\vec{x}_0,h}(T),\textit{tr}^\chop\langle T,\vec{f}_{\vec{x}_0,h},\emptyset\rangle)$ is within distance $\epsilon+\epsilon_0<\epsilon_1$ of $(\vec{p}_{\vec{x}_0}(T), \textit{tr}^\chop\langle d,\vec{p}_{\vec{x}_0},\emptyset\rangle)$, and hence satisfy $Q$. This shows
\[(\vec{f}_{\vec{x}_0,h}(T),\textit{tr}^\chop\langle T,\vec{f}_{\vec{x}_0,h},\emptyset\rangle)\in\mathcal{U}_{-\epsilon}(Q)\] as desired.

Now, we consider the case where $\forall t\in[0,T).\, B[\vec{p}_{\vec{x}_0}(t)/\vec{x}]$ does not hold. Intuitively, this corresponds to the case where $T$ is greater than the time length of execution. Choose $t\in[0,T)$ such that $\neg B[\vec{p}_{\vec{x}_0}(t)/\vec{x}]$. We claim that the first assumption in (DP) fails for this value of $t$. That is,
\[ 
\begin{array}{l}
	\neg (\exists \epsilon_0>0.\, \forall 0<\epsilon<\epsilon_0.\, \exists h_0>0.\, \forall 0<h<h_0.\, \vec{f}_{\vec{x}_0,h}(t)\in \mathcal{U}_{-\epsilon}(B))
\end{array} \]
Suppose otherwise, then take $\epsilon_0$, $\epsilon<\epsilon_0$ and $h_0$ so that for all $0<h<h_0$ the condition $\vec{f}_{\vec{x}_0,h}(t)\in\mathcal{U}_{-\epsilon}(B)$ holds. Take $h$ sufficiently small such that the difference between $\vec{f}_{\vec{x}_0,h}(t)$ and $\vec{p}_{\vec{x}_0}(t)$ is bounded above by $\epsilon$ for all $t\in[0,T)$. Then $\vec{p}_{\vec{x}_0}(t)\in \mathcal{U}_{\epsilon}(\vec{f}_{\vec{x}_0,h}(t))$, so that $\neg B[\vec{p}_{\vec{x}_0}(t)/\vec{x}]$ and $\vec{f}_{\vec{x}_0,h}(t)\in\mathcal{U}_{-\epsilon}(B)$ together gives a deadlock.

Finally, we consider the case where $\neg B[\vec{p}_{\vec{x}_0}(t)/\vec{x}]$ does not hold, in other words $\vec{p}_{\vec{x}_0}(T)$ satisfies $B$. Intuitively, this corresponds to the case where $T$ is less than the time length of execution. We claim that for sufficiently small $\epsilon$ and $h$, the condition $\vec{f}_{\vec{x}_0,h}(T)\in \neg\mathcal{U}_{-\epsilon}(B)$ is false, that is $\vec{f}_{\vec{x}_0,h}(T)\in\mathcal{U}_{-\epsilon}(B)$, hence the implication is vacuously true. Since $\vec{p}_{\vec{x}_0}(T)$ satisfies $B$ and $B$ is open, we can take $\epsilon_1>0$ so that $\mathcal{U}_{\epsilon_1}(\vec{p}_{\vec{x}_0}(T))\subseteq B$. Take $\epsilon_0=\frac{\epsilon_1}{2}$, and take $h_0$ so that for any $0<h<h_0$, we have the distance between $\vec{p}_{\vec{x}_0}(T)$ and $\vec{f}_{\vec{x}_0,h}(T)$ is less than $\epsilon_0$. With this choice of $\epsilon_0$ and $h_0$, for any $\epsilon<\epsilon_0$ and $h<h_0$, we know that any state within $\epsilon$ of $\vec{f}_{\vec{x}_0,h}(T)$ is within $\epsilon+\epsilon_0<\epsilon_1$ of $\vec{p}_{\vec{x}_0}(T)$, and hence satisfy $B$. This shows $\vec{f}_{\vec{x}_0,h}(T)\in\mathcal{U}_{-\epsilon}(B)$ as required.

We have now examined all three cases of $T>0$, and so have derived (DP) from (CP).

(DP) $\rightarrow$ (CP): Assume (DP) is true in state $s$ and trace $\textit{tr}$. We need to show that (CP) holds in state $s$ and trace $\textit{tr}$. That is, given $d>0$ and a solution $\vec{p}_{\vec{x}_0}$ satisfying $\forall t\in[0,d).\, B[\vec{p}_{\vec{x}_0}(t)/\vec{x}]$ and $\neg B[\vec{p}_{\vec{x}_0}(d)/\vec{x}]$, we need to prove $Q(\vec{p}_{\vec{x}_0}(d), \textit{tr}^\chop \langle d, \vec{p}_{\vec{x}_0}, \emptyset\rangle)$.

From (DP), take $T=d$. First, we show that the first assumption in (DP) holds, that is:
\[ 
\begin{array}{l}
	\forall 0\le t<T.\, \exists \epsilon_0>0.\, \forall 0<\epsilon<\epsilon_0.\, \exists h_0>0.\, \\
	\qquad\forall 0<h<h_0.\, \vec{f}_{\vec{x}_0,h}(t)\in \mathcal{U}_{-\epsilon}(B).
\end{array} \]
Fix some $t\in[0,T)$. Then we have $B[\vec{p}_{\vec{x}_0}(t)/\vec{x}]$ holds, and since $B$ is open, we can take $\epsilon_1$ such that $\mathcal{U}_{\epsilon_1}(\vec{p}_{\vec{x}_0}(t))\subseteq B$. Then take $\epsilon_0=\frac{\epsilon_1}{2}$, and take $h_0$ such that for all $0<h<h_0$, the distance between $\vec{p}_{\vec{x}_0}(t)$ and $\vec{f}_{\vec{x}_0,h}(t)$ is bounded by $\epsilon_0$. Then for any $\epsilon<\epsilon_0$, we know that any state within $\epsilon$ of $\vec{f}_{\vec{x}_0,h}(t)$ is within $\epsilon+\epsilon_0<\epsilon_1$ of $\vec{p}_{\vec{x}_0}(t)$, and hence in $B$. This shows $\vec{f}_{\vec{x}_0,h}(t)\in\mathcal{U}_{-\epsilon}(B)$, as desired. This proves the first assumption in (DP) holds.

Next, we show that for any sufficiently small $\epsilon$, there exists $h_0$ such that for any $0<h<h_0$, the condition $\vec{f}_{\vec{x}_0,h}(T)\in\neg\mathcal{U}_{-\epsilon}(B)$ holds. Given $\epsilon>0$, take $h_0$ such that for any $0<h<h_0$, the distance between $\vec{p}_{\vec{x}_0}(T)$ and $\vec{f}_{\vec{x}_0,h}(T)$ is bounded by $\epsilon$. But from $\neg B[\vec{p}_{\vec{x}_0}(t)/\vec{x}]$ this implies $\vec{f}_{\vec{x}_0,h}(T)\in\neg\mathcal{U}_{-\epsilon}(B)$ as desired.

From this, we have shown that for any sufficiently small $\epsilon$, there exists $h_0>0$ such that for any $0<h<h_0$, the condition $(\vec{f}_{\vec{x}_0,h}(T),$ $\textit{tr}^\chop \langle T, \vec{f}_{\vec{x}_0,h}, \emptyset \rangle) \in \mathcal{U}_{-\epsilon}(Q)$ holds. Take such $\epsilon_0$ and $\epsilon<\epsilon_0$. Then take $h_1$ such that the distance between $\vec{p}_{\vec{x}_0}(t)$ and $\vec{f}_{\vec{x}_0,h}(t)$ is bounded above by $\epsilon$ along the interval $t\in[0,T]$. Then for $h<\min(h_0,h_1)$, we get that $(\vec{p}_{\vec{x}_0}(T), \textit{tr}^\chop\langle T,\vec{f}_{\vec{x}_0,h},\emptyset\rangle)$ is within distance $\epsilon$ of $(\vec{f}_{\vec{x}_0,h}(T), $ $\textit{tr}^\chop \langle T, \vec{f}_{\vec{x}_0,h}, \emptyset \rangle)$, and the latter belongs to $\mathcal{U}_{-\epsilon}(Q)$. This implies $(\vec{p}_{\vec{x}_0}(T), \textit{tr}^\chop\langle T,\vec{f}_{\vec{x}_0,h},\emptyset\rangle)$ satisfies $Q$, as desired. This finishes the proof of (CP) from (DP). 
\end{proof}

\subsection{Discretization Rule for Continuous Interrupt}
An approximation to the rule 
 for continuous interrupt can be presented similarly. There are three occurrences  of solutions to ODEs in the rule. The third occurrence deals with the case where the ODE exits without performing a communication, which 
has the same form as the rule for continuous evolution, so it can be handled in the same way. The other two occurrences deal with
the cases where the ODE is interrupted by an output and input communication, respectively. They can be  handled similarly, so we only show the output case. The corresponding assertion in the weakest precondition is
{  \[ 
\begin{array}{l}
	\forall d>0.\, (\forall t\in[0,d).\, B[\vec{p}_{\vec{x}_0}(t)/\vec{x}]) \to 
	\ Q[\vec{p}_{\vec{x}_0}(d)/\vec{x},  \\
	 \quad \quad 
	 \tracev^\chop \langle d,\vec{p}_{\vec{x}_0},\{\cup_{i\in I} ch_i*\}\rangle^\chop \langle ch!,e[\vec{p}_{\vec{x}_0}(d)/\vec{x}] \rangle/\tracev]
\end{array} \tag{CI}
\] }
The discrete version of the assertion is:
{  \[ 
\begin{array}{l} 
	\forall T>0.\, (\forall 0\le t<T.\, \exists \epsilon_0>0.\, \forall 0<\epsilon<\epsilon_0.\,  \\
	 \qquad 
	\exists h_0>0.\, \forall 0<h<h_0.\, \vec{f}_{\vec{x}_0,h}(t)\in\mathcal{U}_{-\epsilon}(B) ) \to \\
 \exists \epsilon_0>0.\, \forall 0<\epsilon<\epsilon_0.\,
	\exists h_0>0.\, \forall 0<h<h_0.\,  
  (\vec{f}_{\vec{x}_0,h}(T), 
   \\   
 \hspace*{-6mm} \tracev^\chop \langle T,\vec{f}_{\vec{x}_0,h},\{\cup_{i\in I} ch_i*\}\rangle^\chop \langle ch!,e[\vec{f}_{\vec{x}_0,h}(T)/\vec{x}] \rangle)  
  \in \mathcal{U}_{-\epsilon}(Q) \tag{DI} 
\end{array} \] }

The following theorem is proved.
\begin{theorem} \label{thm:CI-DI} 
\sm{\models \text{CI} \leftrightarrow \text{DI}}.
\end{theorem}

\begin{proof}[Proof for Theorem~\ref{thm:CI-DI}]
 (CI) $\rightarrow$ (DI): Assume (CI) is true in state $s$ and trace $\textit{tr}$. Fix $T>0$ in (DI), then there are two cases, depending on whether $\forall t\in[0,T).\, B[\vec{p}_{\vec{x}_0}(t)/\vec{x}]$ holds or not.

If it holds, then from (CI), with $d=T$, we get

{ 
\[ Q(\vec{p}_{\vec{x}_0}(d), \textit{tr}^\chop \langle d,\vec{p}_{\vec{x}_0},\{\cup_{i\in I} ch_i*\}\rangle^\chop \langle ch!,s[\vec{x} \mapsto \vec{p}_{\vec{x}_0}(d)](e)\rangle)\]}

\noindent holds. From the assumption that $Q$ is open, we take $\epsilon_1>0$ such that $Q(s',\textit{tr}')$ for all
$(s',tr')$ in the set
{ 
\[ \mathcal{U}_{\epsilon_1}(\vec{p}_{\vec{x}_0}(d), \textit{tr}^\chop\langle d,\vec{p}_{\vec{x}_0},\{\cup_{i\in I} ch_i*\}\rangle^\chop \langle ch!,s[\vec{x} \mapsto \vec{p}_{\vec{x}_0}(d)](e)\rangle).\]}

\noindent Since $e$ is expressed in terms of arithmetic operations, it is a continuous function of its argument. Moreover, since the path $\vec{p}_{\vec{x}_0}$ on the interval $[0,T]$ is compact, we get that $e$ is uniformly continuous on a closed neighborhood of the path. Hence, we can take $\epsilon_0>0$ such that $\epsilon_0<\frac{\epsilon_1}{2}$, and $\|\vec{y}_0-\vec{y}_1\|_\infty<\epsilon_0$ implies $|e(\vec{y_0})-e(\vec{y_1})|<\frac{\epsilon_1}{2}$ on the $\epsilon_0$-neighborhood of the path $\vec{p}_{\vec{x}_0}([0,T])$.

For this choice of $\epsilon_0$, there exists $h_0>0$ such that for all $0<h<h_0$, the distance between $\vec{p}_{\vec{x}_0}(t)$ and $\vec{f}_{\vec{x}_0,h}(t)$ is bounded above by $\epsilon_0$ along the interval $t\in[0,T]$. With this choice of $\epsilon_0$ and $h_0$ in the conclusion of (DI), we get for any $\epsilon<\epsilon_0$ and $h<h_0$, the distance between
{ 
\[ (\vec{f}_{\vec{x}_0,h}(T), \textit{tr}^\chop\langle T,\vec{f}_{\vec{x}_0,h},\{\cup_{i\in I} ch_i*\}\rangle^\chop \langle ch!,s[\vec{x} \mapsto \vec{f}_{\vec{x}_0,h}(T)](e)\rangle) \]}
\noindent and
{ 
\[ (\vec{p}_{\vec{x}_0}(T), \textit{tr}^\chop\langle T,\vec{p}_{\vec{x}_0},\{\cup_{i\in I} ch_i*\}\rangle^\chop \langle ch!,s[\vec{x} \mapsto \vec{p}_{\vec{x}_0}(T)](e)\rangle) \]}

\noindent is at most $\frac{\epsilon_1}{2}$. This shows
{ 
\[ 
\begin{array}{ll}
     (\vec{f}_{\vec{x}_0,h}(T), \textit{tr}^\chop\langle T,\vec{f}_{\vec{x}_0,h},\{\cup_{i\in I} ch_i*\}\rangle^\chop \langle ch!,s[\vec{x} \mapsto \vec{f}_{\vec{x}_0,h}(T)](e)\rangle)
     \in \mathcal{U}_{-\epsilon}(Q) 
\end{array}\]
}
by the same argument as in the continuous evolution case.

Now suppose the condition $\forall t\in[0,T).\, B[\vec{p}_{\vec{x}_0}(t)/\vec{x}]$ does not hold. Intuitively, this corresponds to the case where $T$ is greater than the maximum possible length of execution of the interrupt command. Choose $t\in [0,T)$ such that $\neg B[\vec{p}_{\vec{x}_0}(t)/\vec{x}]$. Then the first assumption of (DI) fails for this value of $t$, as shown in the proof of the continuous evolution case. This means (DI) holds vacuously.

We have now considered both cases of $T>0$, and so have derived (DI) from (CI).

(DI) $\rightarrow$ (CI): Assume (DI) is true in state $s$ and trace $\textit{tr}$. We need to show that (CI) also holds. That is, given $d>0$ and a solution $\vec{p}_{\vec{x}_0}$ satisfying $\forall t\in[0,d).\, B[\vec{p}_{\vec{x}_0}(t)/\vec{x}]$, we need to prove 
\sm{Q(\vec{p}_{\vec{x}_0}(d), \textit{tr}^\chop \langle d,\vec{p}_{\vec{x}_0},\{\cup_{i\in I} ch_i*\}\rangle^\chop \langle ch!,s[\vec{x} \mapsto \vec{p}_{\vec{x}_0}(d)](e)\rangle)}.

From (DI), take $T=d$. The assumption in (DI) holds by the same argument as in the continuous evolution case. Therefore, for any sufficiently small $\epsilon$, there exists $h_0>0$ such that for any $0<h<h_0$, the condition

{ 
\[ \begin{array}{l}
(\vec{f}_{\vec{x}_0,h}(T), \textit{tr}^\chop\langle T,\vec{f}_{\vec{x}_0,h},\{\cup_{i\in I} ch_i*\}\rangle^\chop \langle ch!,s[\vec{x} \mapsto \vec{f}_{\vec{x}_0,h}(T)](e)\rangle) \in \mathcal{U}_{-\epsilon}(Q) 
\end{array} \]
}

\noindent holds. Take such $\epsilon_0$ and $\epsilon<\epsilon_0$. Next, take $\epsilon'$ sufficiently small so that $\|\vec{y_0}-\vec{y_1}\|_\infty < \epsilon'$ implies $|e(\vec{y_0})-e(\vec{y_1})|<\epsilon$ on the $\epsilon'$-neighborhood of the path $\vec{p}_x([0,T])$, and take $h_1$ such that the distance between $\vec{p}_{\vec{x}_0}(t)$ and $\vec{f}_{\vec{x}_0,h}(t)$ is bounded above by $\epsilon'$ along the interval $t\in[0,T]$. Then for $h<\min(h_0,h_1)$, we get that \sm{(\vec{p}_{\vec{x}_0}(d), \textit{tr}^\chop \langle d,\vec{p}_{\vec{x}_0},\{\cup_{i\in I} ch_i*\}\rangle^\chop \langle ch!,s[\vec{x} \mapsto \vec{p}_{\vec{x}_0}(d)](e)\rangle)} is within distance $\epsilon$ of 
$(\vec{f}_{\vec{x}_0,h}(T), \textit{tr}^\chop\langle T,\vec{f}_{\vec{x}_0,h}, \{\cup_{i\in I} ch_i*\}\rangle^\chop \langle ch!,$
$s[\vec{x} \mapsto \vec{f}_{\vec{x}_0,h}(T)](e) \rangle)$. Since the latter belongs to $\mathcal{U}_{-\epsilon}(Q)$, this implies the former satisfies $Q$, as desired. This finishes the proof of (CI) from (DI).   
\end{proof}

Now we present the discrete relative completeness theorem.
\begin{theorem}[Discrete Relative Completeness] \label{thm:dicreteRC}
	The proof system  presented in Sect.~\ref{sec:hoare}, plus \sm{\text{CP} \leftrightarrow \text{DP}} and \sm{\text{CI} \leftrightarrow \text{DI}}, are complete relative to the discrete fragment, without referring to solutions of differential equations. 	
\end{theorem} 

\begin{proof}[Proof for Theorem~\ref{thm:dicreteRC}]
 The hybrid Hoare logic presented in Sect.\ref{sec:hoare} inherits the continuous completeness relative to the first-order theory of differential equations (i.e. FOD). All
that remains to be shown is that we can then prove all those
valid FOD formulas from valid formulas of discrete fragment plus the added formulas $CP \leftrightarrow DP$ and $CI \leftrightarrow DI$. The only question that remains to consider is that,  the restriction that we put when proving $CP \leftrightarrow DP$ and $CI \leftrightarrow DI$: the predicates occurring in the precondition of ODEs (i.e. $Q_1$ and $Q_2$ mentioned in the above proofs) are open, should be removed.

Without loss of generality, we assume all predicates are  first-order formulas of real arithmetic, \footnote{If a formula contains analytic terms, according to the theory of real analysis, these  analytic terms can be approximated by polynomials with respect to arbitrary precision \cite{Krantz2002}.} which can be reduced to the equivalent formula of the following form (denoted by $S$):
\[\bigwedge^m_{i=1}(\bigvee_{j=1}^{l} p_{i,j} > 0\vee \bigvee_{k=1}^{n} q_{i, k} \geq 0)\]

Clearly, $\vee_{j=1}^{l} p_{i,j} > 0$  corresponds to an open basic semi-algebraic set, say $O_i$,  and $\vee_{k=1}^{n} q_{i, k} \geq 0$  corresponds to a closed basic semi-algebraic set, say $C_i$. Being a closed set, $C_i$
is equivalent to $\forall \epsilon_i>0. \, \mathcal{U}_{\epsilon_i}(C_i)$
so $S$ can be reformulated by
\[\bigwedge_{i=1}^{m} \forall \epsilon_i>0.(D_i \vee  \mathcal{U}_{\epsilon_i}(C_i))\] 

For each $D_i \vee  \mathcal{U}_{\epsilon_i}(C_i)$, denoted by $\Phi_i$, it is an open set, or say open predicate. Using the two formulas $CP \leftrightarrow DP$ and $CI \leftrightarrow DI$, 
$\Phi_i$ can be equivalently represented as a discrete formula, denoted by $\mathcal{D}(\Phi_i)$. By $m$ uses of the two formulas, $S$ can be represented by an open predicate, and thus can be derived with respect to the discrete fragment of the logic.  
\end{proof} 

In the proof of Theorem~\ref{thm:dicreteRC}, we exploit the fact that any set can be represented as the intersection of an open set and a closed set, while a closed set can be represented as the intersection of a sequence of (possibly infinitely many) open sets.

Using the same method as for proving the continuous relative completeness of the full HHL, we can show that the proof system without rules of the continuous operations is relatively complete for discrete HCSP. The statement of the theorem is as follows.
\begin{theorem}[Relative Completeness of the Discrete Fragment]\label{thm:relativeRC-DF}
The counterpart of the proof system corresponding to the discrete HCSP, i.e., without rules $\textrm{SyncUnpairE3}$, $\textrm{SyncUnpairE4}$, $\textrm{SyncWaitE1-3}$, $\textrm{Cont}$ and $\textrm{Int}$ is relatively complete in Cook's sense. 
\end{theorem}

\section{Implementation and Case Studies}\label{sec:examples}

\paragraph{Implementation} The HHL logic is implemented in Isabelle/HOL. 
The whole implementation contains the syntax,
semantics and inference rules of HHL.
Especially, it formalizes the inference rules and allows to
perform proofs in HHL within Isabelle. 
In the formalization, we used shallow embedding to represent the ODEs 
of the logic and employed the corresponding library on ODE in Isabelle for defining semantics
and inference rules of ODE. Both kinds of inference rules of continuous evolution based on explicit ODE solutions and differential invariants are proved to be valid.  
The soundness of HHL is proved in  Isabelle, to make sure that the logic is correct and
thus it can be applied for  verification of hybrid systems.

Next, to show the use of HHL,  we apply it to verify two case studies: lunar lander control system, and a scheduler that controls the execution of parallel tasks. The former one involves ODE dynamics for modelling continuous plants, and is proved based on a differential invariant; while the latter one involves communications, interrupts, and complex control behavior with many if-else branches, and is proved with existence of many different execution cases and their parallel composition. 

Before the demonstration, we first introduce the following abbreviations used in our proofs:
\[ 
\begin{array}{rll}
	\textsf{in}(\vec{x}_0,ch,v,\rdy) &\triangleq& \tracev = \langle ch?,v \rangle \vee 
	 \exists d>0.\, \tracev = \langle d,I_{\vec{x}_0},\rdy \rangle \\
	 && ^\chop \langle ch?,v\rangle\vee   \tracev = \langle\infty,I_{\vec{x}_0},\rdy\rangle \\

	\textsf{out}(\vec{x}_0,ch,v,\rdy) &\triangleq& \tracev = \langle ch!,v \rangle \vee 
	 \exists d>0.\, \tracev = \langle d,I_{\vec{x}_0},\rdy \rangle \\
	 &&^\chop \langle ch!,v\rangle\vee \tracev = \langle\infty,I_{\vec{x}_0},\rdy\rangle \\
  \textsf{trout}(d,\vec{p},\textit{ch},v,\textit{rdy}) 
&\triangleq &\tracev = \langle d,\vec{p},\rdy\rangle^{\chop} \langle ch!,v \rangle\\
\textsf{IO}(ch,v) &\triangleq& \tracev = \langle ch,v\rangle\\
 \textsf{traj}(d,\vec{p}_{\vec{x}_0},\rdy) &\triangleq& \tracev = \langle d,\vec{p}_{\vec{x}_0},\rdy\rangle\\
 \textsf{emp} & \triangleq & \tracev = \epsilon \\
  P \joinop Q &\triangleq& \exists tr_1,\,tr_2.\, P[tr_1/\tracev]\wedge Q[tr_2/\tracev]\wedge {tr_1}^\chop tr_2 = \tracev
\end{array}
\]
Without losing information, we abbreviate \sm{\textsf{out}(\vec{x}_0,ch,v,\{ch!\})} as \sm{\textsf{out}(\vec{x}_0,ch,v)} and  \sm{\textsf{in}(\vec{x}_0,ch,v,\{ch?\})} as \sm{\textsf{in}(\vec{x}_0,ch,v)}.
\paragraph{Case 1: Lunar lander}
First,  we demonstrate how the HHL prover is used to verify a simplified lunar lander example adapted from 
\cite{Zhao14}.  
In this example,  a lander descends to the surface under the descent guidance control with the goal to maintain a stable downward velocity of the lander. The continuous evolution is defined by:
$ \dot{v} =  \frac{F_{c}}{m}-g_{M},
     \dot{m} =  -\frac{F_{c}}{I}
 $, where $v$ represents the velocity of descending, $m$ is the mass of the lander, $F_{c}$ is the thrust imposed on the lander; $g_{M}$ and $I$ are constants of gravity acceleration and mass loss rate. 
The thrust is updated according to:
{  \[
F_{c}' = m\cdot (g_{M}-c_{1}(v-v_{s})-c_{2}(\frac{F_c}{m}-g_{M}))
\] }
where $v_{s}$ is the target velocity we want to maintain, i.e. -1.5m/s here. Since the origin ODE is non-polynomial, we replace $\frac{F_c}{m}$ by a new variable $w$. 
After the substitution,  
the whole process is modelled as a parallel composition of $plant$ and $ctrl$, where
{  
\[
\begin{array}{ll}
\textit{plant} \triangleq &(\langle\dot{v}=w-3.732, \dot{w}=\frac{w^{2}}{2500},\dot{t}=1 \&\mathsf{true} \rangle \unrhd \talloblong\\
&\quad({\textit{chv}!v}\rightarrow {\textit{chw}!w; chc?w;t:=0}))^* 
\end{array}
\] }
\noindent The process $plant$ models the continuous behavior of the lander, for which it will be 
 interrupted by a communication, sends the   value of $v$ and $w$ to the controller in sequence, and then receives a new $w$ from the controller and resets time $t$. The control part is given by: 
{ 
\[
\textit{ctrl} \triangleq (\pwait \ T; \textit{chv}?v; \textit{chw}?w; \textit{chc}!(f(v,w)))^*
\] }
\noindent The controller updates the thrust every period $T=0.128s$. After receiving the current velocity and thrust from plant, it updates the thrust according to  $f(v,w) \triangleq -(w - 3.732) * 0.01 + 3.732 - (v + 1.5) * 0.6$. 

We define the trace invariant of \textsf{plant} given the initial value of $v$, $w$, and a list of $w$ inputs, and then prove the following Hoare triple:
{ 
\begin{align*} 
\textsf{plant\_block}(v_0,w_0,[]) = & \textsf{emp} \\
\textsf{plant\_block}(v_0,w_0,(d_1,w_1)\cdot ws') =  & \exists\ \vec{p}_{[v_0,w_0,0],\vec{e}}.\,
\textsf{trout}(d_1,\vec{p}_{[v_0,w_0,0],\vec{e}},\textit{chv},\vec{p}_{[v_0,w_0,0],\vec{e}}(d_1)[v],\{\textit{chv!}\})\joinop \\
& \hspace*{-8mm}  \textsf{out}(\vec{p}_{[v_0,w_0,0],\vec{e}}(d_1),\textit{chw},\vec{p}_{[v_0,w_0,0],\vec{e}}(d_1)[w]) \joinop
\textsf{in}(\vec{p}_{[v_0,w_0,0],\vec{e}}(d_1),\textit{chc}, w_1) \joinop  \\
& \hspace*{-8mm}\textsf{plant\_block}(\vec{p}_{[v_0,w_0,0],\vec{e}}(d_1)[v],w_1,ws') \\
& \hspace*{-2.5cm} 	\{  v = v_0 \wedge w = w_0 \wedge t = 0 \wedge \textsf{emp} \}  
  \textit{plant}  
	\{   \exists ws.\, \textsf{plant\_block}(v_0,w_0,ws) \}  
\end{align*} }
where $\vec{e}$ represents the dynamics in $\mathsf{plant}$ and $\vec{p}_{[v_0,w_0,0],\vec{e}}$ is the solution of this $ode$ starting with the initial vector $[v_0,w_0,0]$.

Similar for \textsf{ctrl}, it follows 
{   
\begin{align*} 
\textsf{ctrl\_block}(v_0,w_0,[]) = &\textsf{emp} \\
\textsf{ctrl\_block}(v_0,w_0,(v_1,w_1)\cdot \textit{vws}') =  & \textsf{traj}(T,I_{[v_0,w_0]},\emptyset)\joinop \textsf{in}([v_0,w_0],\textit{chv},v_1) \joinop\\
&\hspace*{-8mm}
\textsf{in}([v_1,w_0],\textit{chw}, w_1) \joinop
\textsf{out}([v_1,w_1],\textit{chc}, f(v_1,w_1)) \joinop  \textsf{ctrl\_block}(v_1,w_1,\textit{vws}') \\ 
&\hspace*{-2.5cm}	\{   v = v_0 \wedge w = w_0 \wedge \textsf{emp} \}  \textit{ctrl}  
	\{   \exists vws.\, \textsf{ctrl\_block}(v_0,w_0,\textit{vws}) \}
\end{align*}
}
\noindent Synchronizing the two traces, it derives 
a system trace invariant of the whole system as 
{ 
\begin{align*}
\textsf{system}(v_0,w_0,v_0',w_0',0) = &\textsf{emp} \\
\textsf{system}(v_0,w_0,v_0',w_0',n+1) = 
&\exists \vec{p}_{[v_0,w_0,0],\vec{e}} . \, \textsf{traj}(T, \vec{p}_{[v_0,w_0,0],\vec{e}}\uplus I_{[v_0',w_0']},\{chv!\}) \joinop \\
&\hspace*{-8mm}
\textsf{IO}(\textit{chv},\vec{p}_{[v_0,w_0,0],\vec{e}}(T)[v]) \joinop
\textsf{IO}(\textit{chw},\vec{p}_{[v_0,w_0,0],\vec{e}}(T)[w]) \joinop
\\
&\hspace*{-8mm}
\textsf{IO}(\textit{chc},f(\vec{p}_{[v_0,w_0,0],\vec{e}}(T)[v],\vec{p}_{[v_0,w_0,0],\vec{e}}(T)[w])) \joinop \\
&\hspace*{-8mm}
\textsf{system}(\vec{p}_{[v_0,w_0,0],\vec{e}}(T)[v],f(\vec{p}_{[v_0,w_0,0],\vec{e}}(T)[v],\vec{p}_{[v_0,w_0,0],\vec{e}}(T)[w]),\\
&\hspace*{-8mm}
\vec{p}_{[v_0,w_0,0],\vec{e}}(T)[v],\vec{p}_{[v_0,w_0,0],\vec{e}}(T)[w],n)
\end{align*}
}
 However, the solution of the ODE, i.e. $\vec{p}_{[v_0,w_0,0],\vec{e}}$ occurring in \textsf{system}, does not have an explicit definition. Therefore, to prove the final goal of this example, i.e.,  the velocity of the lander keeps within a safe range $(-1.45,$ $-1.55)$,  we use the differential invariant rule (Dbarrier) proposed in Section~\ref{subsec:dinvariant} instead. Assume $Inv(t, v, w)\triangleq q(t,v,w)\le 0$ is an polynomial invariant on variables $t, v, w$, then we have:
{  \[
\begin{array}{ll}
 &
 0\le t \le T \rightarrow \textit{q} = 0 \rightarrow \mathcal{L}_{\vec{e}}(\textit{q})<0\\
  \wedge & \textit{q}(0,-1.5,5670/1519)\le0 \\
 \wedge & 
 \textit{q}(T,v,w)\le 0 \rightarrow \textit{q}(0,v,f(v,w))\le 0 \\
  \wedge &
 (0\le t \le T \wedge\textit{q}\le 0) \rightarrow (v+1.45)*(v+1.55)<0 
\end{array}
\] }
where the formulas of four lines represent: 
the hypothesis of  Rule (Dbarrier) to be the invariant of the ODE; 
the invariant holds for the initial values;
the invariant is preserved by the discrete update on the time and acceleration of each round; 
the invariant is strong enough to imply the final goal. These four constraints constitute a sufficient condition for $Inv(t, v, w)$ to be a global invariant of the whole system~\cite{PlatzerClarke08}.
With these four constraints, 
we invoke the ODE invariant generation tool~\cite{Zhao14} and obtain a differential invariant  $Inv(t, v, w)$.  
We finally prove the following specification for the whole  system:
{  \begin{align*}
	\{ & (v = -1.5 \wedge w = 5670/1519 \wedge t = 0) \uplus ( v = v1 \wedge w = w1)\wedge \textsf{emp} \} \\
	& \quad \quad  \mathit{plant}\,\|_{\{\textit{chv,chw,chc}\}}\,\mathit{ctrl} \\
 &\{  \exists n.\, \textsf{system}(-1.5,5670/1519,v1,w1,n) \}
\end{align*}}
where the invariant $\textsf{system}$ is  strengthened with $Inv(t, v, w)$, thus it  obviously   guarantees the goal   $-1.55\le v\le -1.45$.

Compared to the proof in \cite{Zhao14} based on Duration Calculus, the above proof for  
the parallel composition of  $\mathit{plant}$ and $\mathit{ctrl}$ 
 is compositional, and furthermore more rigorous, as the derived differential invariant rules are all formalised and proved to be valid in the Isabelle implementation of current HHL prover, while proved manually for~\cite{Zhao14}.  

\paragraph{Case 2: Scheduler with two tasks}
Next, we apply HHL to verify a scheduler process controlling two task modules executed in parallel. 
We focus on the interactions of the modules with the scheduler, and prove the correctness of the scheduler,  i.e.,  at any time, the running module is always with the highest priority among the modules that are in  ready state. 

We consider the case where two \textit{module} processes are in parallel with a \textit{scheduler}. In \textit{module}, \textit{state} stands for the state of a module ranging over $\{\textit{WAIT}, \textit{\textit{\textit{READY}}}, \textit{RUN}\}$, $\textit{prior}$ for the module priority,  $\textit{period}$ for the period, and  
$\textit{cost}$ for the module execution time, $T$ for the system 
time, and $\textit{ent}$ indicates  
whether the module starts to  execute
in this period.
{  \[  
\begin{array}{ll}
     &\mathit{module} := (\\
     &\textrm{if}\ \textit{state}=\textit{WAIT} \ \textrm{then}\   \langle \dot{T}=1\& T< \textit{period}\rangle;T:=0;
     \textit{ent}:=0;\textit{state}:= \textit{\textit{READY}}\\
     &\textrm{else if}\ \textit{state}=\textit{\textit{READY}} \  \textrm{then} \  \textit{\textit{req\_ch}}!\textit{prior};\\
     & \qquad\langle \dot{T}=1\& T< \textit{period}\rangle \unrhd \talloblong(\textit{\textit{run\_ch}}?x \rightarrow \textit{state} := \textit{\textit{RUN}});\\
     & \qquad \textrm{if}\ \textit{state} \neq \textit{\textit{READY}}\ \textrm{then}\ \pskip\ \\
     & \qquad \textrm{else} (\textit{\textit{run\_ch}}?x;\textit{state} := \textit{RUN}\sqcup \textit{\textit{exit\_ch}}!0;\textit{state}:=\textit{WAIT})\\
     &\textrm{else}\ 
     (\IFE{\textit{ent}=0}{C:=0;\textit{ent}:=1}{\pskip});\\
     & \qquad\langle \dot{T}=1,\dot{C}=1\& T< \textit{period}\land C< \textit{cost} \rangle \unrhd \talloblong 
     (\textit{\textit{pr\_ch}}?x \rightarrow \textit{state} := \textit{\textit{READY}});\\
     & \qquad \textrm{if}\ \textit{state} \neq \textit{RUN}\ \textrm{then}\ \pskip\ \\
     & \qquad \textrm{else}(\textit{\textit{pr\_ch}}?x;\textit{state} := \textit{WAIT}\sqcup \textit{\textit{free\_ch}}!0;\textit{state}:=\textit{WAIT}))^*
\end{array}
\] }
\noindent In each round, depending on the value of $\textit{state}$, the following execution occurs: if it is  $\textit{WAIT}$, the module waits for $\textit{period}$ time, then resets $T$ and $\textit{ent}$, turns to \textit{READY}; if it is  $\textit{\textit{READY}}$, the module first sends a  request to the scheduler with its priority  $\textit{\textit{req\_ch}}!\textit{prior}$, then waits for a running command $\textit{\textit{run\_ch}}?x$ until the end of this period. Once a running command is received,  performs skip and turns to  $\textit{RUN}$; otherwise, checks again if the running command is enabled and turns to $\textit{RUN}$ if it can occur immediately, or sends an exit command and turns to $\textit{WAIT}$; 
if it is  $\textit{RUN}$, $C$ will be reset if this module is not executed. Then the module starts running and can be preempted at any time by $\textit{\textit{pr\_ch}}$ and turns to $\textit{READY}$
until $T$ reaches $\textit{period}$ or $C$ reaches $\textit{cost}$.
Once it turns to $\textit{READY}$, it skips; otherwise, checks again if the preemption command $\textit{\textit{pr\_ch}}$ is enabled and turns to $\textit{WAIT}$ if it can occur immediately, or sends a free command and turns to $\textit{WAIT}$. 

Below defines process $\mathit{scheduler}$ for the scheduler: 
{  
\[ \hspace*{-4mm}
\begin{array}{ll}
& \mathit{scheduler} := ((\textit{req\_ch}_i?p; \\
& \IFE{\textit{rp}\ge p}{L:=L\sharp(i,p)}{} (\IFE{\textit{ri}\neq-1}{pr\_ch_{ri}!0}{\pskip};\\
& \qquad \qquad \textit{run\_ch}_{i}!0;\textit{ri}:=i;\textit{rp}:=p)) \sqcup (\textit{free\_ch}_i?g;\\
& \qquad \textrm{if}\ \textrm{len}(L)>0\  \textrm{then}\ (\textit{ri},\textit{rp}):=\textrm{max}(L);L:=\textrm{del}(L,\textit{ri});\textit{run\_ch}_{\textit{ri}}!0\\
& \qquad \qquad \quad \quad  \quad  \textrm{else}\ (\textit{ri},\textit{rp}):=(-1,-1))\sqcup
 (\textit{exit\_ch}_i?g;L:=\textrm{del}(L,i)))^*
\end{array} \] }
The scheduler uses a list variable $L$ to record the list of modules waiting for a running command, and  $\textit{ri}$ and $\textit{rp}$ to record the index and priority of the running module (we use $-1$ to represent that no module is running).
 If the scheduler receives a priority $p$  from module $i$, it will compare $p$ with the priority ($\textit{rp}$) of the running module ($\textit{ri}$). If $\textit{rp}$ is greater than or equal to $p$, then the pair $(i,p)$ will be added into the list $L$, otherwise scheduler will send a preempt command to $\textit{ri}$ and a running command to $i$; if the scheduler receives a free command, it will remove the pair with the maximum priority from $L$ and set it as new $\textit{ri}$ and $\textit{rp}$, sending the running command; if the scheduler receives an exit command from module $i$, it will remove the pair with index $i$ from  $L$.

\oomit{
For the scheduler, we proved that it will produce traces:
{   \[  
\begin{array}{ll}
&\textsf{s}(0,L,ri,rp) = \textsf{emp} \\
&\textsf{s}(k+1,L,ri,rp) =\\
&\quad \exists p. \textsf{in}(req\_ch_i,p,(\emptyset,\{req\_ch_{i\in I},exit\_ch_{i\in I},free\_ch_{i\in I}\}))\\
&\quad\joinop(\textrm{if}\ rp\ge p\  \textrm{then}\ \textsf{s}(k,L\sharp(i,p),ri,rp) \ \textrm{else}\\
&\qquad \quad\textrm{if}\ rn=-1\ \textrm{then}\ \textsf{out}(run\_ch_{i},0)\joinop\textsf{s}(k,L,i,p)\ \textrm{else}\\
&\qquad \quad \quad\textsf{out}(pr\_ch_{ri},0)\joinop\textsf{out}(run\_ch_{i},0)\joinop\textsf{s}(k,L,i,p))\bigvee\\
&\quad \exists g. \textsf{in}(free\_ch_i,g,(\emptyset,\{req\_ch_{i\in I},exit\_ch_{i\in I},free\_ch_{i\in I}\}))\\
&\quad\joinop(\textrm{if}\ len(L)>0\  \textrm{then}\ \textsf{out}(run\_ch_{max(L)_i},0)\joinop\\
&\qquad\qquad\qquad\qquad\quad\textsf{s}(k,L,max(L)_i,max(L)_p) \\
&\qquad\qquad\qquad\qquad\textrm{else}\ \textsf{s}(k,[],-1,-1))\\
&\quad \exists g. \textsf{in}(free\_ch_i,g,(\emptyset,\{req\_ch_{i\in I},exit\_ch_{i\in I},free\_ch_{i\in I}\}))\\
&\quad\joinop \textsf{s}(k,del(L,i),ri,rp))
\end{array} \] }
}

We prove trace assertions of each module and scheduler independently. Then synchronize these assertions and 
obtain an assertion $f$ parameterized by states of the modules, that specifies the behavior of the whole process. Part of the synchronization  is shown below:
{  \[\hspace*{-5mm} 
\begin{array}{ll}
&f(\textit{\textit{READY}}_1,\textit{RUN}_2) \\
&\rightarrow \textsf{IO}(\textit{req\_ch}_1,p1)\joinop \textsf{IO}(\textit{pr\_ch}_2,0) 
\joinop\textsf{IO}(\textit{run\_ch}_1,0)\joinop 
f(\textit{RUN}_1,\textit{\textit{READY}}_2)\\
&\rightarrow (\textsf{IO}(\textit{req\_ch}_1,p1)\joinop 
\textsf{IO}(\textit{pr\_ch}_2,0)
\joinop\textsf{IO}(\textit{run\_ch}_1,0)\joinop 
f(\textit{RUN}_1,\textit{WAIT}_2))\\
&\qquad \vee (\textsf{IO}(\textit{free\_ch}_2,0)\joinop 
f(\textit{\textit{READY}}_1,\textit{WAIT}_2))\\
& \qquad\textrm{if}\ \min(\textit{period}_2-T_2,\textit{cost}_2-\textit{ent}_2*C_2)=0\\
&f(\textit{RUN}_1,\textit{\textit{READY}}_2) \\
&\rightarrow f(\textit{WAIT}_1,\textit{\textit{READY}}_2)
\qquad\textrm{if}\ \min(\textit{period}_1-T_1,\textit{cost}_1-\textit{ent}_1*C_1)=0\\
&\rightarrow \textsf{IO}(\textit{req\_ch}_2,p_2)\joinop 
\textsf{traj}(\min(\textit{period}_1-T_1,\textit{cost}_1-\textit{ent}_1*C_1),\vec{p},rdy)\\
&\qquad\joinop\textsf{IO}(\textit{free\_ch}_1,0)\joinop\textsf{IO}(run\_ch_2,0)\joinop
f(\textit{WAIT}_1,\textit{RUN}_2)\\
& \qquad\textrm{if}\ \min(\textit{period}_1-T_1,\textit{cost}_1-\textit{ent}_1*C_1)\le \textit{period}_2-T_2\\
&\rightarrow \textsf{IO}(\textit{req\_ch}_2,p_2)\joinop 
\textsf{traj}(\textit{period}_2-T_2,\vec{p},rdy)
\joinop\textsf{IO}(\textit{exit\_ch}_1,0)\joinop
f(\textit{RUN}_1,\textit{WAIT}_2)\\
& \qquad\textrm{if}\ \min(\textit{period}_1-T_1,\textit{cost}_1-\textit{ent}_1*C_1)\ge \textit{period}_2-T_2\\
&f(\textit{RUN}_1,\textit{RUN}_2)  \rightarrow \textsf{FALSE}
\end{array} \] }
where $\rdy = \{run\_ch_2?, pr\_ch_1?, req\_ch_1?, req\_ch_2?, free\_ch_1?, free\_ch_2?, exit\_ch_1?, exit\_ch_2?\}$.
The first case states that when $\textit{module1}$ is in $\textit{\textit{READY}}$  and $\textit{module2}$ is in $\textit{RUN}$, the operation of $\textit{module2}$ is immediately interrupted by $\textit{module1}$.
The second case states that when $\textit{module1}$ is in $\textit{RUN}$ and $\textit{module2}$ is in $\textit{\textit{READY}}$, $\textit{module2}$ cannot enter $\textit{RUN}$ until the former is finished. The last case states that two modules can never run at the same time, as is required for the scheduler.
The overall theorem is the following Hoare triple proved in Isabelle (where $A_i=\{\textit{req\_ch}_i,\textit{free\_ch}_i,\textit{exit\_ch}_i,\textit{run\_ch}_i,\textit{pr\_ch}_i\}$):
{  \begin{align*} 
	\{ & (L=[] \wedge \textit{ri} = -1 \wedge \textit{rp} = -1)\uplus \textit{state1} = \textit{WAIT} \uplus \textit{state2} = \textit{WAIT} \wedge \textsf{emp}\} 
	\\
	& \quad 
	\quad 
 (\mathit{scheduler} \|_{A_1} \mathit{module1}) \|_{A_2} \mathit{module2}
 \\
 & \{  f(\textit{WAIT}_1,\textit{WAIT}_2) \}
\end{align*} }
The postcondition $f(\textit{WAIT}_1,\textit{WAIT}_2)$ records the execution trace of the whole system, that starts from an initial state that both modules are in wait states. From the above transition rules held by assertion $f$, it satisfies the safety requirement of the system that  at any time the running module is always with the
highest priority and it is not allowed to have more than one modules in running state simultaneously. 

The proof needs to consider each combination of states produced during execution for the scheduler and modules (54 cases in total), due to the complex control logics. For instance, during synchronization when both processes are waiting, we need to consider three cases depending on the comparison of waiting time on the two sides.  Some global invariants need to be shown during the proof: the values of $T$ and $C$ in the module will not exceed $\textit{period}$ and $\textit{cost}$ respectively, $\textit{ri}$ indeed represents which module is running,
and so on\oomit{$\textit{prior}$ of module 1 does not enter $L$, 
and $\textit{prior}$ of module 2 can only enter $L$ in the $\textit{\textit{READY}}$ state}. Our proof system is able to deal with these complexities, and completes the entire proof in around 11,000 lines of code.

This case study can hardly be proved using the DC-based HHL~\cite{LLQZ10}, as it is not compositional with respect to parallel composition and thus needs to define specific inference rule for each case of parallel composition. Other DC-based logics~\cite{WZG12,GWZZ13} are compositional but too complicated to have any implementation support. 

\section{Conclusion} \label{sec:Conclusion}
In this paper, we present a hybrid Hoare logic for reasoning about HCSP processes, which generalizes and simplifies the existing DC-based hybrid Hoare logics,  and prove its soundness, and continuous and discrete relative completeness. 
Finally, we provide an implementation of this logic in Isabelle/HOL and verify
two case studies to illustrate the power and scalability of our logic. 

For future work, we will consider to specify and verify more properties including livelock and total correctness in HHL. 

\bibliographystyle{ACM-Reference-Format}
\bibliography{arxiv}

\end{document}